\documentclass[letterpaper,11pt]{article}
\usepackage[letterpaper,margin=1in]{geometry}

\usepackage{booktabs}

\usepackage{microtype}
\usepackage{xcolor}
\usepackage{xspace}
\usepackage{amsmath}
\usepackage{amssymb}
\usepackage{amsthm}
\usepackage{enumitem}
\usepackage{textgreek}
\usepackage{graphicx}
\usepackage{framed}
\usepackage[small]{caption}
\usepackage{tikz-cd}
\usepackage[unicode]{hyperref}
\usepackage{doi}
\usepackage[capitalize, nameinlink]{cleveref}
\usepackage[linesnumbered,ruled,vlined,algo2e]{algorithm2e}
\usepackage{algorithm}
\usepackage{algpseudocode}
\Crefname{remark}{Remark}{Remarks}
\Crefname{observation}{Observation}{Observations}
\usepackage{todonotes}
\usepackage{thm-restate}

\theoremstyle{plain}
\newtheorem{theorem}{Theorem}
\newtheorem{lemma}[theorem]{Lemma}

\newtheorem{corollary}[theorem]{Corollary}

\theoremstyle{definition}
\newtheorem{definition}[theorem]{Definition}

\newtheorem{observation}[theorem]{Observation}

\theoremstyle{remark}

\DeclareMathOperator{\poly}{poly}

\newcommand{\LOCAL}{\ensuremath{\mathsf{LOCAL}}\xspace}

\definecolor{darkgreen}{rgb}{0,0.5,0}
\definecolor{darkred}{rgb}{0.4,0,0}
\hypersetup{
	colorlinks=true,
	linkcolor=darkred,
	citecolor=darkgreen,
	filecolor=black,
	urlcolor=[rgb]{0,0.1,0.5},
	pdftitle={On the Average Complexity of Locally Checkable Problems on Trees},
	pdfauthor={TBD}
}

\newenvironment{myabstract}
{\list{}{\listparindent 1.5em%
		\itemindent    \listparindent
		\leftmargin    1cm
		\rightmargin   1cm
		\parsep        0pt}%
	\item\relax}
{\endlist}

\newenvironment{mycover}
{\list{}{\listparindent 0pt
		\itemindent    \listparindent
		\leftmargin    1cm
		\rightmargin   1cm
		\parsep        0pt}%
	\raggedright
	\item\relax}
{\endlist}

\newcommand{\myemail}[1]{\,$\cdot$\, {\small #1}}
\newcommand{\myaff}[1]{\,$\cdot$\, {\small #1}\par\smallskip}

\newcommand{\eps}{\varepsilon}
\newcommand{\calA}{\mathcal{A}}

\newcommand{\calG}{\mathcal{G}}

\newcommand{\fA}{\mathcal{A}}

\newcommand{\dec}{\mathsf{Decline}}
\newcommand{\con}{\mathsf{Connect}}
\newcommand{\cop}{\mathsf{Copy}}
\newcommand{\pidd}{\Pi_{\Delta, d, k}}
\newcommand{\pitwo}{\pidd^{2\frac{1}{2}}}
\newcommand{\pithree}{\pidd^{3\frac{1}{2}}}
\newcommand{\inn}{\operatorname{in}}
\newcommand{\out}{\operatorname{out}}
\newcommand{\sinn}{\Sigma_{\inn}}

\newcommand{\act}{\mathsf{Active}}
\newcommand{\wei}{\mathsf{Weight}}
\newcommand{\fa}{\mathcal A}
\newcommand{\fu}{\hat{U}_{\cop}}
\newcommand{\hu}{\hat{U}}
\newcommand{\fp}{\mathcal A_{\poly}}

\DeclareMathOperator{\E}{\mathbb{E}}

\newcommand{\nodeavg}{\mathsf{AVG}_V}

\graphicspath{{figs/}}

\renewcommand{\subparagraph}[1]{\medskip\noindent\textbf{#1}}

\crefname{algocf}{Alg.}{Algs.}
\Crefname{algocf}{Algorithm}{Algorithms}

\newcommand\blfootnote[1]{
    \begingroup
    \renewcommand\thefootnote{}\footnote{#1}
    \addtocounter{footnote}{-1}
    \endgroup
}

\begin{document}

\title{Completing the Node-Averaged Complexity Landscape of LCLs on Trees}
\author{}

\renewcommand*{\thefootnote}{\fnsymbol{footnote}}

\begin{mycover}
    {\huge\bfseries\centering Completing the Node-Averaged Complexity Landscape of LCLs on Trees \par}
  \bigskip
  \bigskip
  \bigskip

  \textbf{Alkida Balliu}
  \myemail{alkida.balliu@gssi.it}
  \myaff{Gran Sasso Science Institute, L'Aquila, Italy}
  \textbf{Sebastian Brandt}
  \myemail{brandt@cispa.de}
  \myaff{CISPA Helmholtz Center for Information Security, Saarbr\"ucken, Germany}
  \textbf{Fabian Kuhn}
  \myemail{kuhn@cs.uni-freiburg.de}
  \myaff{University of Freiburg, Freiburg, Germany}
  \textbf{Dennis Olivetti}
  \myemail{dennis.olivetti@gssi.it}
  \myaff{Gran Sasso Science Institute, L'Aquila, Italy}
  \textbf{Gustav Schmid}
  \myemail{schmidg@informatik.uni-freiburg.de}
  \myaff{University of Freiburg, Freiburg, Germany}

\blfootnote{This work has been partially funded by the PNRR MIUR research project GAMING ``Graph Algorithms and MinINg for Green agents'' (PE0000013, CUP D13C24000430001), and by the research project RASTA ``Realtà Aumentata e Story-Telling Automatizzato per la valorizzazione di Beni Culturali ed Itinerari'' (Italian MUR PON Project ARS01 00540).}
  
\end{mycover}
\bigskip

\renewcommand*{\thefootnote}{\arabic{footnote}}
\addtocounter{footnote}{-1}

\begin{myabstract}
The node-averaged complexity of a problem captures the number of rounds nodes of a graph have to spend \emph{on average} to solve the problem in the \LOCAL model. 
A challenging line of research with regards to this new complexity measure is to understand the complexity landscape of locally checkable labelings (LCLs) on families of bounded-degree graphs. Particularly interesting in this context is the family of bounded-degree trees as there, for the worst-case complexity, we know a complete characterization of the possible complexities and structures of LCL problems. A first step for the node-averaged complexity case has been achieved recently~[DISC '23], where the authors in particular showed that in bounded-degree trees, there is a large complexity gap: There are no LCL problems with a deterministic node-averaged complexity between $\omega(\log^* n)$ and $n^{o(1)}$. For randomized algorithms, they even showed that the node-averaged complexity is either $O(1)$ or $n^{\Omega(1)}$.
In this work we fill in the remaining gaps and give a complete description of the node-averaged complexity landscape of LCLs on bounded-degree trees.
Our contributions are threefold.
\begin{itemize}
    \item On bounded-degree trees, there is no LCL with a node-averaged complexity between $\omega(1)$ and $(\log^*n)^{o(1)}$.
    \item For any constants $0<r_1 < r_2 \leq 1$ and $\eps>0$, there exists a constant $c$ and an LCL problem with node-averaged complexity between $\Omega((\log^* n)^c)$ and $O((\log^* n)^{c+\varepsilon})$.
    \item For any constants $0<\alpha\leq 1/2$ and $\eps>0$, there exists an LCL problem with node-averaged complexity $\Theta(n^x)$ for some $x\in [\alpha, \alpha+\eps]$.
\end{itemize}
\end{myabstract}

\clearpage
\tableofcontents
\clearpage 

\section{Introduction}
\label{sec:intro}

Distributed computation theory has witnessed remarkable progress since the 1980s.
Researchers have not only been able to determine tight complexities of many fundamental distributed graph problems, but they have also been able to develop frameworks that can determine the complexity of entire classes of problems at once. 
Many known results establish lower bounds on the amount of rounds of communication that are required to solve certain problems. These results are proved for ad-hoc, worst-case networks, suffering from notable limitations: in particular, such lower bounds do not provide any information about the time it takes for a randomly selected node to terminate, and they only state that at least one node of the network has to spend a lot of time.
For this reason, the last yearst have seen several attempts to go beyond worst-case complexity.

One recent successful line of research studies the node-averaged complexity of distributed graph problems, which measures the average runtime of a node in the worst-case graph. The study of node-averaged complexity of graph problems is not only interesting per se, but it can also be a powerful tool for developing algorithms with better worst-case complexity. A notable example is the recent algorithm for computing a $(\Delta+1)$-coloring in $O(\log n \log^2 \Delta)$ deterministic worst-case rounds \cite{GhaffariK21}, which is built on top of a subroutine for a variant of coloring (called list-coloring) that has $O(\log^2 \Delta)$ deterministic node-averaged complexity. Any improvement on the node-averaged complexity of this problem would lead to an algorithm for $(\Delta+1)$-coloring with better deterministic worst-case complexity.

\paragraph{Locally Checkable Labelings.} A rich and successful line of research has studied a class of problems called Locally Checkable Labelings (LCLs), that have been introduced in the seminal work of Naor and Stockmeyer \cite{NaorStockmeyer95}. Informally, these problems satisfy that a given solution is correct if and only if the constant-radius neighborhood of each node satisfies some given constraints. LCLs include many classical problems, such as coloring, maximal matching, and maximal independent set. The worst-case complexity of LCLs has been extensively studied in the context of bounded-degree graphs.  In such a setting, we now have an almost complete characterization of what the possible deterministic and randomized worst-case complexities of LCLs are, and sometimes we even have decidability results: in many cases, given an LCL defined in some formal language, it is possible to automatically compute its distributed time complexity, and synthesize an algorithm for it with optimal runtime. These works studied different graph topologies, such as paths and cycles \cite{lcls_on_paths_and_cycles,balliu19lcl-decidability}, trees \cite{B0COSS22_LCLregularTrees,BBOS18almostGlobal,BHOS19HomogeneousLCL,CP19timeHierarchy,LCLs_in_rooted_trees,binary_lcls,brandt21trees,bcmos21,B0FLMOU22,chang20}, grids \cite{lcls_on_grids}, and general graphs \cite{BBOS18almostGlobal,BBOS20paddedLCL,bcmos21,BHKLOS18lclComplexity,CKP19exponential}. 

\paragraph{Node-averaged complexity of LCLs.}
A notable line of research regards the worst-case complexity of LCLs on trees of bounded degree. There, we know that LCLs can only have the following deterministic worst-case complexities: $O(1)$, $\Theta(\log^* n)$, $\Theta(\log n)$, and $\Theta(n^{1/k})$ for any fixed integer $k > 0$. Moreover, randomization can only help for problems with deterministic complexity $\Theta(\log n)$, making their randomized complexity $\Theta(\log \log n)$. 
Building on such remarkable results, a first attempt to generalize our knowledge to the case of node-averaged complexity has been done by Feuilloley in
\cite{Feuilloley17}, who showed that on cycles, the deterministic node-averaged complexity of an LCL is asymptotically the same as its worst-case complexity. Then, \cite{Balliu0KOS23} considered LCLs on trees and proved the following results.
 Any LCL on trees either has node-averaged complexity $O(\log^* n)$, or it has polynomial node-averaged complexity. Thus, differently from the landscape of worst-case complexities, there are no LCLs with node-averaged complexity $\Theta(\log n)$ or $\Theta(\log \log n)$. A concrete and surprising application of this generic result is that a $3$-coloring can be computed in just $O(\log^* n)$ node-averaged rounds in bounded-degree trees. For this problem, it is known that $\Omega(\log n)$ worst-case rounds are required.
 Another result shown in \cite{Balliu0KOS23} states that, if a problem has worst-case complexity $\Theta(n^{1/k})$ for some $k$, then it has $\Omega(n^{1/(2^k - 1)})$ deterministic node-averaged complexity and $\Omega(n^{1/(2^k - 1)} / \log n)$ randomized node-averaged complexity. Hence, a problem that has polynomial worst-case complexity also has polynomial node-averaged complexity.
 Finally, the authors show that any LCL on trees that can be solved in subpolynomial worst-case time can be solved in $O(1)$ randomized node-averaged complexity.

\paragraph{Node-averaged complexity of graph problems.}
Apart from LCLs, the node-averaged complexity has been studied for several other specific problems.
For example, Barenboim and Tzur~\cite{BarenboimT19} considered the problem of vertex-coloring in graphs of small arboricity, showing that the node-averaged complexity can be significantly smaller than what we currently know for worst-case. 
While the node-averaged complexity of some problems is strictly better than their worst-case complexity, in \cite{BalliuGKO22_average} it has been shown that the currently best known randomized worst-case lower bound for the MIS problem, which is $\Omega(\sqrt{\log n/\log\log n})$,  holds in the case of node-averaged complexity as well. On the other hand, the paper also showed that, for a slight relaxation of MIS, called $(2,2)$-ruling set, it is possible to provide an algorithm with a node-averaged complexity that is significantly better than the known worst-case lower bound for this problem.

\subsection{Our Contribution}
In our work, we complete the characterization of the possible node-averaged complexities of LCLs on trees of bounded degree. We show that, perhaps surprisingly, the landscape of possible node-averaged complexities is significantly different from the one of possible worst-case complexities. \cref{fig:LandscapeBefore} gives an overview of everything that was known so far about the node-averaged complexity landscape of LCLs on bounded degree trees. We will later see that with our results we get a complete characterization of the landscape.

\begin{figure}
	\centering
	\includegraphics[width=.6\textwidth]{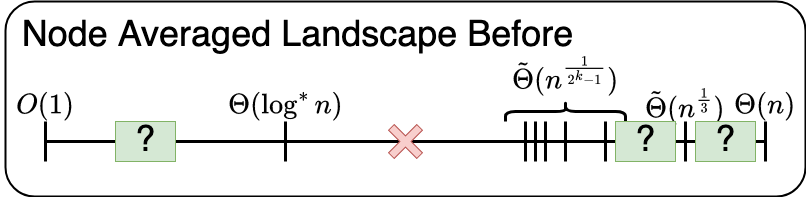}
	\caption{Everything that was known about the Node-Averaged Complexity Landscape before the results in this paper.}
	\label{fig:LandscapeBefore}
\end{figure}

\paragraph{Polynomial regime.}
As already mentioned, in the case of worst-case complexities, the only possible complexities in the polynomial regime are $\Theta(n^{1/k})$ for any constant integer $k$. Prototypical problems exhibiting these worst-case complexities are so-called $k$-hierarchical $2\frac{1}{2}$-coloring problems. In our work, we show that the landscape of possible polynomial node-averaged complexities is substantially different, that is, that the region $n^{\Omega(1)}$--$O(\sqrt{n})$ is infinitely dense. More in detail, we show that, for any rational $0 < a/b \leq 1/2$, there exists an LCL with (deterministic and randomized) node-averaged complexity $\Theta(n^{a/b})$, implying the following.
\begin{restatable}{theorem}{UpperDensity}\label{thm:UpperDensity}
For any two real numbers $0<r_1<r_2\leq\frac{1}{2}$ there exists a constant $r_1<c<r_2$ and an LCL $\Pi$ such that $\Pi$ has node-averaged complexity $\Theta(n^c)$
\end{restatable}

We achieve this by creating a weighted version of $k$-hierarchical $2\frac{1}{2}$ coloring that depends on two more parameters $\Delta,d$. We call it $\Pi^{2.5}_{\Delta, d, k}$ and we prove matching upper and lower bounds for it.

\begin{restatable}{theorem}{UpperUpperBound}\label{thm:UpperUpperBound}
For any $D,d,k$ such that $D\geq d+3$, the node-averaged complexity of $\Pi^{2.5}_{\Delta, d, k}$ is $O(n^{\alpha_1})$, where $\alpha_1 = \frac{1}{\sum_{j = 0}^{k - 1}(2 - x)^j}$ and $x = \frac{\log (\Delta - d - 1)}{\log (\Delta - 1)}$.
\end{restatable}

\begin{restatable}{theorem}{UpperLowerBound}
\label{thm:UpperLowerBound}
For any constants $D,d,k$, such that $D\geq d+3$ the LCL $\Pi^{2.5}_{\Delta, d, k}$ has node-averaged complexity $\Omega(n^{\alpha_1})$, where $\alpha_1 = \frac{1}{\sum_{j=0}^{k-1} (2-x)^j}$ and $x = \frac{\log(\Delta - d - 1)}{\log(\Delta -1)}$.
\end{restatable}

\paragraph{New complexities in the \texorpdfstring{$O(\log^* n)$}{log*(n)} regime.}
In the case of worst-case complexities, it is known that there are no LCLs that have a complexity that lies in the region $\omega(1)$--$o(\log^* n)$. Moreover, it is known that, for randomized algorithms, there are no LCLs that have a node-averaged complexity that lies in the region $\omega(1)$--$n^{o(1)}$. We show that, in the case of deterministic node-averaged complexities, this is false.
We first introduce a new class of problems we call $k$-hierarchical $3\frac{1}{2}$-coloring, which already gives an infinite amount of nonempty complexity classes between $(\log^* n)^{\Omega(1)}$--$O(\log^* n)$.\\
We repeat a similar process as in the polynomial regime to obtain a new class of LCLs $\Pi^{3.5}_{\Delta,d,k}$. We again obtain a strong lower bound.

\begin{restatable}{theorem}{LowerLowerBound}
\label{thm:LowerBoundWeighted3.5}
For any constants $D,d,k$, such that $D\geq d+3$ the LCL $\Pi^{3.5}_{\Delta, d, k}$ has node-averaged complexity $\Omega((\log^* n)^{\alpha_1(x)})$, where $\alpha_1(x) = \frac{1}{1 + (1-x)\sum_{j=0}^{k-2} (2-x)^j}$ and $x = \frac{\log(\Delta - d - 1)}{\log(\Delta -1)}$.
\end{restatable}

However due to the fact that an algorithm for this problem can not even afford to have a linear number of nodes run for $\log^* n$ many rounds, proving a matching upper bound proves quite challenging. We instead get an algorithm that almost matches the lower bound.

\begin{restatable}{theorem}{LowerUpperBound}\label{thm:LowerUpperBound}
For any $D,d,k$ such that $d\geq 3, D\geq d+3$, the node-averaged complexity of $\Pi^{3.5}_{\Delta, d, k}$ is $O((\log^* n)^{\alpha_1(x')})$, where $\alpha_1(x') = \frac{1}{1 + (1-x')\sum_{j=0}^{k-2} (2-x')^j}$ and $x' = \frac{\log(\Delta -d +1)}{\log(\Delta -1)}$.
\end{restatable}

We overcome these complications, by showing that through clever choice of parameters $\Delta$ and $d$, we can get our upper and lower bound to become arbitrarily close, giving us the same kind of guarantee as in the polynomial regime.

\begin{restatable}{theorem}{LowDensity}\label{thm:LowDensity}
For any two real numbers $0 < r_1< r_2 < 1$ and any $\varepsilon > 0$ there exist constants $\Delta, d, k, c$ such that $r_1 \leq c \leq r_2$ and LCL $\Pi^{3.5}_{\Delta,d,k}$ has node-averaged complexity between $\Omega((\log^*n)^c)$ and $O((\log^*n)^{c + \varepsilon})$.
\end{restatable}

\paragraph{New gaps in the \texorpdfstring{$O(\log^* n)$}{log*(n)} regime.}
We complete the results about the $O(\log^* n)$ regime by proving a gap in the possible complexities.
\begin{restatable}{theorem}{lowgap}\label{thm:lowgap}
There are no LCLs with a deterministic complexity that lies in the range $\omega(1)$--$(\log^* n)^{o(1)}$. Moreover, given an LCL, it is decidable whether it can be solved in $O(1)$ deterministic node-averaged rounds.
\end{restatable}

\paragraph{The new node-averaged complexity landscape of LCLs on trees}
By including all of our new results, we get a complete picture about the node-averaged complexity landscape of LCLs on bounded-degree trees. We provide a complete overview in \cref{fig:LandscapeAfter}.

\begin{figure}
	\centering
	\includegraphics[width=.6\textwidth]{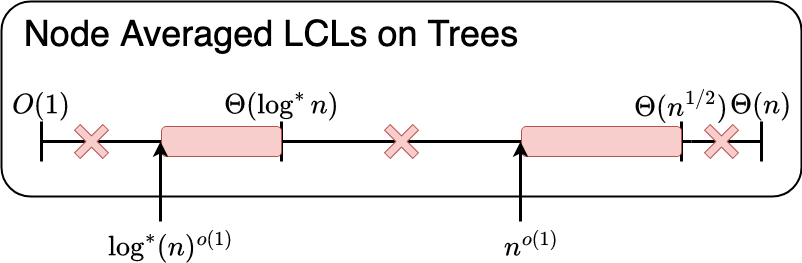}
	\caption{Everything that is known about the node-averaged complexity landscape after including the results in this paper. The gap between constant and $\log^*(n)^{o(1)}$ is due to \cref{thm:lowgap}. The regions of infinite density, represented by red bars, are a result of \cref{thm:LowDensity,thm:UpperDensity}. The gap between $\omega(\sqrt{n})$ and $o(n)$ is due to \cref{cor:linearGap}.}
	\label{fig:LandscapeAfter}
\end{figure}

\subsection{High-level Ideas and Techniques}
In the following, we outline the main ideas that we use to prove our results.
\paragraph{LCLs with weighted node-averaged complexity}
Intuitively the idea is the following: Start with a problem like $2$-hierarchical $2\frac{1}{2}$-coloring, which has worst case complexity $\Theta(n^{1/2})$ \cite{CP19timeHierarchy} and node-averaged complexity $\Theta(n^{1/(2^k-1)})$ \cite{Balliu0KOS23}. Because the problem has worst case complexity $\Theta(n^{1/2})$, there must exist an instance of the problem with a node $v$ which runs for $T_v \in \Omega(n^{1/k})$ rounds. Now assume we could give some nodes more weight during the node-averaged analysis. Then we could give a linear amount of weight $w(v) = \frac{n}{2}$ to this node $v$ and as a result the node-averaged complexity would be at least 
\[
\Bar{T} \geq \frac{1}{n} \cdot T_v \cdot w(v) = \frac{1}{n} \cdot \Omega(n^{1/2}) \cdot \frac{n}{2} \in \Omega(n^{1/2}).
\]
So the node-averaged complexity would match the worst-case complexity. We achieve such a weighted behavior simply by designing an LCL in a clever way.

Assume we have marked some neighbor $u$ of $v$ as a weight node, by giving it an input label $\wei$. Then such a weight node $u$ must copy the output of its adjacent non-weight node $v$. Essentially $u$ has to wait for $v$ to terminate, so it can then output the same thing. However to put a lot of weight on $v$, we also have to attach a lot of nodes to $v$, but since the maximum degree is constant this is not so simple. 
We solve this by having the output label of $v$ propagate through weight nodes. For example we could require weight nodes that are adjacent to other weight nodes to also copy the same output label. We leverage this by attaching a path of $\frac{n}{2}$ nodes to $v$. Then all of the nodes in that path have to copy the output of $v$ and therefore all of the $\frac{n}{2}$ nodes have to wait for $v$ to terminate. However, we have made a grave error here, because the propagation along this long path takes $\Omega(n)$ time. As a result we have increased the worst case complexity of this new LCL to $\Theta(n)$.

Luckily we have a complete understanding of what determines the worst case complexity of LCLs on trees. The problems with worst case complexity $\Theta(n)$ are exactly the problems for which, in long paths, it takes a linear amount of rounds to terminate, like, e.g., $2$-coloring. So we include a way to let nodes terminate earlier, by allowing some weight nodes to output $\dec$ and immediately terminate without copying any other output label.

So the resulting LCL could roughly be described like this: Weight nodes that are adjacent to non weight nodes have to copy the output of one such non weight node. For any weight node that did copy the output, at most $d$ neighboring weight nodes may decline to also copy the output, but the remaining nodes do have to copy the output.

As a consequence of allowing some nodes to $\dec$, not all weight nodes have to actually copy the output and we run into efficiency problems. We will see that there is now some efficiency factor $x$ that determines how many of the weight nodes actually have to wait for our node $v$ to terminate. By scaling this efficiency factor between $(0,1)$, we are able to obtain all intermediate complexities between the worst case complexity of $\Theta(n^{1/2})$ and the node-averaged complexity $\Theta(n^{1/2^k-1})$.

\paragraph{New complexities in the \texorpdfstring{$O(\log^* n)$}{log*(n)} regime.}
We want to get the same kind of density result that we get by introducing weight nodes to the family of $k$-hierarchical $2\frac{1}{2}$-coloring problems, but with complexities in $O(\log^* n)$. Similarly to the polynomial regime, where we can scale between worst case complexity $\Theta(n^{1/k})$ and node-averaged complexity $\Theta(n^{1/(2^k - 1)})$, we introduce a new infinite family of problems we call $k$-hierarchical $3\frac{1}{2}$-coloring. 
These problems have worst case complexity $\Theta(\log^*n)$ and node-averaged complexity $\Theta((\log^*n)^{1/2^{k-1}})$. This allows us to again apply the same tricks with weighted nodes, to scale between the worst case and the node-averaged complexities. However, since an algorithm is only allowed $o(\log^*n)$ rounds \emph{on average}, it becomes highly non-trivial to actually achieve a fast node-averaged complexity. 
Using the algorithm of \cite{Balliu0KOS23} we are able to still achieve a node-averaged complexity that is close to the lower bound. By applying some tricks in the choice of parameters, we can get this upper bound arbitrarily close to the lowerbound (but never actually have them match). Still this will be enough to also prove the same kind of infinite density result also for the $O(\log^* n)$ regime.

\paragraph{New gaps in the \texorpdfstring{$O(\log^* n)$}{log*(n)} regime.}
In order to show that there are no LCLs in the range $\omega(1)$--$(\log^* n)^{o(1)}$, we operate as follows.
As already mentioned, in \cite{Balliu0KOS23} it is shown that there are no LCLs with node-averaged complexity that lies in $\omega(\log^* n)$--$n^{o(1)}$. This result is shown by providing a unified way to solve all problems that have (worst-case or node-averaged) complexity $n^{o(1)}$, and the provided normal-form algorithm has node-averaged complexity $O(\log^* n)$. On a high level, this normal-form algorithm computes a decomposition of the tree that satisfies some desirable properties, and in parallel exploits the decomposition that is being computed in order to let many nodes terminate early. The resulting $O(\log^* n)$ complexity is only needed for splitting some long paths, obtained during the decomposition, into shorter paths. Moreover, as already noted in \cite{Balliu0KOS23}, if such splitting can be avoided, then the node-averaged runtime could be directly improved to $O(1)$. In our work, we provide two main ingredients:
\begin{itemize}
    \item Given an LCL, it is possible to automatically determine whether it is needed to split long paths into shorter paths when computing the decomposition.
    \item If an LCL can be solved in $(\log^* n)^{o(1)}$ deterministic node-averaged rounds, then it is not needed to split paths.
\end{itemize}
The importance of the above results is twofold. On the one hand, we get that, if a problem can be solved in $(\log^* n)^{o(1)}$ deterministic node-averaged rounds, then it has $O(1)$ deterministic node-averaged complexity. On the other hand, we also obtain that we can automatically decide whether a given problem has $O(1)$ node-averaged complexity. We highlight that, in the worst-case setting, determining whether a problem can be solved in $O(1)$ rounds is a long-standing open question.

\section{Preliminaries and Definitions}
\paragraph{Locally Checkable Labeling Problems.}
As already mentioned, on a high level, the class of Locally Checkable Labeling (LCL) problems contains all those problems defined on bounded-degree graphs such that: (i) nodes may have inputs that come from a finite set, (ii) nodes must output labels from a finite set, (iii) the produced solution must be correct in any $r$-radius neighborhood, where $r$ is some constant. More precisely, an LCL problem $\Pi$ is defined as a tuple $(\Sigma_{\mathrm{in}},\Sigma_{\mathrm{out}},C, r)$, where: (i) $\Sigma_{\mathrm{in}}$ is a finite set of possible input labels for $\Pi$; (ii) $\Sigma_{\mathrm{out}}$ is a finite set of possible output labels for $\Pi$; (iii) $r\ge 1$ is an integer called \emph{the checkability radius} of $\Pi$; (iv) $C$ (that stands for ``constraint’’) is a finite set containing the labeled graphs that represent all radius-$r$ neighborhoods that are valid for $\Pi$. More in detail $C$ is a finite set of pairs $(H, v)$, where $H=(V_H, E_H)$ is a graph and $v\in V_H$, satisfying the following: (1) $v$ has eccentricity at most $r$ in $H$; (2) To each pair $(v,e)\in V_H\times E_H$ is assigned a label $\ell_{\mathrm{in}}\in \Sigma_{\mathrm{in}}$ and a label $\ell_{\mathrm{out}}\in \Sigma_{\mathrm{out}}$.
Then, solving an LCL $\Pi=(\Sigma_{\mathrm{in}},\Sigma_{\mathrm{out}},C, r)$ on a graph $G=(V,E)$ (where all node-edge pairs are labeled with a label in $\Sigma_{\mathrm{in}}$) requires to label each pair $(v,e)\in V\times E$ with a label in $\Sigma_{\mathrm{out}}$ satisfying that, for each node $v\in V$, the labeled graph induced by the nodes in the radius-$r$ neighborhood of $v$ is isomorphic to some (labeled) graph in $C$.

\paragraph{Node-averaged complexity.}
In this paper, we use the notion of node-averaged complexity as in \cite{BalliuGKO22_average, Balliu0KOS23}. Let $\mathcal{A}$ be an algorithm solving an LCL $\Pi$. Let $\mathcal{G}$ be a family of graphs, and let $G=(V,E)\in\mathcal{G}$ be a graph where we run $\mathcal{A}$. Let $v\in V$ and let $T_v^G(\mathcal{A})$ be the number of rounds after which $v$ terminates when running $\mathcal{A}$ on $G$. Then, the  node-averaged complexity of $\mathcal{A}$ on the family of graphs $\mathcal{G}$ is defined as follows.
\[
    \nodeavg(\calA)  :=  \max_{G\in \calG} \,\frac{1}{|V|}\cdot\E\left[\sum_{v\in V(G)}T_v^G(\calA)\right]\ =\
    \max_{G\in \calG} \frac{1}{|V|}\cdot\sum_{v\in V(G)}\E\big[T_v^G(\calA)\big]\\
\]

\section{Road Map}
We now provide a summary of the structure of the paper.
\paragraph{The $k$-hierarchical $3\frac{1}{2}$-coloring problems.}
We start, in \Cref{sec:khierarchical312}, by formally defining the $k$-hierarchical $3\frac{1}{2}$-coloring problems. These problems are defined by slightly modifying the definition of the $k$-hierarchical $2\frac{1}{2}$-coloring problems of \cite{CP19timeHierarchy}. In \Cref{sec:khierarchical312}, we show that the $k$-hierarchical $3\frac{1}{2}$-coloring problems have complexity $\Theta((\log^*n)^{1/2^{k-1}})$.

\paragraph{Weighted problems.}
While the results of \Cref{sec:khierarchical312} provide LCLs with some node-averaged complexities in the range $\omega(1)$--$o(\log^* n)$, in our work we also show that the complexity landscape is in fact infinitely dense in that region, and in the polynomial region. For this purpose, we start, in \Cref{sec:weighted-problems}, by introducing weighted versions of the problems of $2\frac{1}{2}$- and $3\frac{1}{2}$-coloring.

\paragraph{Weighted lower bounds.}
We provide lower bounds for the weighted variants of $2\frac{1}{2}$- and $3\frac{1}{2}$-coloring in \Cref{sec:weightedlower}.

\paragraph{Weighted upper bounds.}
In \Cref{sec:AlgWeighted212Coloring,sec:UsingOldAlgo} we provide upper bounds for the weighted variants of $2\frac{1}{2}$-coloring, and $3\frac{1}{2}$-coloring, respectively.

\paragraph{Density results.}
In \Cref{sec:density}, we combine the results of \Cref{sec:weightedlower,sec:AlgWeighted212Coloring,sec:UsingOldAlgo} to show that the node-averaged complexity landscape of LCLs is dense in the regions $(\log^*n)^{\Omega(1)}$--$o(\log^* n)$ and $n^{\Omega(1)}$--$o(\sqrt{n})$. Since to our knowledge this has not been stated anywhere else so far, we also give a proof for the $\omega(\sqrt{n})$--$o(n)$ gap as a direct consequence of a nice lemma by \cite{Feuilloley17}.

\paragraph{More efficient weight.}
What is left out from the results of the previous sections is showing that there are LCL problems in the polynomial regime that have the same worst-case and node-averaged complexity. We do this in \Cref{sec:efficientWeight}, by defining the weighted version of $2\frac{1}{2}$-coloring differently.

\paragraph{The $\omega(1)$--$(\log^* n)^{o(1)}$ gap}
We conclude, in \Cref{sec:gap}, by proving that there are no LCLs with a node-averaged complexity that lies in the range $\omega(1)$--$(\log^* n)^{o(1)}$.

\section{\texorpdfstring{$k$}{k}-hierarchical \texorpdfstring{$3\frac{1}{2}$}{3.5}-coloring}\label{sec:khierarchical312}
We first do a warm up by answering one of the open questions in \cite{Balliu0KOS23}, namely the question if there exist LCLs with node-averaged complexity in the range $\omega(1)$--$o(\log^*n)$. To give a positive answer to this question, we define a new infinite family of LCLs called $k$-hierarchical $3\frac{1}{2}$-coloring. These problems are really just a small twist on the family of the $k$-hierarchical $2\frac{1}{2}$-coloring problems introduced by Chang and Pettie \cite{CP19timeHierarchy}, that we now report.

\begin{definition}[$k$-hierarchical $2\frac{1}{2}$-coloring]\label{def:2.5Col}
There are no input labels. Instead, each node has a level in $\{1,\ldots,k+1\}$, that can be computed in constant time, and the constraints of the nodes depend on the level that they have. The level of a node is computed as follows.
\begin{enumerate}
    \item Let $i \gets 1$.
    \item Let $V_i$ be the set of nodes of degree at most $2$ in the remaining tree. Nodes in $V_i$ are of level $i$. Nodes in $V_i$ are removed from the tree.
    \item Let $i \gets i+1$. If $i \le k$, continue from step 2.
    \item Remaining nodes are of level $k+1$.
\end{enumerate}
Each node must output a single label in $\Sigma_{\mathrm{out}} := \{W,B,E,D\}$, where $W$ stands for White, $B$ for Black, $E$ for Exempt, and $D$ for Decline, and based on their level, they must satisfy the following local constraints.
\begin{itemize}
    \item No node of level $1$ can be labeled $E$.
    \item All nodes of level $k+1$ must be labeled $E$.
    \item Any node of level $2\le i \le k$ is labeled $E$ iff it is adjacent to a lower level node labeled $W$, $B$, or $E$.
    \item Any node of level $1 \le i \leq k$ that is labeled $W$ (resp.\ $B$) has no neighbors of level $i$ labeled $W$ (resp.\ $B$) or $D$. In other words, $W$ and $B$ are colors, and nodes of the same color cannot be neighbors in the same level.
     \item Nodes of level $k$ cannot be labeled $D$. As a result they must be properly 2-colored with colors $W,B$. They may output $E$ only if their lower level neighbours did not output $D$.
\end{itemize}
\end{definition}
The problem of $k$-hierarchical $2\frac{1}{2}$-coloring has worst case complexity $\Theta(n^{1/k})$ \cite{CP19timeHierarchy}. However, in \cite{Balliu0KOS23}, it is shown that the node-averaged complexity of this problem is $\Theta(n^{1/(2^k-1)})$. 

In order to better understand this family of problems, let us consider the case $k=2$, which gives an LCL with worst-case complexity $\Theta(\sqrt{n})$. In a worst-case instance for this problem, there are only nodes of level $1$ and $2$: nodes of level $2$ form a path $P$ of length $\Theta(\sqrt{n})$, and to each node $v$ of $P$ there is a path $Q_v$ of length $\Theta(\sqrt{n})$ of nodes of level $1$ connected to it. Let us call these paths ``$Q$-paths''. The constraints of the problem require that each $Q$-path is either all $2$-colored with labels $W$ and $B$, or all labeled $D$. Then, the subpaths of $P$ induced by nodes that are not connected to a $2$-colored $Q$-path must be properly $2$-colored, while the other nodes of $P$ can output $E$. This implies that either some $Q$-path is $2$-colored, or $P$ is $2$-colored, and it is possible to prove that this implies a worst-case complexity of $\Theta(\sqrt{n})$. In the case of node-averaged complexity, a worst-case instance looks different: $Q$-paths have length $\Theta(n^{1/3})$ and $P$ has length $\Theta(n^{2/3})$, and it is possible to prove that the node-averaged complexity is $\Theta(n^{1/3})$. Intuitively, this happens because the nodes of $P$ are less, so they can spend more time while keeping the average runtime low.

In the following, we will slightly modify the definition of these problems to obtain the family of $k$-hierarchical $3\frac{1}{2}$-coloring problems.
While our new problems will not be interesting from a worst case perspective, we will prove that these problems result in an infinite amount of intermediate complexities in the range $\omega(1)$--$o(\log^*n)$. The rules are almost the same as in $k$-hierarchical $2\frac{1}{2}$-coloring, except that level $k$ nodes must either output $E$, or output a valid $3$-coloring using the colors $\{R,G,Y\}$. For completeness, we restate all of the rules.  

\begin{definition}[$k$-hierarchical $3\frac{1}{2}$-coloring]\label{def:3.5Col}
Nodes are assigned a level in $\{1,\ldots,k+1\}$, in the same way as for $k$-hierarchical $2\frac{1}{2}$-coloring. However the set of possible labels is now $\Sigma_{\mathrm{out}} = \{R,G,Y,W,B,E,D\}$, which now also contains the colors red, green and yellow. 
Nodes must satisfy the following rules based on their level:
\begin{itemize}
    \item No node of level $1$ can be labeled $E$.
    \item All nodes of level $k+1$ must be labeled $E$.
    \item Any node of level $2\le i \le k$ is labeled $E$ iff it is adjacent to a lower level node labeled $W$, $B$, or $E$.
    \item Any node of level $1 \le i < k$ that is labeled $W$ (resp.\ $B$) has no neighbors of level $i$ labeled $W$ (resp.\ $B$) or $D$. In other words, $W$ and $B$ are colors, and nodes of the same color cannot be neighbors in the same level. Also, nodes in these levels must not output any label in $\{R,G,Y\}$.
    \item Nodes of level $k$ cannot be labeled $D,W,B$. Also, adjacent nodes must not both have the same label among $R,G,Y$, that is, nodes of level $k$ not labeled $E$ must be properly $3$-colored with colors in $\{R,G,Y\}$. Nodes of level $k$ may output $E$ only if their lower level neighbours did not output $D$.
\end{itemize}
\end{definition}
This problem can be expressed as a standard LCL by setting the checkability radius $r$ to be~$O(k)$ since, in $O(k)$ rounds, a node can determine its level, and hence which constraints apply. Furthermore, by simply having all nodes of level $1,\ldots, k-1$ output $D$, and having the nodes in level $k$ compute a valid $3$-coloring, we can solve the problem in worst case time $O(\log^*n)$, due to the fact that $3$-coloring a path can be done in $O(\log^*n)$ worst-case rounds by, e.g., using the algorithm by Linial~\cite{Linial92}.
In the rest of the section, we will prove that the node-averaged complexity of $k$-hierarchical $3\frac{1}{2}$-coloring lies strictly in the range $\omega(1)$--$o(\log^*n)$, which in particular implies $\Omega(\log^* n)$ worst-case complexity, due to the fact that, on trees, there are no LCLs in the range $\omega(1)$--$o(\log^*n)$ \cite{brandt21trees}. Hence, we obtain the following corollary.
\begin{corollary}
$k$-hierarchical $3\frac{1}{2}$-coloring has worst case complexity $\Theta(\log^*n)$.
\end{corollary}

The rest of the section is dedicated to proving tight upper and lower bounds for this infinite family of problems, and thereby establishing an infinite amount of complexity classes in the node-averaged landscape. We will prove the following theorem.
\begin{theorem}\label{thm:3.5Col}
    The deterministic node-averaged complexity of computing a $k$-hierarchical $3\frac{1}{2}$-coloring is $\Theta((\log^*n)^{1/2^{k-1}})$.
\end{theorem}
This result, however, tells us nothing about the spaces between these complexities. In later sections, we will build upon these problems to show that, not only are there an infinite amount of intermediate complexities in the range $\omega(1)$--$o(\log^*n)$, but also that the complexity landscape is in fact infinitely dense in that region.

\subsection{A Generic Algorithm for \texorpdfstring{$2\frac{1}{2}$- or $3\frac{1}{2}$}{2.5- or 3.5}-coloring}\label{sec:upperbound}
We first give a generic algorithm that consists of $k$ phases and depends on parameters $\gamma_1, \ldots, \gamma_{k-1}$. Phase $i<k-1$ proceeds as follows.
\begin{itemize}
    \item \textbf{Fixing paths of level $i$:} Consider the subgraph of nodes that did not yet output any label. Then, nodes of level $i$ check if they are in a path $P$, induced by level $i$ nodes, of length at least $\gamma_i$. Every node can have this information after at most $2\gamma_i$ rounds. If $P$ has length at least $\gamma_i$, then all nodes in $P$ output $D$. Otherwise $P$ must have length strictly less than $\gamma_i$. As a result, all nodes in $P$ have seen the entire path and can output a consistent $2$-coloring using the labels $W,B$.
    \item \textbf{Higher level nodes choose $E$:} Since some paths got consistently $2$-colored, some higher-level nodes are allowed to output $E$. Consider some path $P$ of nodes of level $i$ that just decided their output labels. According to the rules of $2\frac{1}{2}$- and $3\frac{1}{2}$-coloring the endpoints of $P$ might be adjacent to higher level nodes $u,v$ (there can only be two, because each endpoint had degree 2 when it was assigned a level). If $P$ did output a proper 2-coloring, then $u,v$ can output $E$. Then again $u$ and $v$ might be adjacent to higher level nodes which can then also output $E$. We iterate this until there are no more nodes that can output $E$, this takes at most $k$ rounds, because there are only $k$ levels.
\end{itemize}
In phase $k$, all remaining level $k$ nodes (the ones that did not yet output $E$) either compute a consistent $2$-coloring in linear time (in the case of $2\frac{1}{2}$-coloring), or compute a consistent $3$-coloring in $O(\log^* n)$ time (in the case of $3\frac{1}{2}$-coloring). This finishes the description of our algorithm.

We note some properties of the provided algorithm. All paths of level $i$ nodes either output a consistent coloring, or $D$. Furthermore, nodes only output $E$ if they have a lower level neighbor that did not output $D$. Finally, nodes of level $1$ do not output $E$, and nodes of level $k$ do not output $D$. Therefore, the algorithm satisfies all constraints, and we get the following corollary.

\begin{corollary}\label{cor:genericcorrect}
The generic algorithm computes a valid solution to the $k$-hierarchical $2\frac{1}{2}$-coloring, or, respectively, the $k$-hierarchical $3\frac{1}{2}$-coloring problem.
\end{corollary}

We will first prove a generic statement that will be useful in many places of the paper. 
\begin{lemma}\label{lem:remainAfterLvlI}
Given any instance of either $k$-hierarchical $2\frac{1}{2}$-coloring, or $k$-hierarchical $3\frac{1}{2}$-coloring, consider the execution of the generic algorithm. Let $i$ be between $1$ and $k-1$, and let $n'$ be an upper bound on the number of nodes of level $\geq i$ that did not yet output a label. Then, after executing phase $i$ of the generic algorithm with parameter $\gamma_i$, the number of remaining nodes (which all have level $>i$) is at most $O(\frac{n'}{\gamma_i})$.
\end{lemma}
\begin{proof}
We argue that for each node of level $i+1$ that remains, there must be at least $\frac{\gamma_i}{2}$ nodes of level $i$ that did already terminate. Consider some remaining node $v$ of level $i+1$, it must have at least one neighbor $u$ of level $i$ that did output $D$. Otherwise $v$ would have output $E$ already. Since $u$ did output $D$, the path of level $i$ nodes containing $u$ must have length at least $\gamma_i$, so charge half of all nodes in this path towards $v$ (the other half might be charged to the node at the other end of the path). So each remaining level $i+1$ node has at least $\frac{\gamma_i}{2}$ nodes that already terminated charged to it (and the sets of charged nodes are disjoint). As a result there can be at most $\frac{2n'}{\gamma_i} \in O(\frac{n'}{\gamma_i})$ such nodes of level $i+1$ remaining.
Any path of level $i+1$ nodes has at most 2 level $i+2$ nodes adjacent to the endpoints. So for any two nodes of level $i+2$ that still remain, there must exist at least one node of level $i+1$ that did not terminate yet. By repeating this reasoning we get that for every $2^{j-i-1}$ nodes of level $j>i+1$ there exists at least one node of level $i+1$ that did not terminate yet. Or conversely, for every level $i+1$ node that remains, there remain at most 
\[
\sum_{i+1<j\leq k}2^{j-i-1} \in O(1)
\]
nodes of higher levels. If such a level $i$ node did not exist, all of these higher level nodes could simply output $E$. So since there exist at most $O(\frac{n'}{\gamma_i})$ level $i+1$ nodes, there also remain at most $O(\frac{n'}{\gamma_i})$ nodes in total.
\end{proof}

Now, equipped with \cref{lem:remainAfterLvlI}, getting our upper bound on the node-averaged complexity of $k$-hierarchical $3\frac{1}{2}$-coloring is just a matter of choosing the parameters $\gamma_1,\ldots, \gamma_{k-1}$ correctly.
Define $t := (\log^*n)^{1/2^{k-1}}$  as the target complexity. For $1 \le i \le k-1$ define $\gamma_i := t^{2^{i-1}}$.
We run the generic algorithm with value $\gamma_i$ for phase $i$.

\begin{lemma}\label{lem:UB3.5Col}
    Let $t:= (\log^* n)^{1/2^{k-1}}$. The algorithm solves $k$-hierarchical $3\frac{1}{2}$-coloring in $O(t)$ node-averaged complexity.
\end{lemma}
\begin{proof}
    By repeatedly applying \cref{lem:remainAfterLvlI}, and since $k$ is a constant, we get that after phase $i<k$ only $r_i = O\left(\frac{n}{t^{2^i-1}}\right)$ nodes remain. In fact,
    \begin{align*}
        r_i &= n \cdot O\left(\frac{1}{\gamma_1}\right) \cdot O\left(\frac{1}{\gamma_2}\right) \cdot \ldots \cdot O\left(\frac{1}{\gamma_{i}}\right)
            = O\left( \frac{n}{\prod_{j\leq i}\gamma_j} \right) = O\left( \frac{n}{\prod_{j \leq i}t^{2^{j-1}}} \right)\\
            &= O\left( \frac{n}{t^{\sum_{j \leq i} 2^{j-1}}} \right) = O\left(\frac{n}{t^{2^i-1}}\right).
    \end{align*}
    Nodes spend $2\gamma_i = 2t^{2^{i-1}}$ rounds in phase $i<k$, and $O(\log^*n) =  O(((\log^*n)^{1/2^{k-1}})^{2^{k-1}}) = O(t^{2^{k-1}})$ rounds in phase $k$. Combining this with the above result of how many nodes remain for a given phase, we get the following upper bound on the node-averaged complexity.
    \[
    \Bar{T} = \frac{1}{n} \left(r_{k-1} \cdot O(t^{2^{k-1}}) + \sum_{i = 1}^{k-1} r_{i-1} \cdot O(\gamma_i)\right) = O\left(\frac{1}{n}\sum_{i = 1}^k \frac{n}{t^{2^{i-1}-1}} \cdot t^{2^{i-1}}\right) = O\left(\frac{1}{n} \sum_{i = 1}^k \frac{n}{t^{-1}}\right) = O(t)\qedhere
    \]
\end{proof}

\subsection{Lower Bounds For \texorpdfstring{$3\frac{1}{2}$}{3.5}-Coloring}\label{sec:lowerbound}
In this section, we prove the following result.
\begin{lemma}\label{lem:lb-3half-noweights}
    The deterministic node-averaged complexity of $k$-hierarchical $3\frac{1}{2}$-coloring is $\Omega((\log^*n)^{1/2^{k-1}})$.
\end{lemma}

Before proving a lower bound for $k$-hierarchical $3\frac{1}{2}$-coloring, we state few important results about node-averaged complexity.
We use this very nice characterisation of the node-averaged complexity of LCLs on paths and cycles by Feuilloley \cite{Feuilloley17}.
\begin{lemma}[\cite{Feuilloley17}]
    For any LCL problem $\Pi$ defined on paths, the following holds:
    \begin{itemize}
        \item $\Pi$ has randomised and deterministic node-averaged complexity $\Theta(n)$ if and only if it has randomised and deterministic worst case complexity $\Theta(n)$;
        \item $\Pi$ has deterministic node-averaged complexity $\Theta(\log^*n)$ if and only if it has deterministic worst case complexity $\Theta(\log^*n)$.
    \end{itemize}
\end{lemma}
This lemma immediately gives us bounds on the deterministic node-averaged complexity of $2$-coloring and $3$-coloring paths, since these problems, on paths, are known to require $\Omega(n)$ and $\Omega(\log^* n)$ worst-case rounds, respectively.
\begin{corollary}\label{cor:3ColNA}
    The $3$-coloring problem, on paths, has deterministic node-averaged complexity $\Omega(\log^* n)$.
\end{corollary}

We now define the family of graphs that we use to prove our lower bounds.
\begin{definition}[$k$-hierarchical lower bound graph]\label{def:LowerBoundGraph}
    For some parameters $\ell_1,\ldots,\ell_k$, consider the following recursive construction. Start from a path of length $\ell_k$, which is called \emph{path of level $k$}, and its nodes are called \emph{nodes of level $k$}. Then, recursively, for $i = k-1,\ldots,1$ do the following. For each path $P$ of level $i+1$, for each node $v$ of $P$, create a path $P'$ of length $\ell_{i}$ and connect one endpoint of $P'$ to $v$. The path $P'$ is a \emph{path of level $i$} and its nodes are \emph{nodes of level $i$}. This graph is called \emph{$k$-hierarchical lower bound graph}.
\end{definition}
An example of this construction, for $k=2$, is depicted in \Cref{fig:lowerBoundGraph}. We note a small corollary about this construction, that will be useful in later sections. It is just an immediate consequence of the definition.

\begin{corollary}\label{cor:SizeOfLevels}
Consider the computation of levels as in \cref{def:2.5Col}. For all $i \in \{1, \ldots, k\}$, let $L_i$ be the set of level $i$ nodes. Then, $|L_i| \in \Omega(\prod_{i \leq j \leq k} \ell_{j})$.
\end{corollary}

We now prove the following lemma.
\begin{lemma}\label{lem:slow-or-d}
    Let $n$ be the size of the $k$-hierarchical lower bound graph, as a function of  $\ell_1,\ldots,\ell_k$.
    Let $\mathcal{S}= \{1,\ldots,n^c\}$ be the ID space. Let $\mathcal{G}_k$ be the family of $k$-hierarchical lower bound graphs obtained by assigning IDs from $\mathcal{S}$.
    Let $\mathcal{A}$ be an algorithm for $k$-hierarchical $3\frac{1}{2}$-coloring. 
    Then, one of the following holds:
    \begin{itemize}
        \item either there exists an instance in $G \in \mathcal{G}_k$ and a value $i \in \{1,\ldots,k-1\}$ such that at least half of the nodes of level $i$ spend at least $\ell_i / 10$, or 
        \item there exist $n^{c-1} / 4^k$ graphs in $\mathcal{G}_k$, using disjoint sets of IDs, such that, for each of them, for all $i$ in $\{1,\ldots,k-1\}$, all nodes of level $i$ output $D$.
    \end{itemize}
\end{lemma}
\begin{proof}
    We split the ID space into $n^{c-1}$ sets of size $n$, and for each set $S$ we construct an instance by using an arbitrary ID assignment over the IDs from $S$. Let $I_0$ be the set of resulting instances.
    We will prove by induction that, either we can construct an instance in which the first condition of the lemma is satisfied, or we can construct a sequence of sets $I_1,\ldots, I_k$ satisfying that, for all $i$, for all instances in $I_i$, by running $\mathcal{A}$, all nodes at level at most $i$ output $D$, where the size of $I_i$ is at least $|I_0| / 4^i$. Note that for $i=0$ the claim trivially holds. 
    
    Hence, in the following, assume that, for all instances in $I_{i-1}$, by running $\mathcal{A}$, all nodes at level at most $i-1$ output $D$. 
    Then, construct $I_i$ as follows. We run $\mathcal{A}$ on all instances in $I_{i-1}$, and:
    \begin{itemize}
        \item either there exists an instance $I \in I_{i-1}$ in which at least half of the nodes of level $i$ spend at least $\ell_i / 10$, or
        \item in all instances $I \in I_{i-1}$, at least half of the nodes of level $i$ spend at most $\ell_i /10$.
    \end{itemize}
    Note that, if the first case applies, then the claimed statement holds. Hence, in the following assume that the second condition holds. 
    
    We now prove that all nodes that run in at most $\ell_i /10$ must output $D$. Since $n \ge \ell_i$, for each node $u$ that runs in at most $\ell_i/10$, we can find some other node $v$ that also runs in at most $\ell_i/10$, and such that $u$ and $v$ have disjoint views (that is, their radius-$\ell_i/10$ neighborhood is disjoint).
    Note that, by inductive hypothesis, all nodes in level $i$ cannot output $E$, and hence they either output $B$,$W$, or $D$.
    If both $u$ and $v$ output $B$ or $W$, then by adjusting the parity accordingly, we can create a single path containing both $u$ and $v$ in which their outputs cannot be completed into a proper $2$-coloring, which is a contradiction. If at least one node outputs $B$ or $W$, and the other outputs $D$, we can create a single path that directly gives a contradiction. Hence, all nodes that run in at most $\ell_i /10$ rounds output $D$.
    
    We now lower bound the number of paths containing at least one node $u$ that satisfies the following properties:
    \begin{itemize}
        \item Node $u$ terminates in at most $\ell_i /10$ rounds (and hence it outputs $D$).
        \item Node $u$ is at distance at least $\ell_i /10 + 1$ from the endpoints of the path containing $u$ (and hence $u$ does not see outside the path). 
    \end{itemize}
    Let $x$ be the total number of paths of level $i$, and let $p$ the number of paths satisfying the desired properties. We have $x \ell_i /2$ nodes that output $D$ and terminate in at most $\ell_i/10$ rounds. At most $ 2 x \ell_i / 10$ of them can be too close to the endpoints of their paths. In the worst case, the remaining ones are in the same paths. Hence, we can lower bound $p$ as follows.
    \[
    p \ge (\frac{x \ell_i}{2} - \frac{2x \ell_i}{10})\cdot \frac{10}{8 \ell_i} = \frac{3x}{8} > x/4
    \]
    We thus get that, for every $4$ instances in $I_{i-1}$ (recall that they use disjoint sets of IDs), we can construct a single instance in which all nodes of level $i$ output $D$. Such instances are added to $I_i$.
\end{proof}

We are now ready to prove \Cref{lem:lb-3half-noweights}. Assume for a contradiction that, for some $k$, the $k$-hierarchical $3\frac{1}{2}$-coloring can be solved in $o((\log^*n)^{1/2^{k-1}})$ rounds. 
Let $t := (\log^*n)^{1/2^{k-1}}$, and for $1 \le i \le k-1$, let $\ell_i := t^{2^{i-1}}$. Then, let $\ell_k := \lfloor n / (\prod_{1 \le i \le k-1} \ell_i) \rfloor$.
We consider the $k$-hierarchical lower bound graphs for these parameters. Observe that, by construction, the total number of nodes is $\Theta(n)$.

We apply \Cref{lem:slow-or-d}. If the first case of the lemma applies, we obtain that there exists an instance and a value of $i$ in which at least half of the nodes of level $i$ spend at least $\ell_i / 10$. In this case, we get that the node-averaged complexity is at least 
\[
\frac{1}{2n} \cdot \prod_{i \le j \le k} \ell_j \cdot \frac{\ell_i}{10} \ge \frac{1}{2n} \cdot \frac{n}{ 2(\prod_{1 \le i \le k-1} \ell_i)} \cdot \prod_{i \le j \le k-1} \ell_j \cdot \frac{\ell_i}{10} = \frac{\prod_{i \le j \le k-1} \ell_j}{4(\prod_{1 \le i \le k-1} \ell_i)} \cdot \frac{\ell_i}{10} = \frac{\ell_i}{40 \prod_{1 \le j \le i-1} \ell_j}.
\]
Note that $\ell_i = t^{2^{i-1}}$, and that $\prod_{1 \le j \le i-1} \ell_j = t^{2^{i-1}-1}$. Hence, the node-averaged complexity is at least $t/40$, which is a contradiction.

Hence, the second case of the lemma applies. We obtain that there are many instances in which all nodes, except the ones at level $k$, output $D$, which implies that the nodes at level $k$ need to properly $3$-color the path at level $k$. We now prove that the nodes in the path at level $k$ need to spend $\Omega(\log^* n)$ rounds on average, by a reduction from the hardness of $3$-coloring given by \Cref{cor:3ColNA}.
\begin{lemma}\label{lem:3ColLastPath}
    Assume all nodes at levels strictly smaller than $k$ output $D$. Then, the nodes in the path at level $k$ need to spend $\Omega(\log^* \ell_k)$ rounds on average.
\end{lemma}
\begin{proof}
Assume the ID space in which $\mathcal{A}$ runs is $\{1,\ldots,n^c\}$.
Let $f$ be an injective function that maps the IDs from $\{1,\ldots,n^{c-1}\}$ into instances given by \Cref{lem:slow-or-d}. 
For $c$ large enough, the $3$-coloring problem on paths has still $\Omega(\log^* n)$ deterministic node-averaged complexity even if the IDs are in the range $\{1,\ldots,n^{c-1}\}$. We now show that we can use $\mathcal{A}$ to solve $3$-coloring on paths, assuming the IDs are in $\{1,\ldots,n^{c-1}\}$, by using a standard simulation argument. Given a $3$-coloring instance of size $\ell_k$, nodes create a virtual instance of $k$-hierarchical $3\frac{1}{2}$-coloring as follows. Each node $v$ takes an arbitrary node $u$ of degree $2$ in the path of level $k$ in $f(v)$, and takes the tree $T_v$ induced by all nodes of lower layers reachable from $u$. Each node $v$ pretends to be connected not only with its neighbors in the path, but also to $T_v$. By simulating $\mathcal{A}$ on this instance, we obtain that all nodes in the virtual trees output $D$, and hence we get a $3$-coloring in the path. Since the nodes in the path are $\ell_k$, and since by \Cref{cor:3ColNA} the $3$-coloring problem on a path of length $\ell_k$ requires $\Omega(\log^* \ell_k)$, the claim follows.
\end{proof}

Since $\log^* \ell_k \in \Omega(\log^* n)$,  for some large-enough constant $c$ we obtain the following lower bound on the node-averaged complexity of $\mathcal{A}$:
\[
\frac{1}{n} \cdot \ell_k \cdot c \log^* n \ge \frac{1}{2c} \frac{\log^* n}{t^{2^{k-1}-1}} = \frac{1}{2c} t,
\]
which is a contradiction.

\begin{figure}
	\centering
	\includegraphics[width=0.6\textwidth]{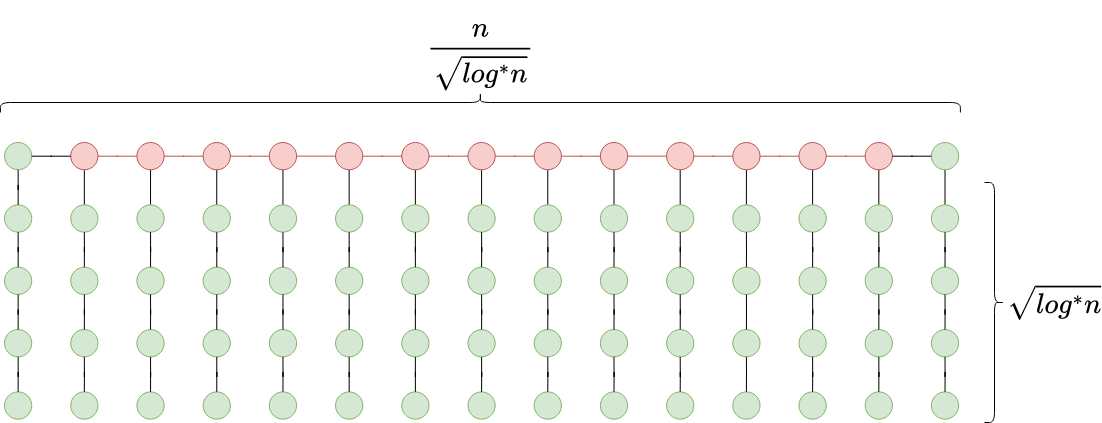}
	\caption{The figure illustrates the lower bound graph for $2$-hierarchical $3\frac{1}{2}$-coloring. All nodes of level 1 are colored green and all nodes of level 2 are colored red. The nodes in level 1 form paths of length $\sqrt{\log^*n}$ (and $\sqrt{\log^*n}+1$ for the left and rightmost). The level 2 nodes are just one long path of length $\frac{n}{\sqrt{\log^*n}} - 2$}
	\label{fig:lowerBoundGraph}
\end{figure}

This, together with \cref{lem:UB3.5Col}, concludes the proof of \cref{thm:3.5Col}. In the next section, we will see a way of extending $2\frac{1}{2}$- and $3\frac{1}{2}$-coloring in a way that turns them into something like a weighted version of original problems. Here weighted means weighted in terms of the node-averaged complexity. The idea is to attach some special \emph{weight} nodes to every normal node that has to wait until the normal node terminates. We start with doing this for $2\frac{1}{2}$-coloring, as there the analysis is easier.

\section{Weighted Problems}\label{sec:weighted-problems}
In this section, we introduce weighted versions of the problems of $2\frac{1}{2}$- and $3\frac{1}{2}$-coloring.
These problems will be instrumental in showing the density of non-empty complexity classes in the polynomial and the sub-$\log^* n$ regimes.
%\bs{write more nice sentences that connect everything together}

Formally, we will denote the weighted versions of $k$-hierarchical $2\frac{1}{2}$- and $3\frac{1}{2}$-coloring by $\Pi^Z_{\Delta, d, k}$, where $Z \in \{ 2\frac{1}{2}, 3\frac{1}{2} \}$ indicates whether the problem is a weighted version of $2\frac{1}{2}$- or $3\frac{1}{2}$-coloring, and the other three parameters $\Delta$, $d$, and $k$ indicate an upper bound for the maximum degree of the considered graphs, a parameter for tuning the complexity of the problem, and the parameter $k$ from the definition of $k$-hierarchical $2\frac{1}{2}$- or $3\frac{1}{2}$-coloring, respectively.

% \paragraph{The LCL $\Pi^Z_{\Delta, d, k}$}
% The aforementioned LCLs that provide the desired node-averaged complexities depend on three parameters $\Delta$, $d$, and $k$ and can be described as follows.

\begin{definition}\label{def:WeightedColoring}
Let $\Delta$, $d$, and $k$ be positive integers satisfying $\Delta \geq d + 3$, and let $Z \in \{ 2\frac{1}{2}, 3\frac{1}{2} \}$. 

The LCL $\pidd^Z$ has input label set $\sinn := \{ \act, \wei \}$, i.e., each node is labeled with either $\act$ or $\wei$. %and output label set $\sout := \{ \dec, \cop, \con \}$.\bs{add here the output labels from $2.5$-coloring}
In the former case, we call the node an \emph{active node}, in the latter we call the node a \emph{weight node}.

Each active node has to output a label from $\Sigma^{\act}_{\out}$, where $\Sigma^{\act}_{\out}$ is the output label set of $k$-hierarchical $Z$-coloring.
Each weight node $v$ has to output a label from the set $\{ \dec, \con, \cop \}$.
If $v$ outputs $\cop$, then it has to additionally output a label from $\Sigma^{\act}_{\out}$.
%and select an adjacent node $u$.
%In that case, we say that $v$ \emph{points} to $u$, and we will imagine the edge $\{ v,u \}$ as being oriented from $v$ to $u$.
%Moreover,
We will call this additional output label from $\Sigma^{\act}_{\out}$ the \emph{secondary} output of $v$.

The output is correct if it satisfies the following properties.
\begin{enumerate}
    \item \label{prop:twoand} The connected components induced by active nodes satisfy the constraints of $k$-hierarchical $Z$-coloring provided in \Cref{def:2.5Col} and \Cref{def:3.5Col}.
    \item \label{prop:conorcop} Each weight node that is adjacent to at least one active node must output $\con$ or $\cop$.
    \item \label{prop:con} For each weight node $v$ that outputs $\con$, at least two neighbors of $v$ are active or output $\con$ as well. (To be clear, if $v$ has one active neighbor and one weight neighbor that outputs $\con$, then this property is satisfied.)
    \item \label{prop:allbutd} For each node $v$ that outputs $\cop$, at most $d$ neighbors of $v$ output $\dec$.
    \item \label{prop:copylab} If a weight node $v$ that outputs $\cop$ has an active neighbor, then the secondary output of $v$ is identical to the output of at least one active neighbor of $v$. Moreover, for any two adjacent weight nodes $v, w$ that both output $\cop$, the secondary output of $v$ is identical to the secondary output of $w$.%\bs{give intuitive explanation of these problems if there is time}
\end{enumerate}
\end{definition}

As the constraints of $k$-hierarchical $Z$-coloring mentioned in Property~\ref{prop:twoand} as well as all other properties mentioned in the problem description only depend on a constant-radius neighborhood of the respectively considered nodes, the problem $\pidd^Z$ is an LCL problem.

We start with some intuition on these problems. We use inputs to decide which nodes are weight nodes and which nodes have to solve $2\frac{1}{2}$-coloring (resp. $3\frac{1}{2}$-coloring). Property~\ref{prop:conorcop} and \ref{prop:copylab} ensure that weight nodes don't simply all output $\dec$. Notice that in a disconnected component consisting of just weight nodes, they can all simply output $\dec$. However, if there is at least one active node adjacent, then Property~\ref{prop:conorcop} creates at least one weight node with $\con$ or $\cop$.\\
We first ignore $\con$ labels and just think about a component of weight nodes $C$ with just a single active node $v$ adjacent to $C$. As a result of Property~\ref{prop:conorcop}, the weight nodes adjacent to $v$ all have to output $\cop$, and by Property~\ref{prop:copylab} they must have as secondary output the same output as $v$. Furthermore, this output \emph{spreads} through $C$ because nodes adjacent to a $\cop$ node also have to output $\cop$. As a result all of the nodes in $C$ must wait for $v$ to decide on its output label and can only then propagate this output throughout $C$. This will cause very long dependency chains and make this propagation take very long. To avoid very long dependency chains we allow some nodes to output $\dec$ in Property~\ref{prop:allbutd} and exclude $\dec$ nodes from the need to propagate the labels further.
To understand Property~\ref{prop:con} and the $\con$ label we have to consider components of weight nodes with multiple active neighbors, for example a path of 4 nodes with both endpoints being active nodes and the two middle nodes being weight nodes. If we didn't have the $\con$ label, both the middle nodes would have to output $\cop$ and have as secondary output the labels that the endpoints chose. However, if these two labels are different we would get a contradiction to Property~\ref{prop:copylab}. So if we consider a larger component of weight nodes with multiple active nodes adjacent to it, we can use the $\con$ label to \emph{connect} different active nodes that are too close together. By Property~\ref{prop:con} the nodes that output $\con$ are paths between active nodes and nothing else, so this is the only use case for this $\con$ label. We will later see, in the upper bound, that if these active nodes are sufficiently far apart we don't need to use $\con$, and in our lower bound constructions weight components will only ever be adjacent to one active node, so the label $\con$ can never be used.

\paragraph{Tour of the lower bound.}
In order to show the effectiveness of our construction, we start by showing that if we attach $W$ weight nodes in a balanced $\Delta$-regular tree to a single active node $v$, then $W^x$ of these weight nodes have to actually output $v$'s output as secondary output (\cref{lem:weightOfTree}). Here $x$ denotes some efficiency factor that is based on the parameters chosen. We then elaborate on this by showing that if we instead spread these $W$ nodes evenly in weight trees attached to $\ell$ different active nodes, we have that $\ell (\frac{W}{\ell})^x$ many nodes must copy an output from an active node (\cref{cor:EvenWeightWorst}). We use these insights to extend the classical lower bound constructions for these problems by putting a linear amount of weight evenly distributed on nodes of each level (\cref{def:WeightedConstruction}). Then, in order to actually prove our lower bounds, we first find ID assignments for these problem instances such that the algorithm behaves in a predictable way, such that it either completely colors all nodes of a level, or none (\cref{lem:SameBehaviourLowerRegime}). We then calculate how much time it would cost to color each level as a function of the length of the paths given by parameters $n^{\alpha_i}$ for level $i$ (\cref{cor:PolyRegimeTerms}). We get an optimisation problem that makes it so that no matter which level the algorithm decides to color, it should take as long as possible (\cref{cor:OptGivesLB}). 
We obtain an analytic solution that gives us the length of the paths as a function of the efficiency factor $x$, and obtain a lower bound for every choice of parameters $\Delta, d, k$ (\cref{lem:OptValues}).

\paragraph{Tour of the upper bound.}
We start by restating the problem of outputting the correct weight labels as its own problem, and then develop an algorithm that only takes care of solving this weight problem.
The guarantees that this algorithm is able to achieve matches the efficiency factor from the lower bound section (\cref{lem:copynumber}). 
We then use the generic algorithm from \cref{sec:khierarchical312} to allow the active nodes to compute a correct solution to the $2\frac{1}{2}$-coloring problem (resp. $3\frac{1}{2}$).
To upper bound the node-averaged complexity, we calculate how much time would be spent coloring each level based on the parameters of the generic algorithm and assuming a worst case distribution of weight. We then see that, with the choice of parameters given by the lower bound section, we obtain a tight upper bound (Proof of \cref{thm:UpperUpperBound}).

\section{Weighted Lower Bounds}\label{sec:weightedlower}
For the rest of this section, fix $k,\Delta,d$ to be some integer constants, such that $\Delta \geq d +3$. 
Since not necessarily all of the weight nodes need to copy a label from an active node, the efficiency of the weight decreases. The next lemma gives a lower bound on how many nodes have to wait for active nodes. 

\begin{lemma}\label{lem:weightOfTree}
    Let $\Delta,k,d\geq 2$ be constant such that $\Delta\geq d+3$. Consider the LCL $\Pi^{Z}_{\Delta,d,k}$ for any $Z$. Consider an \emph{active} node $v$ with a balanced $\Delta$ regular tree $T_w$ of $w$ weight nodes attached to it, then at least $w^{\log(\Delta -1 -d)/\log(\Delta -1)}$ of these weight nodes have to output $\cop$ and have as secondary output the same output as $v$.
\end{lemma}
\begin{proof}
Let $r$ be the weight node that is directly adjacent to $v$, think of the tree of weight nodes as rooted at $r$. Since $v$ is the only \emph{active} node attached to $T_w$, none of the nodes in $T_w$ can output $\con$. Consider some non-leaf node $u$ in this tree of weight nodes. It has to have one edge towards its parent (or for $r$ one edge towards $v$) and then $\Delta -1$ children. As a result this tree has height $\log_{\Delta -1}(w)$. 
Now if $u$ does not output $\dec$, then  $\Delta -1 -d$ of $u$'s children have to copy the output of $u$. Since $r$ is directly adjacent to an active node it has to output $\cop$ and use $v$'s output as secondary output. Then at least  $\Delta -d -1$ of $r$'s children have to also copy the output, then $\Delta -1 -d$ of the childrens children have to do the same and so on. As a result the number of nodes copying the output of $v$ is at least:
\begin{gather*}
(\Delta -1 -d)^{\log_{\Delta -1}(w)} = (\Delta - 1)^{\log_{\Delta-1}(\Delta -1 - d) \cdot \log_{\Delta -1}(w)} = \\
w^{\log_{\Delta-1}(\Delta -1 - d)} = w^{\log(\Delta -1 -d)/\log(\Delta -1)}\qedhere
\end{gather*}
\end{proof}

From now on, let $x = \frac{\log(\Delta -d - 1)}{\log(\Delta - 1)}$. Notice that the more weight we put on a tree, the less efficient it becomes, so to maximize the efficiency of our weight it makes sense to distribute it as evenly as possible. Consider if we attach a tree of $w/\ell$ weight nodes to $\ell$ active nodes. By \Cref{lem:weightOfTree} for each of these trees $(w/\ell)^{x}$ nodes have to copy the output of their respective active node. Now clearly the total amount of copying nodes is:
\[
(w/\ell)^{x} \ell = w^x \ell^{1-x} > w^x
\]

\begin{corollary}\label{cor:EvenWeightWorst}
Let $x = \log(\Delta -1 -d)/\log(\Delta -1)$. By distributing $w$ weight nodes evenly among $\ell$ $\Delta$-regular trees, each attached to an active node the amount of nodes that have to output $\cop$ is $w^x \ell^{1 - x}$. Additionally such an even distribution results in the maximum amount of nodes copying.
\end{corollary}
Notice that the additionally part is a direct consequence of Jensen's Inequality for the concave function $\phi(a) = a^x$ (for $x \leq 1$), if we split the weight into $\ell$ parts $w_1, \dots, w_\ell$:
\[
\phi\left(\frac{\sum w_i}{\ell}\right) \geq \frac{\sum\phi(w_i)}{\ell}. \text{ So, we get that: } \ell \phi\left(\frac{w}{\ell}\right) \geq \sum \phi(w_i)
\]

Equipped with the knowledge from these last two results, we now know how to extend the lower bound construction \cref{def:LowerBoundGraph}. We define a new weighted version of it, illustrated in \cref{fig:WLBConstruction}. The idea is to use $\frac{n}{k}$ nodes to create the lower bound graph from \cref{def:LowerBoundGraph}. In that construction the number of nodes of level 1 is already linear in $n$, so we do not need additional weight there. For the rest of the levels, we balance things out, by putting $\frac{n}{k}$ weight nodes on the nodes of level $2, \ldots k$. As a result every layer has a linear amount of weight. Note that since the number of nodes in lower levels is significantly larger, this also means that the weight is more efficient. We distribute this weight evenly in balanced $\Delta$-regular trees, so we get the bounds from \cref{cor:EvenWeightWorst}. This gives us our new construction.

\begin{definition}[Weighted Construction]\label{def:WeightedConstruction}
Given a target size of $n$ and parameters $\ell_1, \ldots, \ell_k$ such that $\prod_{1\leq i\leq k}\ell_i = n$, define $n' = n/k$, then start with the counterexample graph $G'$ from \cref{def:LowerBoundGraph} with $n'$ nodes, by setting $\ell_i' = \ell_i \cdot \frac{1}{k^{1/k}}$ (Then $\prod_{1\leq i\leq k} \ell_i' = n'$).\\
Define $L_1, \ldots L_k$ to be sets of nodes such that $L_i$ is exactly the set of level $i$ nodes in $G'$ ($V(G') = \bigcup_{1 \leq i \leq k} L_i$). For each of $L_2, \ldots, L_k$ distribute $\frac{n}{k}$ nodes evenly in balanced $\Delta$-regular trees, each attached to one of the nodes in $L_i$. This is our weighted construction $G$, with a total of $n$ nodes, $\frac{n}{k}$ for $G'$ and then $k-1$ times $\frac{n}{k}$ for the weight nodes. \newline
\textbf{The Problem Instance:} To obtain an instance of $\Pi^{Z}_{\Delta,d,k}$, we have to assign input labels $\act,\wei$. All nodes in $G'$ get input label $\act$, all of the other nodes in the trees get input label $\wei$. As a result we have a valid instance for $\Pi^{Z}_{\Delta,d,k}$.
\end{definition}

\begin{figure}
	\centering
	\includegraphics[width=0.6\textwidth]{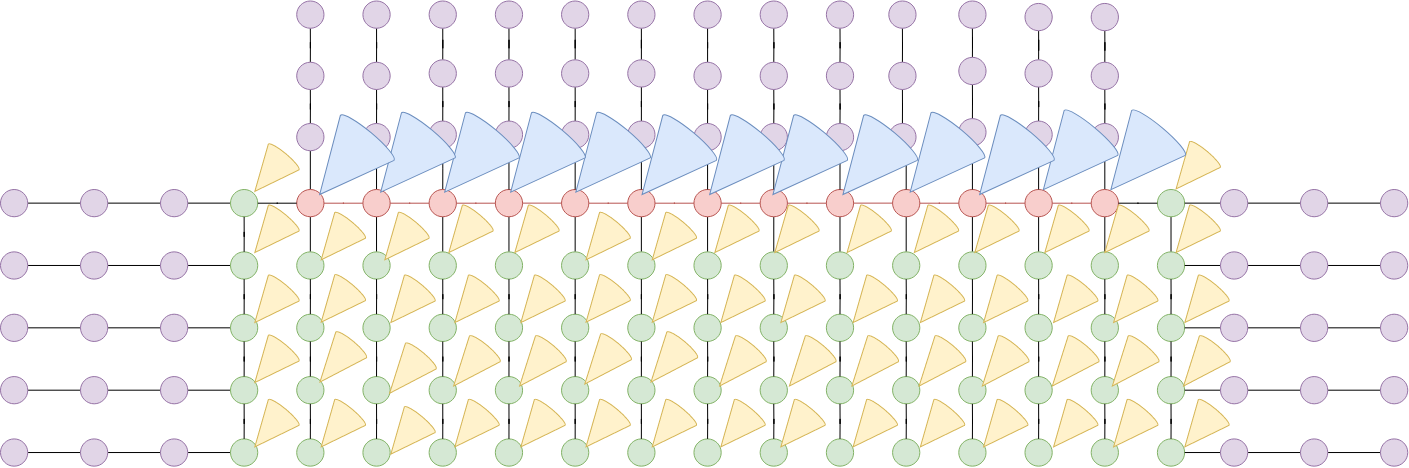}
	\caption{Showing the construction for $k=3$. The red nodes have level 3, the green nodes level 2 and the purple nodes have level 1 (Some purple paths from the green nodes in the middle are omitted). No trees are attached to level 1 nodes, but every level 2, or 3 node has a tree attached to it. The trees attached to level 3 nodes are larger, as the same number of $\frac{n}{k}$ nodes is evenly distributed among less nodes.}
	\label{fig:WLBConstruction}
\end{figure}

\iffalse 
We first study the construction a bit and start by making a crucial observation.

\begin{observation}\label{obs:indistinguishable}
For each path of level $i$ nodes in $G'$, the middle third (from $v_{|P|/3}$ to $v_{2|P|/3}$) of the nodes in the path have an indistinguishable $\frac{\ell'_i}{4}$ hop neighbourhood. As a result they must either spend $\Omega(\ell_i')$ time, or all output decline. This is true, because 2 coloring a path of length $\ell'_i$ takes $\Omega(\ell_i')$ time. 
\end{observation}
\fi

Crucially, our lower bound construction contains the original lower bound graph from \cref{def:LowerBoundGraph}. As a result any algorithm that correctly solves $\Pi^{Z}_{\Delta,d,k}$ inside any graph from \cref{def:WeightedConstruction} also solves either $k$-hierarchical $2\frac{1}{2}$-coloring, or $k$-hierarchical $3\frac{1}{2}$-coloring in $G'$.
So for any algorithm that correctly solves $\Pi^{3.5}_{\Delta,d,k}$, we can apply \cref{lem:slow-or-d}.
\begin{lemma}\label{lem:SameBehaviourLowerRegime}
Consider any correct deterministic algorithm that solves $\Pi^{3.5}_{\Delta,d,k}$, given any parameters $\ell_1, \ldots, \ell_k$ and an instance of the Weighted Construction. Then one of the following is true:
\begin{itemize}
    \item There exists an assignment of IDs to the Weighted Construction and a value $i \in \{1, \ldots, k-1\}$, such that at least half of the nodes of level $i$ run for at least $\ell'_i/10$ rounds.
    \item There exists an assignment of IDs to the Weighted Construction, such that for all of $i \in \{1, \ldots, k-1\}$ all nodes of level $i$ output $D$. 
\end{itemize}
\end{lemma}

Instead consider a randomised algorithm trying to solve $\Pi^{2.5}_{\Delta,d,k}$ in the weighted lower bound construction, then nodes cannot distinguish based on IDs, so the argumentation becomes a lot easier than the proof of \cref{lem:slow-or-d}.

\begin{lemma}\label{lem:SameBehaviourPolyRegime}
Consider any correct randomised algorithm that solves $\Pi^{2.5}_{\Delta,d,k}$, given any parameters $\ell_1, \ldots, \ell_k$ and an instance of the Weighted Construction. Then one of the following is true:
\begin{itemize}
    \item There exists a value $i \in \{1, \ldots, k-1\}$, such that at least a third of the nodes of level $i$ run for at least $\ell_i'/4$ rounds
    \item For all $i \in \{1, \ldots, k-1\}$ all nodes of level $i$ output $D$. 
\end{itemize}
\end{lemma}
\begin{proof}
Fix $i \in \{1, \ldots, k-1\}$ and let $P$ be an arbitrary path of level $i$ nodes. Then by construction $P$ has length $\ell_i'$ and furthermore the middle third of all of the nodes in $P$ have an indistinguishable $\frac{\ell'_i}{4}$ hop neighbourhood. As a result they all have to behave the same in the first $\frac{\ell'_i}{4}$ rounds. If they terminate in strictly less than $\frac{\ell'_i}{4}$ rounds, it means that they all have to follow the same output distribution. As a result, they cannot output a consistent 2-coloring and so for the algorithm to be correct, they have to output $D$. If instead they terminate after at least $\frac{\ell'_i}{4} \in \Omega(\ell'_i)$ time we are in exactly the other case.
\end{proof}

We can now use these results to obtain lower bounds on the total time an algorithm would spend on nodes of level $i$, regardless of whether or not we are solving $\Pi^{2.5}_{\Delta,d,k}$, or $\Pi^{3.5}_{\Delta,d,k}$. Notice that regardless of whether we invoke \cref{lem:SameBehaviourLowerRegime}, or \cref{lem:SameBehaviourPolyRegime}, at least a third of the level $i$ nodes spend $\Omega(\ell')$ time. So in either case we can prove the following.

\begin{lemma}\label{lem:ColLowLvlBound}
Consider the nodes of some level $1\leq i < k$. If at least a third of all level $i$ nodes spend $\Omega(\ell'_i)$ time, then the total amount of time used is $\Omega \left( n^x \cdot (\prod_{i\leq j\leq k }\ell_j)^{1-x} \cdot \ell_i \right)$
\end{lemma}
\begin{proof}
By \cref{cor:SizeOfLevels} $|L_i| \in \Omega(\prod_{i\leq j\leq k }\ell'_j)$. Since we distribute $\frac{n}{k}$ many nodes in balanced $\Delta$-regular trees of weight nodes, by \cref{cor:EvenWeightWorst} the amount of nodes that have to wait for nodes in $L_i$ to terminate is at least $(\frac{n}{k})^x|L_i|^{1-x}$. So if a third of the nodes in $|L_i|$ spend at least $\Omega(\ell'_i)$ time, then since the weight is distributed evenly, also a third of the $(\frac{n}{k})^x|L_i|^{1-x}$ nodes have to spend $\Omega(\ell'_i)$ time. So for the total amount of time spent, we get at least 
\begin{align*}
\left(\frac{n}{k}\right)^x|L_i|^{1-x} \cdot \Omega(\ell'_i) = \left(\frac{n}{k}\right)^x \left(\Omega\left(\prod_{i\leq j\leq k }\ell'_j\right) \right)^{1-x} \cdot \Omega \left( \frac{\ell_i}{k^{1/k}} \right) = \Omega \left( n^x \cdot \left(\prod_{i\leq j\leq k }\ell_j\right)^{1-x} \cdot \ell_i \right)
\end{align*}
The last inequality just comes from hiding a bunch of constants in the $\Omega$ notation. $\frac{1}{k^x}$ is a constant, because $k$ and $x$ are constants, the same goes for the $k-i + 1$ factors of $\frac{1}{k^{1/k}}$ from transforming the $\ell'_j$ into $\ell_j$.
\end{proof}

We still need something to argue about the level $k$ path, since \cref{lem:ColLowLvlBound} only holds for $i<k$.

\begin{lemma}\label{lem:ColHighLvlBound}
Computing a $2$-coloring of the level $k$ path  takes $\Omega(n^x \cdot \ell_k^{1-x} \cdot \ell_k)$ total time for any randomised algorithm and if we are in the second case of \cref{lem:SameBehaviourLowerRegime} then deterministically computing a $3$-coloring takes $\Omega(n^x \cdot \ell_k^{1-x} \cdot \log^* \ell_k)$ total time.
\end{lemma}
\begin{proof}
Again, because of \cref{cor:SizeOfLevels}, $|L_k| \in \Omega(\ell'_k) = \Omega(\ell_k)$, so since we distribute $\frac{n}{k}$ many nodes in balanced $\Delta$-regular trees of weight nodes, by \cref{cor:EvenWeightWorst} the amount of nodes that have to wait for nodes in $L_k$ to terminate is at least $(\frac{n}{k})^x|L_k|^{1-x}$.\\
If we want to $3$-color the level $k$ path, then by \cref{lem:3ColLastPath} this takes node-averaged time $\Omega(\log^* \ell'_k)$. This then immediately gives the desired bound.

Now consider the case where we want to 2-color the path. By construction the middle third of the nodes in the level $k$ path have an indistinguishable $\frac{\ell'_k}{4}$ hop neighborhood (that is, they cannot see the ends). If they spend less than $\frac{\ell'_k}{4}$ rounds, then they have to follow the same output distribution, so they cannot output a consistent $2$-coloring. So the total time to $2$-color the path is lower bounded by
\[
\Omega\left(\left(\frac{n}{k}\right)^x \frac{\left(\ell'_k\right)^{1-x}}{3} \cdot \frac{\ell'_k}{4} \right) \in \Omega(n^x \cdot \ell_k^{1-x} \cdot \ell_k)\qedhere
\]
\end{proof}

In the next section we will choose explicit $\ell_1, \ldots, \ell_k$ and prove a lower bound for weighted $2\frac{1}{2}$-coloring. For weighted $3\frac{1}{2}$-coloring the lower bound proof is very similar with different choices of $\ell_1, \ldots, \ell_k$. 

\subsection{The Polynomial Regime} \label{sec:polylowerbound}
Since the lemmas are very generic, we will work with some more concrete values. Let $\ell_1 = n^{\alpha_1}, \ell_2 = n^{\alpha_2},\ldots, \ell_{k-1} = n^{\alpha_{k-1}}$. Lastly to get $\prod_{1\leq i \leq k} \ell_i = n$, we choose $\ell_k = n^{1 - \sum_{j<k}\alpha_j}$. We will deduce an optimisation problem, that will give us the exact values for $\alpha_1, \ldots, \alpha_{k-1}$.\\

We now just plug these values in the lemmas from before and immediately get lower bounds on the total amount of time spent. By normalising with $\frac{1}{n}$ we immediately get lower bounds on the node-averaged complexity.

\begin{corollary}\label{cor:PolyRegimeTerms}
Let $\ell_1, \ldots, \ell_k$ be as above. Then, if an algorithm computes a 2-coloring of the level $k$ nodes, the node-averaged complexity is 
\[
\Omega \left(n^{1 + (x-2)\sum_{j<k}\alpha_j} \right) =: A_k
\]
For any $1 \leq i < k$ if a third of all of the level $i$ nodes spend $\Omega(\ell_i)$ rounds, then the node-averaged complexity is 
\[
\Omega\left(n^{(x-1)\sum_{j<i}\alpha_j + \alpha_i }\right) =:A_i
\]
\end{corollary}
\begin{proof}
According to \cref{lem:ColHighLvlBound}, the total amount of time spent for the 2-coloring of the level $k$ nodes is lower bounded by:
\[
\Omega(n^x \cdot \ell_k^{1-x} \cdot \ell_k) = \Omega\left(n^x \cdot \left(n^{1 - \sum_{j<k}\alpha_j}\right)^{1-x} \cdot n^{1 - \sum_{j<k}\alpha_j}\right) = \Omega\left(n^{2 + (x-2)\sum_{j<k}\alpha_j} \right).
\]
So the node-averaged time is at least 
\[
\Omega\left( \frac{1}{n} \cdot n^{2 + (x-2)\sum_{j<k}\alpha_j} \right) = \Omega \left(n^{1 + (x-2)\sum_{j<k}\alpha_j} \right).
\]

For any $1\leq i <k$, if we are in the case of \cref{lem:ColLowLvlBound} the total amount of time spent by level $i$ nodes is at least:
\begin{align*}
\Omega \left( n^x \cdot (\prod_{i\leq j\leq k }\ell_j)^{1-x} \cdot \ell_i \right) &= \Omega \left( n^x \cdot (n^{\sum_{i\leq j\leq k}\alpha_j})^{1-x} \cdot n^{\alpha_i} \right) = \left( n^{x + (1-x) (1 - \sum_{j<i}\alpha_j) + \alpha_i} \right)\\
    &= \Omega \left(n^{1 + (x-1)\sum_{j<i}\alpha_j + \alpha_i }\right),
\end{align*}
where we used that $\sum_{i\leq j\leq k}\alpha_j = 1 - \sum_{j<i}\alpha_j$, since $\prod_{1\leq i \leq k} \ell_i = n$.
Again we normalize by $\frac{1}{n}$ to obtain a lower bound on the node-averaged complexity.
\[
\frac{1}{n} \cdot \Omega \left(n^{1 + (x-1)\sum_{j<i}\alpha_j + \alpha_i }\right) = \Omega \left(n^{(x-1)\sum_{j<i}\alpha_j + \alpha_i }\right)\qedhere
\]
\end{proof}

Because of \cref{lem:SameBehaviourPolyRegime}, the paths in level $i$, either all decline, or nodes of level $i$ spend a lot of time. If the level $i$ nodes spend a lot of time, then by \cref{cor:PolyRegimeTerms} the node-averaged complexity is at least $A_i$. But if all of them decline, then the nodes in level $k$ must output a $2$-coloring, so according to \cref{cor:PolyRegimeTerms} the node-averaged complexity is $A_k$. Since one of the two cases must happen because of \cref{lem:SameBehaviourPolyRegime}, we always get a lower bound on the node-averaged complexity.
\begin{corollary}\label{cor:OptGivesLB}
Let $A_1, \ldots, A_k$ be as in \cref{cor:PolyRegimeTerms}, then any algorithm that correctly solves the weighted $2\frac{1}{2}$-coloring problem has node-averaged complexity
\[
\Omega\left(\min_{\alpha_1,\dots,\alpha_{k-1}}
	\left\{\begin{array}{lr}
       		A_1:= \;  n^{\alpha_1}, \\
		A_2 := \; n^{(x-1)\alpha_1+\alpha_2}, \\
		A_3 := \; n^{(x-1)(\alpha_1+\alpha_2) + \alpha_3},\\
		\dots,\\
		A_i := \; n^{(x-1)(\sum_{j<i} \alpha_j) + \alpha_i},\\
		\dots,\\
		A_k := \; n^{1 + (x-2) \sum_{j<k} \alpha_j} 
        \end{array}\right\} \right)
\]
\end{corollary}

As a result our lower bound is the strongest when we choose $\alpha_1,\ldots,\alpha_{k-1}$ such that the smallest of the $A_i$ terms is maximised. 
\[
\max_{\alpha_1,\dots,\alpha_{k-1}} \min
	\left\{\begin{array}{lr}
       	A_1:= \;  n^{\alpha_1}, \\
		A_2 := \; n^{(x-1)\alpha_1+\alpha_2}, \\
		A_3 := \; n^{(x-1)(\alpha_1+\alpha_2) + \alpha_3},\\
		\dots,\\
		A_i := \; n^{(x-1)(\sum_{j<i} \alpha_j) + \alpha_i},\\
		\dots,\\
		A_k := \; n^{1 + (x-2) \sum_{j<k} \alpha_j} 
        \end{array}\right\}
\]
Since all terms have the same base, we can instead optimise just the exponents. As long as the exponents are all positive, an optimal solution to this new problem is also an optimal solution the the original problem.
\[
\max \min_{\alpha_1,\dots,\alpha_{k-1}}
	\left\{\begin{array}{lr}
       	B_1:= \;  \alpha_1, \\
		B_2 := \; (x-1)\alpha_1+\alpha_2, \\
		B_3 := \; (x-1)(\alpha_1+\alpha_2) + \alpha_3,\\
		\dots,\\
		B_i := \; (x-1)(\sum_{j<i} \alpha_j) + \alpha_i,\\
		\dots,\\
		B_k := \; 1 + (x-2) \sum_{j<k} \alpha_j
        \end{array}\right\} 
\]

\begin{lemma}\label{lem:setTermsEqual}
The optimal solution is achieved when setting all terms equal.
\end{lemma}
\begin{proof}
By interpreting each $B_i$ as a function and noticing that $B_1$ only depends on $\alpha_1$, while all other terms have a negative dependence on $\alpha_1$, we can eliminate variables.
Since $\frac{\partial B_1}{\partial \alpha_1} = 1$ and $\frac{\partial B_i}{\partial \alpha_1} \leq 0$ for all $1<i \leq k$ the optimal value of 
\[
\max_{\alpha_1,\dots, \alpha_{k-1} } \min \{B_1, \min\{B_2, B_3,\dots,B_k \} \}
\] is achieved by setting $B_1 =  \min\{B_2, B_3,\dots,B_k \}$. We then iteratively apply this same argument on $\min\{B_2,\dots, B_k\}$ while treating $\alpha_1$ as a constant.
\end{proof}

As a result we can obtain a formula for the optimal solutions depending only on $x$, by setting the terms equal.

\begin{lemma}\label{lem:OptValues}
The optimal solution is given by

\begin{align}
\alpha_i = (2-x)\alpha_{i-1}, \text{ and }
\end{align}
\begin{align}
\alpha_1(x) = \frac{1}{\sum_{j=0}^{k-1} (2-x)^j}
\end{align}
Furthermore $\alpha_1(x) = B_1 = B_2 = \ldots = B_k$ holds for the optimal solution.
\end{lemma}
\begin{proof}
For $1<i<k$
\begin{align*}
B_i = B_{i-1} &\Rightarrow (x-1)(\sum_{j<i} \alpha_j) + \alpha_i = (x-1)(\sum_{j<i-1} \alpha_j) + \alpha_{i-1} \\
&\Rightarrow (x-1)\alpha_{i-1} + \alpha_i = \alpha_{i-1} \Rightarrow \alpha_i = (2-x) \alpha_{i-1} \\
\end{align*}

and using the above
\begin{align*}
B_1 = B_k  &\Rightarrow \alpha_1 = 1 + (x-2) \sum_{j=1}^{k-1} \alpha_j = 1 - (2-x)\alpha_1 \sum_{j=1}^{k-1} (2-x)^{j-1} \\
&\Rightarrow \alpha_1 = 1-\alpha_1\sum_{j=1}^{k-1} (2-x)^j \\
&\Rightarrow 1= \alpha_1 + \alpha_1 \sum_{j=1}^{k-1} (2-x)^j = \alpha_1 (2-x)^0 \sum_{j=1}^{k-1} (2-x)^j = \alpha_1 \sum_{j=0}^{k-1} (2-x)^j\\
&\Rightarrow  \alpha_1 = \frac{1}{ \sum_{j=0}^{k-1} (2-x)^j}\qedhere
\end{align*}
\end{proof}

Now, because of \Cref{cor:OptGivesLB}, we immediately get \Cref{thm:UpperLowerBound}. 
%
%the following theorem.
%
%\UpperLowerBound*
%

\subsection{The \texorpdfstring{$\log^*n$}{log* n} regime} \label{sec:LowerRegimeLB}
We prove a lower bound on the weighted version of $k$-hierarchical $3\frac{1}{2}$-coloring, so fix parameters $d,\Delta, k$ and let $x = \frac{\log(\Delta -d -1)}{\log(\Delta -1)}$ be the efficiency factor from \cref{lem:weightOfTree}. We follow the exact same procedure as in \cref{sec:polylowerbound}, but for $\ell_1 = (\log^*n)^{\alpha_1}, \ell_2 = (\log^*n)^{\alpha_2},\ldots, \ell_{k-1} = (\log^*n)^{\alpha_{k-1}}$. Lastly to get $\prod_{1\leq i \leq k} \ell_i = n$, we choose $\ell_k = n \cdot (\log^*n)^{-\sum_{j<k}\alpha_j}$. Again we will deduce an optimisation problem, that will give us the exact values for $\alpha_1, \ldots, \alpha_{k-1}$.

First we prove that any algorithm must spend a lot of time if it wants to color paths.
\begin{corollary}\label{cor:LowerRegimeTerms}
Let $\ell_1, \ldots, \ell_k$ be as above, then if an algorithm computes a $3$-coloring of the level $k$ nodes the node-averaged complexity is 
\[
\Omega \left((\log^*n)^{1 + (x-1)\sum_{j<k}\alpha_j} \right) =: A_k
\]
For any $1 \leq i < k$ if half of all of the level $i$ nodes, spend at least $\frac{\ell'_i}{10}$ time, then the node-averaged complexity is at least 
\[
\Omega\left( (\log^*n)^{(x-1) \sum_{j<i}\alpha_j) + \alpha_i} \right) =:A_i
\]
\end{corollary}
\begin{proof}
According to \cref{lem:ColHighLvlBound} the total amount of time spent for the 3-coloring of the level $k$ nodes is lower bounded by:
\begin{align*}
\Omega(n^x \cdot \ell_k^{1-x} \cdot (\log^*\ell_k)) &= \Omega\left(n^x \cdot \left(n \cdot (\log^*n)^{-\sum_{j<k}\alpha_j}\right)^{1-x} \cdot \log^*\left(n \cdot (\log^*n)^{-\sum_{j<k}\alpha_j}\right)\right) \\
    &= \Omega\left(n^1 \cdot (\log^*n)^{1 + (x-1)\sum_{j<k}\alpha_j} \right).
\end{align*}

So the node-averaged time is at least 
\[
\Omega\left( \frac{1}{n} \cdot n^1 \cdot (\log^*n)^{1 + (x-1)\sum_{j<k}\alpha_j} \right) = \Omega \left((\log^*n)^{1 + (x-1)\sum_{j<k}\alpha_j} \right).
\]

For any $1\leq i <k$, if half of all of the level $i$ nodes spend at least $\frac{\ell'_i}{10}$ time, then according to \cref{lem:ColLowLvlBound} the total amount of time spent for the 2-coloring of the level $i$ nodes is lower bounded by
\begin{align*}
\Omega \left( n^x \cdot (\prod_{i\leq j\leq k }\ell_j)^{1-x} \cdot \ell_i \right) 
&= \Omega \left( n^x \cdot \left(n \cdot (\log^*n)^{-\sum_{j<i}\alpha_i}\right)^{1-x} \cdot (\log^*n)^{\alpha_i} \right) \\
&= \Omega\left( n^1 \cdot (\log^*n)^{(x-1) \sum_{j<i}\alpha_j) + \alpha_i} \right),
\end{align*}
where we used that $\prod_{i\leq j\leq k }\ell_j = n \cdot \left((\log^*n)^{-\sum_{j<i}\alpha_i}\right)$, since $\prod_{1\leq i \leq k} \ell_i = n$.
Again we normalize by $\frac{1}{n}$ to obtain a lower bound on the node-averaged complexity.
\[
\frac{1}{n} \cdot \Omega\left( n^1 \cdot (\log^*n)^{(x-1) \sum_{j<i}\alpha_j) + \alpha_i} \right) = \Omega\left( (\log^*n)^{(x-1) \sum_{j<i}\alpha_j) + \alpha_i} \right)
\]
\end{proof}

Because of \cref{lem:SameBehaviourLowerRegime}, the paths in level $i$, either all decline, or nodes in $i$ spend at least $A_i$ node-averaged time, according to \cref{cor:LowerRegimeTerms}. But if all of them decline, then the nodes in level $k$ must output a $3$-coloring, which by \cref{cor:LowerRegimeTerms} also implies node-averaged complexity at least $A_k$. Now because \cref{lem:SameBehaviourLowerRegime} states that one of these cases must happen, we get the following corollary.
\begin{corollary}\label{cor:LowerRegimeOptGivesLB}
Let $A_1, \ldots, A_k$ be as in \cref{cor:LowerRegimeTerms}, then any algorithm that correctly solves the weighted $3\frac{1}{2}$-coloring problem has node-averaged complexity
\[
\Omega\left(\min_{\alpha_1,\dots,\alpha_{k-1}}
	\left\{\begin{array}{lr}
       	A_1:= \;  (\log^*n)^{\alpha_1}, \\
		A_2 := \; (\log^*n)^{(x-1)\alpha_1+\alpha_2}, \\
		A_3 := \; (\log^*n)^{(x-1)(\alpha_1+\alpha_2) + \alpha_3},\\
		\dots,\\
		A_i := \; (\log^*n)^{(x-1)(\sum_{j<i} \alpha_j) + \alpha_i},\\
		\dots,\\
		A_k := \; (\log^*n)^{1 + (x-1) \sum_{j<k} \alpha_j} 
        \end{array}\right\} \right)
\]
\end{corollary}

As a result, our lower bound is the strongest when we choose $\alpha_1,\ldots,\alpha_{k-1}$ such that the smallest of the $A_i$ terms is maximised. 
\[
\max_{\alpha_1,\dots,\alpha_{k-1}} \min
	\left\{\begin{array}{lr}
       	A_1:= \;  (\log^*n)^{\alpha_1}, \\
		A_2 := \; (\log^*n)^{(x-1)\alpha_1+\alpha_2}, \\
		A_3 := \; (\log^*n)^{(x-1)(\alpha_1+\alpha_2) + \alpha_3},\\
		\dots,\\
		A_i := \; (\log^*n)^{(x-1)(\sum_{j<i} \alpha_j) + \alpha_i},\\
		\dots,\\
		A_k := \; (\log^*n)^{1 + (x-1) \sum_{j<k} \alpha_j} 
        \end{array}\right\}
\]
Since all terms have the same base, we can instead optimise just the exponents. As long as the exponents are all positive, an optimal solution to this new problem is also an optimal solution the the original problem.
\[
\max \min_{\alpha_1,\dots,\alpha_{k-1}}
	\left\{\begin{array}{lr}
       	B_1:= \;  \alpha_1, \\
		B_2 := \; (x-1)\alpha_1+\alpha_2, \\
		B_3 := \; (x-1)(\alpha_1+\alpha_2) + \alpha_3,\\
		\dots,\\
		B_i := \; (x-1)(\sum_{j<i} \alpha_j) + \alpha_i,\\
		\dots,\\
		B_k := \; 1 + (x-1) \sum_{j<k} \alpha_j
        \end{array}\right\} 
\]

Exactly in the same way as was proven in \cref{lem:setTermsEqual}, we also get that in this optimisation problem it is enough to set all the terms equal, so we do it to obtain a formula for the optimal solutions depending only on $x$.

\begin{lemma}\label{lem:OptValuesLowerRegime}
The optimal solution is given by
\begin{align}
\alpha_i = (2-x)\alpha_{i-1}
\end{align}
and
\begin{align}
\alpha_1(x) = \frac{1}{1 + (1-x)\sum_{j=0}^{k-2} (2-x)^j}
\end{align}
Furthermore $\alpha_1(x) = B_1 = B_2 = \ldots = B_k$ holds for the optimal solution.
\end{lemma}
\begin{proof}
For $1<i<k$
\begin{align*}
B_i = B_{i-1} &\Rightarrow (x-1)(\sum_{j<i} \alpha_j) + \alpha_i = (x-1)(\sum_{j<i-1} \alpha_j) + \alpha_{i-1} \\
&\Rightarrow (x-1)\alpha_{i-1} + \alpha_i = \alpha_{i-1} \Rightarrow \alpha_i = (2-x) \alpha_{i-1} \\
\end{align*}

and using the above
\begin{align*}
B_1 = B_k  &\Rightarrow \alpha_1 = 1 + (x-1) \sum_{j=1}^{k-1} \alpha_j = 1 - (1-x)\alpha_1 \sum_{j=1}^{k-1} (2-x)^{j-1} \\
&\Rightarrow \alpha_1 = 1-(1-x)\alpha_1\sum_{j=0}^{k-2} (2-x)^j \\
&\Rightarrow 1= \alpha_1 + (1-x)\alpha_1 \sum_{j=0}^{k-2} (2-x)^j = \alpha_1 (1 + (1-x)) \sum_{j=0}^{k-2} (2-x)^j\\
&\Rightarrow  \alpha_1 = \frac{1}{1 + (1-x)\sum_{j=0}^{k-2} (2-x)^j}
\end{align*}
\end{proof}

Now, because of \Cref{cor:LowerRegimeOptGivesLB}, we immediately get \cref{thm:LowerBoundWeighted3.5}.

\LowerLowerBound*

\section{Algorithm for Weighted \texorpdfstring{$2\frac{1}{2}$}{2.5}-Coloring}\label{sec:AlgWeighted212Coloring}

% \gs{this is not actually proven here, but instead in \cref{sec:density}}
% \bs{something like the following text (just in correct) should go somewhere (else) to glue things together}
% In this section, we will show that there are infinitely many nonempty asymptotic complexity classes in the polynomial regime \bs{put stuff in the preliminaries so that this statement makes sense}, and that these are ``dense'' in the node-averaged complexity landscape up to $\Theta(n^{1/2})$.
% More specifically, we will show that for any real number $0 \leq r \leq 1/2$ and any $\varepsilon > 0$, there exists some $r' \in [r - \varepsilon, r + \varepsilon]$ such that there exists an LCL with node-averaged complexity $\Theta(n^{r'})$.
% We will do so by explicitly constructing LCLs with node-averaged complexity $n^{r'}$, for any $r'$ from a selected dense subset of the interval $[0, 1/2]$.
% More precisely, let $R'$ be the set of all \bs{define $R'$ suitably, and anyways, rephrase this introduction in a way that is more conform with the actual theorems we are proving}

% For any $r' \in R'$ we will provide an LCL that we will show to have complexity $\Theta(n^{r'})$.
% \bs{end of what should go somewhere else}

In this section, we provide an algorithm that solves problem $\pitwo$ with node-averaged complexity $O(n^{\alpha_1})$, where $\alpha_1 = \frac{1}{\sum_{j = 0}^{k - 1}(2 - x)^j}$ and $x = \frac{\log (\Delta - 1 - d)}{\log (\Delta - 1)}$.
To this end, we introduce a new problem, called the $d$-free weight problem.
The $d$-free weight problem essentially constitutes a subproblem of $\pitwo$ (and simultaneously $\pithree$) that has to be solved to solve $\pitwo$ (and $\pithree$).
After stating the problem, we will design an algorithm that solves the problem, which will become an essential part of our algorithm for solving $\pitwo$.
In \Cref{sec:UsingOldAlgo}, we will make use of the $d$-free weight problem in a similar manner (using a different algorithm for solving it) to design an algorithm for $\pithree$.

\paragraph{The \texorpdfstring{$d$}{d}-free weight problem.}\label{def:WeightProblem}
\begin{restatable}{thm}{WeightProblem}
Let $\Delta$ and $d$ be positive integers satisfying $d < \Delta$ and $\Delta \geq 3$.
The $d$-free weight problem is an LCL on trees with input label set $\Sigma_{\inn} = \{A,W\}$ and output label set $\Sigma_{out}  = \{\dec, \con, \cop\}$.
Each node $v$ has one input label from $\Sigma_{\inn}$ and must output one output label from $\Sigma_{out}$.
%If $v$ has input label $W$ and outputs $\cop$, then it has to additionally select an adjacent node $u$.
%In that case, we say that $v$ \emph{points} to $u$, and we will imagine the edge $\{ v,u \}$ as being oriented from $v$ to $u$.
We call nodes that have input label $A$ \emph{adjacent} nodes and nodes with input label $W$ \emph{weight} nodes.
The (global) output is correct if it satisfies the following (local) properties:
\begin{enumerate}
    \item \label{prop:contwo} If a node $v$ with input label $A$ outputs $\con$, at least one neighbor of $v$ outputs $\con$ as well. If a node $v$ with input label $W$ outputs $\con$, at least two neighbors of $v$ output $\con$ as well.
    %\item \label{prop:point} If a node $v$ with input label $W$ outputs $\cop$ and points to a neighbor $u$, then $u$ outputs $\cop$ as well.
    %\item \label{prop:nocrossing} There are no two nodes $u, v$ with input label $W$ such that both $u$ and $v$ output $\cop$ and point to each other, i.e., the edge orientations obtained by orienting an edge $\{u, v\}$ from $u$ to $v$ if $u$ points to $v$ are well-defined in the sense that no edge is oriented in both directions (though an edge can be unoriented).
    \item \label{prop:dfree} For each node $v$ that outputs $\cop$, at most $d$ neighbors of $v$ output $\dec$. %and point to $v$.
    \item \label{prop:adjacent} Each node with input label $A$ outputs $\con$ or $\cop$.
\end{enumerate}
\end{restatable}

In the following, we provide an algorithm $\fa$ for solving the $d$-free weight problem with worst-case complexity $O(\log n)$ rounds.
We will later show that the output produced by $\fa$ does not only produce a correct output but additionally has useful properties that we will make use of in the analysis of the algorithm we design for $\pitwo$.

\paragraph{The algorithm.}
Algorithm $\fa$ proceeds as follows.
First, each node collects its $(3\lceil \log_{d + 1} n \rceil + 3)$-hop neighborhood.
Then, based on the collected information each node chooses its output according to the following rules.

Any node that lies on a path of length at most $(2\lceil \log_{d + 1} n \rceil + 2)$ between two nodes with input $A$ outputs $\con$.
All other nodes output $\dec$ or $\cop$ (as specified in the following).

Let $v$ be a node with input $A$ that does not output $\con$, and denote the set of nodes in the $( \lceil \log_{d + 1} n \rceil)$-hop neighborhood of $v$ by $U_v$ and the set of nodes in the $(\lceil \log_{d + 1} n \rceil + 1)$-hop neighborhood of $v$ by $\hat{U}_v$.
Let $\varphi \colon \hat{U}_v \rightarrow \{ \dec, \cop \}$ be a function with the following properties.
\begin{enumerate}
    \item \label{prop:first} For each $u \in \hat{U}_v \setminus U_v$, we have $\varphi(u) = \dec$.
    \item \label{prop:second} For each $u \in U_v$, we have $\varphi(u) \in \{ \dec, \cop \}$.
    \item $\varphi(v) = \cop$.
    \item \label{prop:last} For each node $u$ that outputs $\cop$, at most $d$ neighbors of $u$ output $\dec$.
    \item \label{prop:five} Among all function satisfying Properties (\ref{prop:first}) to (\ref{prop:last}), $\varphi$ is one that assigns $\cop$ to the minimum amount of nodes in $\hat{U}_v$ possible.
\end{enumerate}
Now, each node $u$ that is contained in the $( \lceil \log_{d + 1} n \rceil + 1)$-hop neighborhood of a node $v$
%with input $A$ that does not output $\con$
as described above outputs $\varphi(u)$ (where $\varphi$ is the function defined on the respective $\hat{U}_v$).
(Note that a node $u$ cannot be contained in the $( \lceil \log_{d + 1} n \rceil + 1)$-hop neighborhood of two such nodes $v$ as otherwise those two nodes would output $\con$.)
All nodes for which the above rules do not uniquely specify the output output $\dec$.
This concludes the description of $\fa$.

Note that the information contained in the $(2\lceil \log_{d + 1} n \rceil + 2)$-hop view of a node $v$ suffices to determine whether $v$ lies on a path of length at most $(2\lceil \log_{d + 1} n \rceil + 2)$ between two nodes with input $A$.
Consequently, the information contained in the $(3\lceil \log_{d + 1} n \rceil + 3)$-hop neighborhood of a node $u$ suffices to determine whether $u$ is contained in the $( \lceil \log_{d + 1} n \rceil + 1)$-hop neighborhood of a node $v$ with input $A$ that does not output $\con$.
Hence, the information contained in the $(3\lceil \log_{d + 1} n \rceil + 3)$-hop neighborhood collected by each node at the beginning of $\fa$ indeed suffices for each node to perform all step specified in the description of $\fa$.

However, from the description of $\fa$, it is not obvious that $\fa$ is well-defined, which we take care of with the following lemma.

\begin{lemma}\label{lem:fawell}
    Algorithm $\fa$ is well-defined.
\end{lemma}
\begin{proof}
    To show well-definedness, it suffices to prove that for any node $v$ with input $A$ that does not output $\con$ there exists a function $\psi \colon \hat{U}_v \rightarrow \{ \dec, \cop \}$ that satisfies Properties (\ref{prop:first}) to (\ref{prop:last}), which we do in the following.

    Let $v$ be a node with input $A$ that does not output $\con$.
    Consider the following (sequential) algorithm $\fa^*$ for defining $\psi$, where, abusing notation, we consider $\hat{U}_v$ as a rooted tree with root $v$.
    Set $\psi(v) := \cop$.
    Select the $\min\{d, \deg(v)\}$ children of $v$ for which the subtrees hanging from the children have the largest numbers of nodes (breaking ties arbitrarily).
    Set $\psi(u) := \dec$ for any node $u$ in any of those $\min\{d, \deg(v)\}$ subtrees (including their roots).
    Set $\psi(u) := \cop$ for any child of $v$ that is not in any of these subtrees.
    Then iterate on each subtree hanging from such a child $u$ with $\psi(u) = \cop$ where, for the min expression, we use $\min\{d, \deg(u) - 1\}$ (where, as before, $\deg(u)$ denotes the degree of $u$ \emph{in $\hat{U}_v$}).
    This concludes the description of $\fa^*$.

    We argue that the function $\psi$ produced by $\fa^*$ satisfies Properties (\ref{prop:first}) to (\ref{prop:last}).
    From the description of $\fa^*$, it is immediate that Properties (\ref{prop:second}) to (\ref{prop:last}) are satisfied.
    To show that also Property (\ref{prop:first}) is satisfied, consider the following claim: if a node $u \in \hat{U}_v$ with $\psi(u) = \cop$ has distance $i$ from $v$, then the number of nodes in the subtree hanging from $u$ is at most $|\hat{U}_v| / \left( (d + 1)^{i}\right)$.
    This claim implies that Property (\ref{prop:first}) is satisfied, as any node $u \in \hat{U}_v \setminus U_v$ has distance $\lceil \log_{d + 1} n \rceil + 1$ from $v$ and $|\hat{U}_v| / \left( (d + 1)^{i}\right) < 1$ for $i = \lceil \log_{d + 1} n \rceil + 1 \geq \lceil \log_{d + 1} |\hat{U}_v| \rceil + 1$, which implies that such a node $u$ cannot output $\cop$.

    In the following, we prove the claim by induction in $i$.
    For $i = 0$, the claim is trivially true.
    Now assume that the claim holds for some $i$, and consider some node $u \in \hat{U}_v$ with $\psi(u) = \cop$ that has distance $i + 1$ from $v$.
    Let $w$ denote the parent of $u$, which implies that $w$ has distance $i$ from $v$.
    Moreover, by the design of $\fa^*$, we know that $\psi(w) = \cop$. 
    By applying the induction hypothesis to $w$, we know that the subtree hanging from $w$ has at most $|\hat{U}_v| / \left( (d + 1)^{i}\right)$ nodes.
    Since, by the design of $\fa^*$, the subtree hanging from $u$ contains at most as many nodes as any of the $d$ ``heaviest'' subtrees hanging from children of $w$, it follows that the subtree hanging from $u$ has at most $\left(|\hat{U}_v| / \left( (d + 1)^{i}\right)\right) / (d+1) = |\hat{U}_v| / \left( (d + 1)^{i+1}\right)$ nodes, as desired.
\end{proof}

As it immediately follows from the algorithm description that $\fa$ has a runtime of $O(\log n)$, we obtain the following corollary.

\begin{corollary}\label{cor:dfreerun}
    Algorithm $\fa$ solves the $d$-free weight problem in $O(\log n)$ rounds.
\end{corollary}

Next, we collect some useful properties of the output that $\fA$ produces in \Cref{obs:conncopy,lem:copynumber}.

\begin{observation}\label{obs:conncopy}
    Each maximal connected component $C$ of nodes that output $\cop$ contains exactly one node with input $A$.
    Moreover, for each node $v$ with input $A$ that outputs $\cop$, exactly one such maximal connected component $C$ contains nodes that are also contained in the $( \lceil \log_{d + 1} n \rceil + 1)$-hop neighborhood of $v$, and this connected component $C$ is a subgraph of the $( \lceil \log_{d + 1} n \rceil)$-hop neighborhood of $v$.
\end{observation}
\begin{proof}
    From the description of $\fa$, it follows that any two nodes with input $A$ that output $\cop$ are at distance at least $2 \lceil \log_{d + 1} n \rceil + 3$ from each other.
    Moreover, all nodes that output $\cop$ are contained in the $( \lceil \log_{d + 1} n \rceil)$-hop neighborhood of such a node $v$.
    Hence, each maximal connected component of nodes outputting $\cop$ must be entirely contained in the $( \lceil \log_{d + 1} n \rceil)$-hop neighborhood of such a node $v$.
    Observe further that, by the design of $\fa$, the nodes in the $( \lceil \log_{d + 1} n \rceil + 1)$-hop neighborhood of such a node $v$ that output $\cop$ form a connected component and contain $v$.
    All statements made in the observation follow.
\end{proof}

\begin{lemma}\label{lem:copynumber}
    Let $v$ be a node with input $A$ that outputs $\cop$, and let $\hu$ denote the set of nodes of the $( \lceil \log_{d + 1} n \rceil)$-hop neighborhood of $v$.
    Let $\fu \subseteq \hu$ denote the subset of nodes in $\hu$ that output $\cop$.
    Then $|\fu| \leq 6|\hu|^{\frac{\log (\Delta - 1 - d)}{\log (\Delta - 1)}}$.
\end{lemma}
\begin{proof}
    We upper bound the number of nodes in $\fu$ by upper bounding the number of nodes that output copy under the function $\psi$ obtained by executing the sequential algorithm $\fa^*$ from the proof of \Cref{lem:fawell} on $\hu$.
    Let $K$ denote the number of nodes $u \in \hu$ with $\psi(u) = \cop$.
    Furthermore, let $K_1$ denote the number of nodes $u \in \hu$ with $\psi(u) = \cop$ that are at distance at most $\lfloor\log_{\Delta - 1} |\hu|\rfloor$ from $v$, and $K_2$ the number of nodes $u \in \hu$ with $\psi(u) = \cop$ that are at distance at least $\lfloor\log_{\Delta - 1} |\hu|\rfloor + 1$ from $v$.
    In particular, we have $k = K_1 + K_2$.

    We first bound $K_1$.
    By the description of $\fa^*$, the nodes in $\hu$ that output $\cop$ induce a subtree (which we can assume to be rooted at $v$) in which each node has at most $\Delta - 1 - d$ children, possibly except for $v$, which has at most $\Delta - d$ children.
    This implies
    \begin{align*}
        K_1 &\leq \frac{\Delta - d}{\Delta - 1 - d} \sum_{i = 0}^{\lfloor\log_{\Delta - 1} |\hu|\rfloor} (\Delta - 1 - d)^i \leq 2 \cdot 2(\Delta - 1 - d)^{\lfloor\log_{\Delta - 1} |\hu|\rfloor}\\
        &\leq 4 \cdot 2^{\log (\Delta - 1 - d) \frac{\log |\hu|}{\log (\Delta - 1)}} \leq 4 |\hu|^{\frac{\log (\Delta - 1 - d)}{\log (\Delta - 1)}}.
    \end{align*}

    Next, we bound $K_2$.
    Let $L_i$ denote the number of nodes in $\fu$ that are contained in trees hanging from nodes in $\hu$ that are at distance $i$ from $v$.
    Observe that the design of $\fa^*$ ensures that $L_{i + 1} \leq \frac{\Delta - 1 - d}{\Delta - 1} L_i$ for any $i \geq 1$, and $L_1 \leq \frac{\Delta - d}{\Delta} L_0 = \frac{\Delta - d}{\Delta} |\hu|$.
    Therefore, we obtain
    \begin{align*}
        K_2 &\leq \frac{\Delta - d}{\Delta} \cdot \left(\frac{\Delta - 1 - d}{\Delta - 1}\right)^{\lfloor\log_{\Delta - 1} |\hu|\rfloor} \cdot |\hu| \\
        &\leq 2 \cdot \left(\frac{\Delta - 1 - d}{\Delta - 1}\right)^{\log_{\Delta - 1} |\hu|} \cdot |\hu| = 2 (\Delta - 1 - d)^{\log_{\Delta - 1} |\hu|} = 2 |\hu|^{\frac{\log (\Delta - 1 - d)}{\log (\Delta - 1)}}.
    \end{align*}

    Combining the bounds on $K_1$ and $K_2$, we obtain
    \[
        K = K_1 + K_2 \leq 6 |\hu|^{\frac{\log (\Delta - 1 - d)}{\log (\Delta - 1)}}.
    \]
    Observe that Property~\ref{prop:five} in the definition of Algorithm $\fa$ ensures that the number of nodes in $\fu$ is upper bounded by $K$.
    Hence, we obtain 
    \[
        |\fu| \leq 6|\hu|^{\frac{\log (\Delta - 1 - d)}{\log (\Delta - 1)}},
    \]
    as desired.
\end{proof}

\subsection{The Upper Bound}
\label{sec:polyupper}

Now we are set to describe the algorithm $\fp$ for $\pitwo$ that will achieve the desired upper bound.

\paragraph{The algorithm \texorpdfstring{$\fp$}{A-poly}}
Let $x := \frac{\log (\Delta - 1 - d)}{\log (\Delta - 1)}$, and set
\begin{align}
\alpha_i := (2-x)\alpha_{i-1},
\end{align}
for all $2 \leq i \leq k - 1$, and
\begin{align}
\alpha_1 := \frac{1}{\sum_{j=0}^{k-1} (2-x)^j},
\end{align}
as in \Cref{lem:OptValues}.
Moreover, set $\gamma_i := n^{\alpha_i}$, for all $1 \leq i \leq k - 1$.

In $\fp$, each \emph{active} node $v$ executes the generic algorithm for solving $k$-hierarchical $2\frac{1}{2}$-coloring from \Cref{sec:upperbound} on the maximal connected component of active nodes containing $v$ with parameters $\gamma_i$ as specified above.

Each \emph{weight} node $w$ starts by solving the $d$-free weight problem on the subgraph induced by the maximal connected component of weight nodes containing $w$, using Algorithm $\fa$.
For this, each weight node that is adjacent to an active node assumes that it has input $A$, while any other weight node assumes that it has input $W$.
After $3\lceil \log_{d + 1} n \rceil + 3$ rounds, each weight node has finished executing $\fa$.
If a weight node outputs $\con$ or $\dec$ in the execution of $\fa$, then it also outputs the respective label in $\fp$ and terminates.
If a weight node outputs $\cop$ in the execution of $\fa$, then it will also output $\cop$ in $\fp$ but it also has to compute its secondary output, so it will not terminate yet.
Instead it will wait until $3\lceil \log_{d + 1} n \rceil + 3$ rounds have passed (counted from the beginning of the algorithm), and then proceed to the next phase, described in the following.

After round $3\lceil \log_{d + 1} n \rceil + 3$, as soon as an active neighbor of a weight node $w$ with input label $A$ decides on its output label in the $k$-hierarchical $2\frac{1}{2}$-coloring problem, $w$ will flood this output label through the connected component of weight nodes outputting $\cop$ that contains $w$ (breaking ties arbitrarily in case $w$ has two or more active neighbors that decide on their output label simultaneously).
If immediately after round $3\lceil \log_{d + 1} n \rceil + 3$, a weight node $w$ with input label $A$ already knows about a neighbor that has decided on its output label in the $k$-hierarchical $2\frac{1}{2}$-coloring problem, then $w$ will use this output label for flooding through the connected component of weight nodes outputting $\cop$ that contains $w$ (again, breaking ties arbitrarily.)
As soon as a node in the connected component learns this output label $\ell$, it will output $\ell$ as its secondary output (and $\cop$ as its primary output).
This concludes the description of Algorithm $\fp$.

Algorithm $\fp$ is well-defined, by \Cref{obs:conncopy} and the fact that the design of $\fa$ ensures that any two weight nodes with input label $A$ outputting $\cop$ are at distance at least $2\lceil \log_{d + 1} n \rceil + 3$ from each other.

In the following, we will prove the correctness of $\fp$ and bound its node-averaged complexity.

\begin{lemma}
    Algorithm $\fp$ computes a correct output for $\pitwo$.
\end{lemma}
\begin{proof}
To show correctness of $\fp$, it suffices to show that the five properties of a correct output for $\pitwo$ specified in \Cref{def:WeightedColoring} are satisfied.
Property (\ref{prop:twoand}) (with $Z = 2\frac{1}{2}$) follows from \Cref{cor:genericcorrect}.
Properties (\ref{prop:conorcop}), (\ref{prop:con}), and (\ref{prop:allbutd}) follow from Properties (\ref{prop:adjacent}), (\ref{prop:contwo}), and (\ref{prop:dfree}) in the definition of the $d$-free weight problem, respectively. 
Property (\ref{prop:copylab}) follows from the way in which weight nodes that output $\cop$ determine their secondary output in $\fp$ and \Cref{obs:conncopy}.
\end{proof}

\UpperUpperBound*
\iffalse
\begin{lemma}
    The node-averaged complexity of $\fp$ is
    \[
        O(n^{\alpha_1}) = O\left(n^{1/(\sum_{j = 0}^{k - 1}(2 - x)^j)}\right),
    \]
    where $x = \frac{\log (\Delta - 1 - d)}{\log (\Delta - 1)}$.
\end{lemma}
\fi
\begin{proof}
    We start by observing that the weight nodes that output $\cop$ form connected components containing exactly one node $u$ that has at least one active neighbor and that all nodes in such a component output the same secondary label that furthermore is the output label of an active neighbor of $u$.
    (This follows from the design of $\fp$ and \Cref{obs:conncopy}.)
    Let $v$ be an active neighbor of such a node $u$ such that $u$ (and therefore all nodes in the connected component $C(u)$ of $\cop$ nodes containing $u$) ``copied'' the output of $v$ and returned it as secondary output during the execution of $\fp$.
    Then we say that each node in $C(u)$ is \emph{assigned} to $v$ (where we break ties arbitrarily so that each weight node outputting $\cop$ is assigned to exactly one active node).
    For each active node $v$, let $D(v)$ denote the set of all weight nodes assigned to $v$.
    Note that the design of $\fp$ ensures that all nodes in $D(v)$ terminate at most $O(\log n)$ rounds after $v$, by \Cref{cor:dfreerun}.
    Hence, for our calculations we can (and will) assume that all nodes in $D(v)$ terminate at the same time as node $v$ as the $O(\log n)$ additive overhead per node does not exceed the targeted node-averaged complexity.
    %Furthermore, for each active node of level $i$, let $K_i$ denote the number of nodes that 
    Furthermore, observe that the weight nodes that output $\con$ or $\dec$ terminate in $O(\log n)$ rounds (by \Cref{cor:dfreerun}) and can therefore be ignored for the calculations of the node-averaged complexity as, again, their contribution does not exceed the targeted node-averaged complexity. 
    
    We proceed by bounding the sum of the individual termination times of all other nodes.
    As a first step, we compute, for each phase $1 \leq i \leq k$ of the generic algorithm applied on the active nodes (described in \Cref{sec:upperbound}), the number of nodes that have still not terminated at the \emph{start} of the phase.
    By \Cref{lem:remainAfterLvlI}, at the beginning of phase $i$, there are at most still $n \cdot\prod_{j = 1}^{i-1} \frac{1}{\gamma_j}$ active nodes remaining (where we assume an empty product to evaluate to $1$).
    Moreover, the argument used for proving the second part of \Cref{cor:EvenWeightWorst}, combined with \Cref{lem:copynumber}, implies that the number of weight nodes that have not terminated at the beginning of phase $i$ is at most
    \[
        6 \cdot n \prod_{j = 1}^{i-1} \frac{1}{\gamma_j} \cdot \left( n / \left( n \prod_{j = 1}^{i-1} \frac{1}{\gamma_j} \right)\right)^x = 6n\left(\prod_{j = 1}^{i-1} \gamma_j \right)^{x - 1}.
    \]
    Hence, the total number of nodes that have not terminated at the beginning of phase $i$ is at most
    \[
        n \cdot\prod_{j = 1}^{i-1} \frac{1}{\gamma_j} + 6n\left(\prod_{j = 1}^{i-1} \gamma_j \right)^{x - 1} \leq 7n\left(\prod_{j = 1}^{i-1} \gamma_j \right)^{x - 1}.
    \]
    Multiplying with the runtime of the respective phase and summing up over all $k$ phases, we obtain that the aforementioned sum of individual termination times is upper bounded by
    \[
        7n \left(\sum_{i = 1}^{k - 1} \left( \gamma_i \left(\prod_{j = 1}^{i-1} \gamma_j \right)^{x - 1} \right) + \left( \frac{n}{\prod_{j = 1}^{k - 1} \gamma_j} \right) \cdot \left(\prod_{j = 1}^{k-1} \gamma_j \right)^{x - 1} \right).
    \]
    By using the fact that $\gamma_i = n^{\alpha_i}$, for each $1 \leq i \leq k - 1$, our upper bound can written as
    \[
        7n\left( \sum_{i = 1}^{k - 1} \left( n^{\alpha_i + (x-1)\sum_{j = 1}^{i-1} \alpha_j } \right) + n^{1 + (x-2)\sum_{j = 1}^{k-1} \alpha_j} \right).
    \]
    Observe that the exponents of the $k$ summands are precisely those that also appeared in \Cref{sec:weightedlower}.
    In particular, \Cref{lem:OptValues} ensures that all $k$ exponents are equal to $\alpha_1$, enabling us to rewrite our upper bound as
    \[
        7n(k\cdot n^{\alpha_1}).
    \]
    As $k$ is a constant, we obtain that the aforementioned sum of individual termination times is in $O(n \cdot n^{\alpha_1})$, implying a node-averaged complexity of $\fp$ of
    \[
        O(n^{\alpha_1}) = O\left(n^{1/(\sum_{j = 0}^{k - 1}(2 - x)^j)}\right),
    \]
    where $x = \frac{\log (\Delta - 1 - d)}{\log (\Delta - 1)}$, as desired.
    
\end{proof}

\section{Solving Weighted \texorpdfstring{$3\frac{1}{2}$}{3.5}-coloring} \label{sec:UsingOldAlgo}
When trying to solve the weighted versions of $2\frac{1}{2}$- and $3\frac{1}{2}$-coloring, the new challenge is to deal with the weight nodes. This challenge is significantly harder in the lower regime, where we often have only $o(\log^* n)$ node-averaged time to work with. We formalised the problem that the weight nodes want to solve as the $d$-free Weight Problem in \cref{def:WeightProblem}. We restate it here for completeness.

\paragraph{The \texorpdfstring{$d$}{d}-free weight problem}
\WeightProblem*

We are interested in solutions where not too many nodes must output the label $\cop$. The ultimate goal would be to match \cref{lem:weightOfTree} and show that if a component of $w$ weight nodes is attached to an adjacent node $v$, at most $w^x$ many nodes have to output $\cop$, where $x = \frac{\log(\Delta - d -1)}{\log(\Delta - 1)}$.
In the polynomial regime we were able to give an algorithm that achieves such a behaviour, but this algorithm was allowed to spend $O(\log n)$ rounds. In the $\log^*$ regime we can not afford such luxuries. We will instead show that we can achieve something similar, but with a slightly worse efficiency factor of $x' = \frac{\Delta -d + 1}{\Delta -1}$.

\subsection{Using the algorithm from \texorpdfstring{\cite{fullversion}}{On the Node-Averaged Complexity of LCLs on Trees}}
To achieve the goal of an efficiency factor $x' = \frac{\Delta -d + 1}{\Delta -1}$ we use the algorithm of \cite{fullversion} which computes a $(\gamma,\ell, O(\log n))$-decomposition (\cref{def:modi}) with a node-averaged complexity of just $O(\log^*n)$. We will call it the \emph{Fast Decomposition Algorithm}. This decomposition can then be used to obtain a good solution for the weight problem. The full description and analysis of their algorithm is quite lengthy and we encourage the reader to look up the details in the original paper. To keep this work at a reasonable length, we will only restate the relevant details.\\

The layers of such a decomposition implicitly define an ordering and the main point of interest in \cite{fullversion} are local maximums with regards to the partial order of layers defined in \cref{def:ordering}.

We adapt the naming convention and call nodes that already have a layer assigned \emph{assigned nodes} and all other nodes \emph{free nodes}.

\begin{definition}[local maximum \cite{fullversion}]\label{def:locMax}
    A \emph{local maximum} is an assigned node $v \in V(G)$ with the following two properties:
    \begin{enumerate}
        \item Node $v$ and all of its neighbors are assigned.
        \item For each neighbor $w$ of $v$, the layer of $w$ is strictly smaller than the layer of $v$.
    \end{enumerate}
\end{definition}

However the algorithm of \cite{fullversion} requires $\Theta(\log^* n)$ rounds of precomputation, which is not fast enough for us. The bottleneck of their algorithm is the \emph{Compress With Slack} procedure, which requires us to be able to split paths into short subpaths of length in $[\ell, 2\ell]$ for some constant $\ell$. This is required to satisfy the definition of a proper $(\gamma, \ell, L)$-decomposition. We define a relaxed version of a $(\gamma, \ell, L)$-decomposition that does not need to split paths into small subpaths. We emphasise that this is the only difference to the partial $(\gamma, \ell, i)$-decomposition in \cite{fullversion}.

\begin{definition}[relaxed $(\gamma,\ell,i)$-decomposition]
Given three integers $\gamma, \ell, i$, a \emph{$(\gamma,\ell,i)$-decomposition} is a partition of a subset $V' \subset V(G)$ into $2i-1$ layers $V_1^R = (V^R_{1,1},\ldots,V^R_{1,\gamma}), \ldots, V_{i}^R  = (V^R_{i,1},\ldots,V^R_{i,\gamma})$, $V_1^C, \ldots, V_{i-1}^C$ such that the following hold for all layer numbers $1 \leq \ell \leq i$.
    \begin{enumerate}
		\item Compress layers: The connected components of each $G[V_\ell^C]$ are paths of length at least $\ell$, the endpoints have exactly one neighbor in a higher layer, and all other nodes do not have any neighbor in a higher layer, or that is not yet assigned a layer.
		\item Rake layers: The diameter of the connected components in $G[V_\ell^R]$ is $O(\gamma)$, and for each connected component at most one node has a neighbor in a higher layer, or that is not yet assigned a layer.
		\item The connected components of each sublayer $G[V^R_{\ell,j}]$ consist of isolated nodes. Each node in a sublayer $V^R_{\ell,j}$ has at most one neighbor in a higher layer, or that is not yet assigned a layer. 
	\end{enumerate}    
\end{definition}

The following is implicit in Section 5.4 of \cite{fullversion}.
\begin{corollary}[\cite{fullversion}~Section 5.4]
For any constant $\ell$, the \emph{Fast Decomposition Algorithm} computes a relaxed $(\gamma, \ell, O(\log n))$-decomposition with node-averaged complexity $O(1)$ and worst case complexity $O(\log n)$.
\end{corollary}

Similar to the algorithm presented in \cref{ssec:treeDec}, the Fast Decomposition Algorithm consists of iteratively performing Rakes and Compresses. The key idea to obtain a fast node-averaged complexity is to insert additional compress paths, to create more local maximums. As a result, they will mess with the ordering of the layers. To still be able to argue about which nodes were assigned first, they introduce an orientation of the edges in the following way. When a node $v$ is raked, the node orients its unique remaining edge (if it exists) towards itself. When a path is compressed, the first and last $\ell$ edges are oriented inwards. Refer to \cref{fig:orientation} for an illustration. 

\begin{figure}
	\centering
	\includegraphics[width=0.6\textwidth]{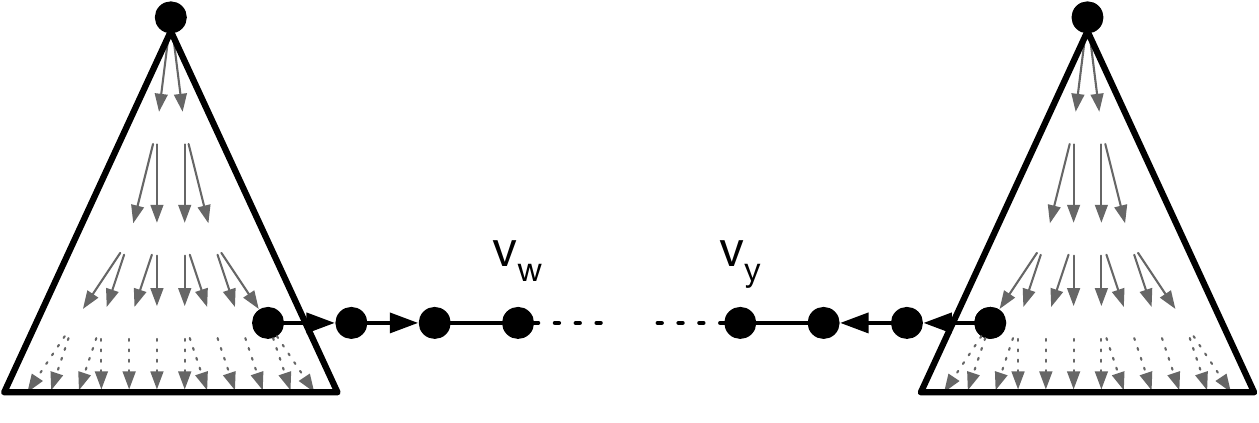}
	\caption{This figure from \cite{fullversion} illustrates the way edges are oriented in the Fast Decomposition Algorithm. The two black nodes at the top are nodes that are not yet assigned a layer. The arrows in the tree show how edges are oriented when such a tree is raked away and the bold path in the middle represents a compress path. In such a compress path, only the first and last $\ell$ edges are oriented.}
	\label{fig:orientation}
\end{figure}

During the analysis the authors mark nodes when they are able to terminate. For any iteration $i$, this is only done in the following two cases (refer to \cite{fullversion}~Lemmas 34-36 for details):
\begin{enumerate}
    \item Whenever a node $v$ becomes a local maximum it becomes \emph{marked}. Furthermore all nodes that can be reached from $v$ through a consistently oriented path also become \emph{marked} in $O(i)$ rounds. 
    \item For all compress layers $V^C_{i-1}$ \footnote{In iteration $i$, compress layer $i-1$ is assigned.} in the relaxed decomposition, all nodes $u$ that are at distance at least $\ell$ from any endpoint of a path become \emph{marked}. Furthermore all nodes that can be reached from $u$ through a consistently oriented path also become \emph{marked} in $O(i)$ rounds. 
\end{enumerate}

They then use these marks to obtain a fast node-averaged complexity with the following lemma.
\begin{lemma}[\cite{fullversion} Lemma 36]
There exists a constant $0 < \sigma < 1$, such that for every iteration $i > 5$ of the Fast Decomposition Algorithm at most $2 \Delta^b n \sigma^{i}$ nodes are not marked.
\end{lemma}

We will not cover more details of the original algorithm, as we believe this would only cause more confusion. The original analysis is quite long and complicated, so instead we will state the following observations from the original paper which we will need to solve the $d$-free weight problem efficiently:
\begin{observation}\label{obs:ObservedProperties}
\noindent
\begin{enumerate}
    \item If $\ell  \geq 2$ any assigned node has at most one incoming edge. This holds because edges are only ever oriented from unassigned nodes to newly assigned nodes. (If $\ell$ is too small, a node in a compress path might have two incoming edges.)
    \item If $\ell>2$ then two local maxima must have distance at least $3$. Consider a path between two local maxima, there must necessarily be a compress path, as otherwise the layers on this path must be strictly decreasing in both directions.
    \item The only case in which the endpoint of a compress path $v$ gets marked, is when there is a consistently oriented path from another node $u$ towards it and $u$ got marked.
    \item Unoriented edges between assigned nodes are only in compress layers, at distance at least $\ell$ from the endpoints.
    \item In iteration $i$ of the algorithm, the connected components induced by consistently oriented edges have diameter at most $O(i)$. This is a direct result of each Rake or Compress only extending such a component by only a constant length.
\end{enumerate}
\end{observation}

We use these observations to find a good solution to the weight problem. To be more explicit, we change the algorithm in the following ways.

\paragraph{Adapted Algorithm:}
Each node $v$ with input label $A$ checks its 5-hop neighborhood. If there is at least one other node with input label $A$, then the unique path between $u$ and $v$ all output $\con$. As a result all remaining nodes with input label $A$ have distance at least $5$. We then proceed to start to run the Fast Decomposition Algorithm with parameter $\ell = 3$ (so all of the observations hold) in the remaining graph (so ignoring all nodes with output $\con$). In each iteration we apply these additional rules whenever they come up:
\begin{enumerate}
    \item If a node $v$ with input label $A$ gets assigned a layer in iteration $i$, it immediately does the following based on the type of layer:
    \begin{itemize}
        \item \textbf{Case 1 Rake Layer:} Then by the properties of a relaxed $(\gamma,\ell,i)$ decomposition, $v$ has at most a single edge towards an active node $u$. If such a $u$ exists, we label it $\dec$ and call it a \emph{border node}. Furthermore $v$ outputs $\cop$ and all nodes that can be reached from $v$ through a consistently oriented path (that did not already output $\dec$) also output $\cop$ in $O(i)$ rounds. 
        \item \textbf{Case 2 Compress Layer:} Then by the properties of a relaxed $(\gamma,\ell,i)$ decomposition, $v$ has at most 2 neighbors $u,w$ that are in the same layer, or that are not yet assigned. Both of these nodes immediately output $\dec$ and and are now \emph{border nodes}. Furthermore $v$ outputs $\cop$ and all nodes that can be reached from $v$ through a consistently oriented path (that did not already output $\dec$) also output $\cop$ in $O(i)$ rounds. If either $u,w$, or both were already an assigned node all nodes that can be reached from them through a consistently oriented path also output $\dec$ in $O(i)$ rounds. 
    \end{itemize}
    \item If a \emph{border node} $v$ becomes assigned a layer all nodes (that did not already output $\dec$) that can be reached from $v$ through a consistently oriented path also output $\dec$ in $O(i)$ rounds. 
    \item If a node $v$ becomes a local maximum, it immediately adapts label $\dec$. Furthermore all nodes (that did not already output $\dec$) that can be reached from $v$ through a consistently oriented path also output $\dec$ in $O(i)$ rounds. 
    \item For all compress layers $V^C_{i-1}$ in the relaxed decomposition, all nodes $u$, that are at distance at least $\ell$ from any endpoint of a path, adapt label $\dec$. Furthermore all nodes (that did not already output $\dec$) that can be reached from $u$ through a consistently oriented path also output $\dec$ in $O(i)$ rounds. 
\end{enumerate}

All nodes that become marked in the original Fast Decomposition Algorithm fix their output label in our adapted version, in fact we are more aggressive with having nodes terminate than the base algorithm. As a consequence we obtain the following result.

\begin{corollary}\label{cor:FewNotDecline}
Let $R(i) = \{v \in V \mid v \text{ did not output }\cop \text{ or }\dec \text{ after iteration }i\}$. Then there exists a constant $0 < \sigma < 1$, such that for every iteration $i > 5$ of the adapted Fast Decomposition Algorithm the number of nodes in $R$ is at most $2 \Delta^b n \sigma^{i}$.
\end{corollary}

To prove that we actually compute a solution to the $d$-free Weight Problem, we have to make sure that any node with output $\cop$ has at most $d$ neighbors with output $\dec$. The following lemma will be enough to prove correctness. Note also that, by ignoring the precomputation, we do not actually use the $\con$ label at all. Furthermore, since the rest of the algorithm runs on the graph without the nodes that did output $\con$ in the beginning, we never actually encounter a node that has output $\con$ for the rest of the algorithm.

\begin{lemma}\label{lem:FewDecline}
Assuming nodes with input label $A$ have distance at least $5$ and $\ell>2$, then for all iterations $i$, any node $v$ that did not output $\dec$ has at most 2 neighbors that do. Furthermore, at most one of them is not a \emph{border node}.
\end{lemma}
\begin{proof}
There are only three different reasons why nodes output $\dec$. As the result of a local maximum, in the middle of compress paths, or because of border nodes. If a node $v$ does not have input label $A$, then it has at most one neighboring border node. This is true because, if there were at least two such border nodes, then the respective nodes with input $A$ would have distance exactly $4<5$ a contradiction.\\ 
If instead $v$ has input label $A$ itself, it can only have border nodes when $v$ itself is assigned a layer. We distinguish whether or not $v$ is a rake or compress node.
\begin{itemize}
    \item If $v$ is a rake node, it has exactly one border node and can not have more because of the distance to other nodes that have input $A$.
    \item If $v$ is in a compress layer, then it has two border nodes. Furthermore if $v$ is an endpoint of the compress layer, the two border nodes are the unique higher layer neighbor and its unique neighbor inside its compress layer. If instead $v$ is not an endpoint of the compress layer, then the two border nodes are its two unique neighbors inside the layer.
\end{itemize}

Next we consider how many neighbors with output $\dec$ a node $v$ has without considering border nodes. We consider the three cases unassigned, assigned to a rake layer, and assigned to a compress layer. Since we already took care of border nodes, we only consider local maximums and the nodes in compress paths.
\begin{itemize}
    \item \textbf{Case 1 $v$ is unassigned:} Since local maximums must be surrounded by assigned nodes, none of the neighbors of $v$ can be local maximums. Furthermore for any already assigned neighbor $u$ of $v$, $v$ must be the unique unassigned neighbor of $u$. As a result the edge between them is oriented from $v$ to $u$, because of \cref{obs:ObservedProperties}. So $u$ must either be a rake node itself, or the endpoint of a compress path, so $u$ did not output $\dec$.
    \item \textbf{Case 2 $v$ is in a rake layer:}  If $v$ is in a rake layer, then the edge to its unique higher layer (or unassigned) neighbor $u$ must be oriented from $u$ to $v$. So if $u$ had output $\dec$, then also $v$ would have output $\dec$. Clearly no other neighbor $w$ of $v$ can be a local maximum and again, cannot have output $\dec$ because of some other local maximum because of \cref{obs:ObservedProperties}. If $u$ was the endpoint of a path, it also would not have output $\dec$ because of \cref{obs:ObservedProperties}.
    \item \textbf{Case 3 $v$ is in a compress layer:} If $v$ is in a compress layer, then it can have at most one neighbor $u$ in the same layer that did output $\dec$. This would be the case only if $u$ is at distance exactly $\ell$ from an endpoint of the path and $v$ is at distance exactly $\ell-1$ from an endpoint of the path. This holds because otherwise they either both output $\dec$ (when they are deep enough in the path), or both do not (because both are not and the edge is oriented from $v$ to $u$). 
    The other possibility is that $v$ is an endpoint of the path and the unique higher layer neighbor $u$ is a local maximum. Then the edge between them must be oriented from $v$ to $u$, since otherwise $v$ would have also output $\dec$. Because $\ell>2$ these two cases cannot happen simultaneously.\\
    Whichever one of the two possibilities for compress nodes is true (either $v$ is the endpoint, or $v$ is exactly at distance $\ell$ from an endpoint), if $v$ also has input $A$, it still will have at most 2 neighbors that do output $\dec$, since there is an overlap.
\end{itemize}
\end{proof}

As a result of \cref{lem:FewDecline}, we immediately get that at the end we will have a valid solution to the $d$-free weight problem, as long as $d\geq2$. Additionally, because of \cref{cor:FewNotDecline}, we get that we will be done after $O(\log n)$ rounds and that the node-averaged complexity is $O(1)$.
\begin{corollary}\label{cor:TerminationFastDecomp}
The adapted Fast Decomposition Algorithm computes a valid solution to the $d$-free weight problem, for $d\geq 2$. Furthermore the node-averaged complexity is $O(1)$ and all nodes have decided on an output after $O(\log n)$ rounds.
\end{corollary}

Assume some node $v$ that has input label $A$ gets assigned a layer in iteration $i$ and as a consequence immediately outputs $\cop$. Then all of the nodes that have not yet chosen an output label and can be reached from $v$ over a consistently oriented path also output $\cop$. Let us denote by $\mathcal{C}(v)$ all of the nodes that do output $\cop$ as a consequence of $v$ being assigned a layer. 
\begin{lemma}\label{lem:CopyComponentsIsolated}
For any node $v$ with input label $A$, the set $\mathcal{C}(v)$ is a tree rooted at $v$ with diameter in $O(i)$, where $i$ is the iteration in which $v$ is assigned a layer.
Assuming nodes with input label $A$ have distance at least $5$, then for any $v$, if $u \notin \mathcal{C}(v)$ is adjacent to $\mathcal{C}(v)$, then $u$ has output $\dec$.
Furthermore for two nodes $u \neq v$, both with input label $A$, $\mathcal{C}(v)$ and $\mathcal{C}(u)$ are disjoint. 
\end{lemma}
\begin{proof}
By \cref{obs:ObservedProperties}, we get the claim on the diameter and that no node has more than one incoming edge. The incoming edge of $v$ is exactly its edge to its unique unassigned neighbor. But that neighbor does output $\dec$. So all other edges of $v$ are oriented away from $v$.
Since also all other nodes in $\mathcal{C}(v)$ have only one incoming edge and only nodes reachable by a consistently oriented path output $\cop$, the subtree $\mathcal{C}(v)$ is rooted at $v$. \\
Consider the orientation of an edge between a node $w \in \mathcal{C}(v)$ and a node $u \notin \mathcal{C}(v)$.
\begin{itemize}
    \item \textbf{Oriented from $w$ to $u$:} either $u$ did already output $\dec$ previously, or $u$ would have to also output $\cop$. Furthermore $u$ cannot have $\cop$ as the result of some other propagation of $\cop$ labels, since its unique incoming edge comes from $w$.
    \item \textbf{Oriented from $u$ to $w$:} $w$ would have two incoming edges, since $w$ must have a consistently oriented path from $v$ to $w$, a contradiction.
    \item \textbf{The edge $\{u,w\}$ is not oriented:} by \cref{obs:ObservedProperties}, $u$ must be a node inside a compress layer that is at distance at least $\ell$ from the endpoints, so $u$ did output $\dec$. If $u$ also has input $A$, then $w$ cant also have input $A$ and so $w$ would be a border node and output $\dec$.
\end{itemize}

To see the furthermore part, assume this is not the case. So take a node $w \in \mathcal{C}(u) \cup \mathcal{C}(v)$. Because of the distance between two nodes that have input $A$, we have that $v \neq w \neq u$. So $w$ can be reached over a consistently oriented path from $u$ and another such path from $v$. But nodes only have one incoming edge, a contradiction.
\end{proof}

As a result, the components that output $\cop$ are completely separated by $\dec$. This will be enough for the rest of our analysis.

\subsection{Generic algorithm for \texorpdfstring{$\Pi^{3.5}_{\Delta, d, k}$}{weighted 3.5-coloring}}
We are ready to give an algorithm for the LCL $\Pi^{3.5}_{\Delta, d, k}$ from \cref{def:WeightedColoring}. So assume $\Delta, d, k$ to be fixed constants satisfying $\Delta \geq d + 3$ and $d \geq 3$, and let $x' = \frac{\log(\Delta -d +1)}{\log(\Delta -1)}$. We will give an algorithm that achieves node-averaged complexity $\Theta((\log^*n)^{\alpha_1(x')})$, where the function $\alpha_1(x') = \frac{1}{1 + (1-x')\sum_{j=0}^{k-1} (2-x')^j}$ is the one from \cref{thm:LowerBoundWeighted3.5}. Note that this is not tight because $x' > x$, however we will see in the proof of \cref{thm:LowDensity} that with some tricks we will get $x' < x + \varepsilon$ for any arbitrarily small constant $\varepsilon$.
Our generic algorithm will be a mixture of the generic algorithm for $k$-hierarchical $3\frac{1}{2}$-coloring from \cref{sec:upperbound} and the algorithm for the $d$-free weight problem that we just discussed. 

\paragraph{Algorithm Description}
Each node has input either $\act$ or $\wei$, let $V^A$ be the set of nodes that have input $\act$ and $V^W$ be the set of nodes that have input $\wei$. We call $V^A$ the set of active nodes and $V^W$ the set of weight nodes, we further note that $V = V^A \cup V^W$. 
Let $\alpha_1, \ldots, \alpha_{k-1}$ be the values obtained in \cref{lem:OptValuesLowerRegime} by using $x' = \frac{\log(\Delta -d +1)}{\log(\Delta -1)}$. We set $\gamma_1 = (\log^*n)^{\alpha_1}, \ldots \gamma_{k-1} = (\log^*n)^{\alpha_{k-1}}$ and run the generic algorithm with these parameters on the graph induced by $V^A$. We note that computing these $\gamma_i$ values is just local computation.

On the other hand, nodes in $V^W$ run the adapted Fast Decomposition Algorithm to compute a solution to the $d$-free Weight Problem. Each node in $V^W$ that is adjacent to an active node $v \in V^A$ gets input label $A$, while all other nodes in $V^W$ adopt input label $W$. The graph induced by $V^W$ is now a valid instance for the $d$-free weight problem (the parameter $d$ is the same as in  $\Pi^{3.5}_{\Delta, d, k}$). When running the adapted Fast Decomposition Algorithm in $V^W$, every node that chooses output label $\dec$ as its solution for the $d$-free weight problem, also immediately outputs $\dec$ as its output for $\Pi^{3.5}_{\Delta, d, k}$.\\
Some nodes will be assigned $\con$ right at the beginning, these nodes also immediately output $\con$ as their output for $\Pi^{3.5}_{\Delta, d, k}$.

Any node $v$ with input label $A$ waits for the adapted fast decomposition algorithm to assign it the label $\cop$, let the iteration in which that happens be iteration $i$. Then $v$ does the following with the set $\mathcal{C}(v)$ from \cref{lem:CopyComponentsIsolated}.
\begin{itemize}
    \item \textbf{Case 1 $v$ has an active neighbor that already terminated with output $o$:} Then all nodes in $\mathcal{C}(v)$ output label $\cop$ with secondary output $o$ in $O(i)$ rounds.
    \item \textbf{Case 2 all of $v$'s active neighbors still did not terminate yet:} Then it first collects the entire topology of the rooted subtree induced by $\mathcal{C}(v)$ in $O(i)$ time. By \cref{lem:FewDecline} any node in $\mathcal{C}(v)$ has at most 2 neighbors with output label $\dec$. $v$ changes some labels in $\mathcal{C}(v)$ to $\dec$ in a way that minimises the number of nodes that have label $\cop$, while still maintaining a valid solution to the $d$-free weight problem. We note that we do not change nodes that previously had output $\dec$. Then, all nodes that are newly assigned $\dec$ immediately output $\dec$ as their output for $\Pi^{3.5}_{\Delta, d, k}$.
\end{itemize}
When an active node $v \in V^A$ terminates with output $o$ and an adjacent active node $u \in V^W$ already was assigned a set $\mathcal{C}(u)$ in iteration $i$, in $O(i)$ rounds all nodes in $\mathcal{C}(u)$, that still have label $\cop$ as their solution for the $d$-free weight problem, output $\cop$ with secondary output $o$ as their output for $\Pi^{3.5}_{\Delta, d, k}$.\\
This finishes the description of the algorithm.
We will first argue that the algorithm is correct.

\begin{lemma}
The algorithm computes a correct solution to $\Pi^{3.5}_{\Delta, d, k}$.
\end{lemma}
\begin{proof}
Because of the correctness of the generic algorithm from \cref{sec:upperbound} the nodes in $V^A$ output a correct solution to $k$-hierarchical $3\frac{1}{2}$-coloring. So we only need to show that the nodes in $V^W$ output correct solutions. \\
First notice that all of the nodes that output $\con$ at the beginning of the fast decomposition algorithm are paths between active nodes. So they satisfy the constraints. All other nodes with input $A$ do output $\cop$ so we also get that all nodes $u$ adjacent to an active node $v \in V^A$ output either $\cop$ or $\con$. \\
So the last thing we need to show is that for each node $u \in V^W$ that does output $\cop$ at most $d$ neighbors output $\dec$. The only nodes that output $\cop$ are ones that were in a set $\mathcal{C}(v)$ for some node $v$ adjacent to an active node. By \cref{lem:FewDecline} we get that when $v$ is first assigned a layer in some iteration $i$ each node in $\mathcal{C}(v)$ has at most 2 neighbors that output $\dec$ . Note that these nodes did output $\dec$ for both the $d$-free weight problem and $\Pi^{3.5}_{\Delta, d, k}$. This only changes if $v$ computes a reassignment of output labels and changes some nodes from $\cop$ to $\dec$. However when $v$ does so, it knows the entire topology of $\mathcal{C}(v)$, and does so such that these labels are a valid solution to the $d$-free weight problem. So as a result all nodes that still output $\cop$ have at most $d$ neighbors that do output $\dec$, satisfying the constraints of $\Pi^{3.5}_{\Delta, d, k}$.
\end{proof}

Next we consider some node $v \in V^W$ that gets input label $A$ for the $d$-free weight problem. Assume it gets assigned a label by the adapted Fast Decomposition Algorithm in iteration $i$ and the active node adjacent to $v$ has not yet terminated.

\begin{lemma}\label{lem:BoundOnCopyTrees}
If a node $v \in V^W$ with input label $A$ for the $d$-free weight problem reassigns some labels in $C(v)$, the number of nodes that still have label $\cop$ is at most $2|C(v)|^{x'}$, where $x' = \frac{\log(\Delta -d +1)}{\log(\Delta -1)}$.
\end{lemma}
\begin{proof}
By \cref{lem:CopyComponentsIsolated} the set $C(v)$ is a tree rooted at $v$ with diameter at most $O(i)$. By \cref{lem:FewDecline} any node in $C(v)$ initially has at most 2 neighbors that output $\dec$. Since the maximum degree is $\Delta$, the fan-out of the tree $\mathcal{C}(v)$ is at most $\Delta -1$, and as a result the diameter is at least $\log_{\Delta -1}(|\mathcal{C}(v)|)$. 
We start at $v$ and for every one of the layers  $0 \leq j \leq \log_{\Delta -1}(|\mathcal{C}(v)|)$ (layer 0 is $v$) of the rooted tree $\mathcal{C}(v)$ we do the following. Each node $u$ in layer $j$ has at most $2$ neighbors that already have output $\dec$, so it can assign at least $d-2$ of its at most $\Delta -1$ children the output $\dec$. 
We always assign $\dec$ to the children with the heaviest subtrees and then also change all of the $\cop$ outputs in that subtree to $\dec$. As a result with every layer we lose at least $\frac{d-2}{\Delta -1}$ nodes. Let $\mathcal{C'}(v)\subset \mathcal{C}(v)$ be all nodes that still have output $\cop$ after the reassignment. In the first $\log_{\Delta -1}(|\mathcal{C}(v)|)$ layers, the tree $\mathcal{C'}(v)$ now has fan-out $\Delta - d - 1$, so the number of nodes in these first layers is at most
\[
(\Delta - d - 1)^{\log_{\Delta -1}(|\mathcal{C}(v)|)} = (\Delta -1)^{\log_{\Delta -1}(\Delta -d -1) \cdot \log_{\Delta -1}(|\mathcal{C}(v)|)}  = |\mathcal{C}(v)|^{\log_{\Delta -1}(\Delta -d -1)} = |\mathcal{C}(v)|^x.
\]
Since we loose a $\frac{d-2}{\Delta - 1}$ fraction of the nodes every layer until $\log_{\Delta -1}(|\mathcal{C}(v)|)$, the number of remaining nodes after that layer is upper bounded by 
\begin{align*}
|\mathcal{C}(v)| \cdot \left(\frac{\Delta - d + 1}{\Delta - 1}\right)^{\log_{\Delta -1}(|\mathcal{C}(v)|)} &= |\mathcal{C}(v)| \cdot \frac{(\Delta - d + 1)^{\log_{\Delta -1}|\mathcal{C}(v)|)}}{|\mathcal{C}(v)|} \\
    &= (\Delta -1)^{\log_{\Delta -1}(\Delta - d +1) \cdot \log_{\Delta -1}|\mathcal{C}(v)|)} \\
    &= |\mathcal{C}(v)|^{\log_{\Delta -1}(\Delta - d +1)} = |\mathcal{C}(v)|^{x'}.
\end{align*}
Combining the two, the size of $\mathcal{C'}(v)$ is upper bounded by 
\[
|\mathcal{C'}(v)| \leq  |\mathcal{C}(v)|^x + |\mathcal{C}(v)|^{x'} \leq 2 |\mathcal{C}(v)|^{x'}.
\]
\end{proof}

Next we analyse the guarantees that we get from the generic algorithm for $k$-hierarchical $3\frac{1}{2}$-coloring. By the description and \cref{lem:remainAfterLvlI} we get the following results.
\begin{corollary}\label{cor:GenericGuarrantees}
Phase $i<k$ of the generic algorithm that runs on $V^A$ takes $O((\log^*n)^{\alpha_i})$ rounds and after it all nodes of level $i$ have terminated. Furthermore the number of remaining nodes in $V^A$ is upper bounded by 
\[
O \left( \frac{n}{\prod_{j\leq i} (\log^*n)^{\alpha_j}} \right).
\]
\end{corollary}

Let $L_i \subset V^A$ be all nodes that terminate in phase $1\leq i \leq k$ of the generic algorithm. Let $A_i$ be all weight nodes adjacent to a node in $L_i$ that did already get assigned a layer by the adapted Fast Decomposition Algorithm at the end of phase $i$ of the generic algorithm.
Then we also define $W_i = L_i \bigcup_{v \in A_i} \mathcal{C'}(v)$ as the set that includes all weight nodes that were in $\mathcal{C'}(v)$ components for any $v$ adjacent to a node in $L_i$ (We ignore nodes adjacent nodes that don't have a $\mathcal{C'}(v)$ yet, as we will deal with them later.) Note that we refer to the sets $\mathcal{C'}(v)$ from \cref{lem:BoundOnCopyTrees}.

\begin{lemma}\label{lem:totalTimeBeforeK}
For all $i<k$, the total time spent by nodes in $W_i$ is upper bounded by 
\[
O\left(n^{1} \cdot \left(\log^*(n)^{\alpha_i + (x'-1)\sum_{j < i}\alpha_j}\right) \right).
\]
\end{lemma}
\begin{proof}
By \cref{cor:GenericGuarrantees}, we get 
\[
|L_i| \in O \left( \frac{n}{\prod_{j < i} (\log^*n)^{\alpha_j}} \right) = O \left( \frac{n}{(\log^*n)^{\sum_{j < i}\alpha_j}} \right).
\]
Because of \cref{lem:BoundOnCopyTrees}, the size of $W_i$ is upper bounded by
\begin{align*}
|W_i| &= |L_i| + \sum_{v \in A_i} |\mathcal{C'}(v)| \leq |L_i| + \sum_{v \in A_i} 2|\mathcal{C}(v)|^{x'} \leq 2|L_i| \left( \frac{\sum_{v \in A_i}|\mathcal{C}(v)|}{|A_i|} \right)^{x'} \\
&\leq 2|L_i| \left( \frac{\sum_{v \in A_i}|\mathcal{C}(v)|}{\Delta|L_i|} \right)^{x'} \leq \frac{2}{\Delta^{x'}} |L_i| \left( \frac{n}{|L_i|} \right)^{x'} \in O\left(  n^{x'} \cdot \left(\frac{n}{(\log^*n)^{\sum_{j< i}\alpha_j}}\right)^{1-x'} \right)\\
&= O\left(  n^{1} \cdot \left((\log^*n)^{-\sum_{j < i}\alpha_j}\right)^{1-x'} \right) = O\left(n^{1} \cdot \left((\log^*n)^{(x'-1)\sum_{j < i}\alpha_j}\right) \right).
\end{align*}
All of the nodes in $L_i$ terminate after $t \in O(\gamma_1 + \gamma_2 + \ldots + \gamma_i) = O(\gamma_i)$ rounds, where the equality comes from the fact that $\alpha_1 \leq \alpha_2 \leq \ldots \leq \alpha_{k-1}$. The diameter of the $\mathcal{C'}(v)$ components can clearly also be at most $t$ (since only $t$ rounds happened so far). So after at most an additional $t$ rounds all nodes in $|W_i|$ have terminated.  As a result, the total time spent by nodes in $W_i$ is upper bounded by 
\[
O\left(n^{1} \cdot \left((\log^*n)^{(x'-1)\sum_{j < i}\alpha_j}\right) \right) \cdot O(\gamma_i) = O\left(n^{1} \cdot \left((\log^*n)^{\alpha_i + (x'-1)\sum_{j < i}\alpha_j}\right) \right).
\]
\end{proof}

For $i = k$ the calculation is a bit different 
\begin{lemma}\label{lem:TotalTimeInK}
The total time spent by nodes in $W_k$ is upper bounded by $O\left(n^{1} \cdot \left((\log^*n)^{1 + (x'-1)\sum_{j < k}\alpha_j}\right) \right)$.
\end{lemma}
\begin{proof}
In the same way as before we get 
\begin{align*}
|W_k| = O\left(n^{1} \cdot \left((\log^*n)^{(x'-1)\sum_{j\leq k}\alpha_j}\right) \right).
\end{align*}
All of the nodes in $L_k$ terminate after $t \in O(\log^*n)$ rounds. The diameter of the $\mathcal{C'}(v)$ components can clearly also be at most $t$ (since only $t$ rounds happened so far). So after at most an additional $t$ rounds all nodes in $|W_k|$ have terminated. As a result, the total time spent by nodes in $W_k$ is upper bounded by 
\[
O\left(n^{1} \cdot \left((\log^*n)^{(x'-1)\sum_{j < i}\alpha_j}\right) \right) \cdot O(\log^*n) = O\left(n^{1} \cdot \left((\log^*n)^{1 + (x'-1)\sum_{j < k}\alpha_j}\right) \right).
\]
\end{proof}

Now we take care of the nodes in $V^W$ which are not yet accounted for, that is, all those that are immediately assigned $\dec$, or those that when reaching an adjacent node $v$ with input $A$ can all immediately terminate. Remember that the second case happens only when a neighbor $u\in V^A$ of $v$ had already terminated with some output $o$, when $v$ was assigned.

\begin{lemma}\label{lem:ConstantTimeInW}
All nodes in $W := V^W \setminus \left(\bigcup_{1\leq i \leq k} W_i\right)$ terminate in a constant number of rounds \emph{on average}.
\end{lemma}
\begin{proof}
This is an immediate consequence of \cref{cor:FewNotDecline}. Let 
\[
R(i) = \{v \in V^W \mid v \text{ did not output }\cop \text{ or }\dec \text{ after iteration }i\}
\]
be the set as in \cref{cor:FewNotDecline}. We argue that all of these nodes have already terminated, or are in one of the sets $W_1,\ldots, W_k$. All nodes that get assigned output label $\dec$ for the $d$-free weight problem all immediately terminate with output $\dec$ for $\Pi^{3.5}_{\Delta, d, k}$. If instead the nodes did output $C$, then they are in one of the $\mathcal{C}(v)$ for some $v$. If this $v$ has a neighbor that did already terminate, then also all of the nodes in $\mathcal{C}(v)$ can terminate $O(i)$ iterations later. If they don't have such a neighbor, then they are necessarily part of one of the $W_i$.\\
So the remaining nodes in $W$ after iteration $i$ are exactly the set $R(i)$. By \cref{cor:FewNotDecline}, for constant $b \in O(1), 0 < \sigma< 1$, the following holds after iteration $i\geq 5$:
\[
|R(i)| \leq 2 \Delta^b n \sigma^{i}.
\]
Since by \cref{cor:TerminationFastDecomp} the algorithm terminates after $c \log n \in O(\log n)$ rounds for some constant $c$, we get for the node-averaged complexity of the nodes in $W$,
\begin{align*}
\Bar{T} \leq \frac{1}{n} \cdot \left( 5n +  \sum_{5 < i \leq c \log(n)} |R(i)| \right) \leq 5 + \frac{1}{n} \cdot  \sum_{5 < i \leq c \log(n)}  2 \Delta^b n \sigma^{i} = 5 + 2\Delta^b \sum_{5 < i \leq c \log(n)} \sigma^i \in O(1).
\end{align*}
\end{proof}

Now we can bound the total node-averaged complexity by differentiating between the nodes in $W$ and nodes in $W_1, \ldots, W_k$. So we finally give a proof for \cref{thm:LowerUpperBound}.

\LowerUpperBound*
\begin{proof}
We can write down the node-averaged runtime using $W$ and $W_1, \ldots, W_k$. We denote the termination time of a node $v$ by $T_v$.
\[
\Bar{T} = \frac{1}{n} \cdot \sum_{v\in G} T_v = \frac{1}{n} \cdot \left( \sum_{v\in W} T_v  + \sum_{1 \leq i \leq k} \sum_{v \in W_i} T_v\right)
\]
By \cref{lem:ConstantTimeInW} we get that the contribution of the nodes in $W$ is a constant. We can use \cref{lem:TotalTimeInK} to bound the total time of nodes in $W_k$ and \cref{lem:totalTimeBeforeK} to bound the total time for nodes in $W_i$ for $i<k$. We get:
\begin{align*}
\Bar{T} &=O(1) + \frac{1}{n} \cdot \left(\sum_{1 \leq i < k} O\left(n^{1} \cdot \left((\log^*n)^{\alpha_i + (x'-1)\sum_{j < i}\alpha_j}\right) \right) +  O\left(n^{1} \cdot \left((\log^*n)^{1 + (x'-1)\sum_{j < k}\alpha_j}\right) \right)\right)\\
&=O\left(\sum_{1 \leq i < k} \left((\log^*n)^{\alpha_i + (x'-1)\sum_{j < i}\alpha_j}\right) +  \left((\log^*n)^{1 + (x'-1)\sum_{j < k}\alpha_j}\right) \right).
\end{align*}
These exponents are exactly the terms as in \cref{lem:OptValuesLowerRegime}, and we choose $\alpha_1, \ldots, \alpha_{k-1}$ as the optimal solutions given there. Furthermore, all of these exponents equal each other and are $\alpha_1(x')$. Hence, we get the following:
\begin{align*}
\Bar{T} &= O\left(\sum_{1 \leq i < k} \left((\log^*n)^{\alpha_i + (x'-1)\sum_{j < i}\alpha_j}\right) +  \left((\log^*n)^{1 + (x'-1)\sum_{j < k}\alpha_j}\right) \right)\\
&= O\left(\sum_{1 \leq i < k} \left((\log^*n)^{\alpha_1(x'}\right) +  \left((\log^*n)^{\alpha_1(x')}\right) \right) = O\left(k \cdot \left((\log^*n)^{\alpha_1(x'}\right) \right)\\
&= O\left(\left((\log^*n)^{\alpha_1(x'}\right) \right),
\end{align*}
as claimed.
\end{proof}

\section{Density Results}\label{sec:density}
We will now show that the node-averaged complexity landscape of LCLs on bounded degree trees is dense in the region $n^{\Omega(1)}$--$o(n^{1/2})$. To do this, we will use \cref{thm:UpperLowerBound} and \cref{thm:UpperUpperBound}, that we restate here for ease of reading.
\UpperLowerBound*
\UpperUpperBound*

We first need to show that $\alpha_1(x)$ is well behaved.
\begin{lemma} \label{lem:niceFunction}
$\alpha_1(x) = \frac{1}{\sum_{j=0}^{k-1} (2-x)^j}$ is a continuous monotonically increasing function over the interval $[0,1] \subset \mathbb{R}$.
\end{lemma}
\begin{proof}
Notice that $P(x) := \sum_{j=0}^{k-1} (2-x)^j$ is just a polynomial with real coefficients and therefore continuous. Clearly $\frac{\partial P}{\partial x}(x) =\sum_{j=1}^{k-1} -j(2-x)^{j-1}$, which, for $x \in [0,1]$, is always a negative number. Therefore, $P(x)$ is strictly decreasing on $[0,1]$, and as a result $\alpha_1(x) = \frac{1}{P(x)}$ is strictly increasing and continuous on $[0,1]$, as long as we do not divide by $0$. But since $\alpha_1(0) = \frac{1}{2
^{k} -1}$, and as just discussed $\alpha_1(x)$ is strictly increasing on $[0,1]$, this also means that $\alpha_1(x)$ is well behaved on our interval.
\end{proof}

\begin{lemma}\label{lem:density}
For any integer $k$ and any two real numbers $\frac{1}{2^k -1}\leq r_1< r_2 < \frac{1}{k}$, there exist constants $\Delta, d, c$ such that $r_1 \leq c \leq r_2$ and $\Pi^{2.5}_{\Delta,d,k}$ has node-averaged complexity $n^c$.
\end{lemma}
\begin{proof}
The complexities with $\frac{1}{2^k-1}$ are simply given by the original $k$-hierarchical $2\frac{1}{2}$-coloring~\cite{Balliu0KOS23}. So in the following assume $\frac{1}{2^k -1}< r_1$.
Notice that the complexity induced by \cref{thm:UpperLowerBound,thm:UpperUpperBound} is exactly
\[n^{\alpha_1(x)},\] so if we find some constants $\Delta, d$, such that 
\[
\alpha_1(x) = \frac{1}{\sum_{j=0}^{k-1} (2-x)^j} =  c,
\]
we are done.
Since according to \Cref{lem:niceFunction} we have that $\alpha_1(x)$ is continuous and strictly increasing, we can obtain $x_1 = \alpha_1^{-1}(r_1)$ and $x_2 = \alpha_1^{-1}(r_2)$. Furthermore, for any $ x_1 \leq a \leq x_2$, it holds that $r_1 = \alpha_1(x_1) \leq \alpha_1(a) \leq \alpha_1(x_2) = r_2$, again because $\alpha_1$ is a continuous strictly increasing function. We also note that since $\alpha_1(0) = \frac{1}{2^k -1} $ and $\alpha_1(1) = \frac{1}{k}$, we get that $0< x_1<x_2<1$.\\
Now the only thing left to do is to choose $a$ appropriately. Notice that we may only choose a value from among the set
\[
X:= \{\log(\Delta d -1)/\log(\Delta-1) \; |\; s,\Delta \in \mathbb{N} \text{ s.t. } \Delta\geq d+3 \}.
\]
However, given any rational number $p/q$ with $p<q$ we can choose $\Delta := 2^q +1$ and then solve for $d$ in 
\[
\Delta-d-1 = 2^p \Rightarrow 2^q +1 -d -1 = 2^p \Rightarrow d= 2^q -2^p >0
\]
and therefore $p/q \in X$. Now since between any two real numbers $0<x_1 < x_2<1$ there exists a rational number $0<a<1 \in \mathbb{Q}$ satisfying $r_1 \leq a \leq r_2$ and as just discussed also $a \in X$, we simply choose $c = \alpha_1(a)$ and get that $r_1 \leq c \leq r_2$.
\end{proof}

Observe that \Cref{lem:density} does not actually allow us to reach a node-averaged complexity of $\Theta(\sqrt{n})$. We will take care of such a complexity in \Cref{lem:TightNodeAveraged} of \Cref{sec:efficientWeight}. By combining \Cref{lem:TightNodeAveraged} and \cref{lem:density}, we obtain the following.

\UpperDensity*
\begin{proof}
If we want an LCL with complexity $\Theta(n^{1/k})$, then we invoke \cref{lem:TightNodeAveraged}. Otherwise, since there exists $k$ such that $r_1 < \frac{1}{2^k-1} < r_2$, we just invoke \cref{lem:density} with $r_1' = \frac{1}{2^k-1}, r_2$ and get the desired result.
\end{proof}

We conclude the discussion about the polynomial regime by showing that there is a gap between $\Theta(n)$ and $\Theta(\sqrt{n})$ in the node-averaged complexity landscape. This is a simple consequence of a result by Feuilloley.
\begin{lemma}[\cite{Feuilloley17}]\label{lem:neighborsrun}
    Let $A$ be an algorithm that solves an LCL problem $\Pi$ with node-averaged complexity $T$. Then, if a node $v$ runs for $t$ rounds, either there exist at least $t/2 - 1$ nodes in the $t/2$-radius neighborhood of $v$ that run for at least $t/2$ rounds, or $v$ could terminate earlier.
\end{lemma}

When we apply this lemma to a node that runs for linear in $n$ many rounds we immediately get that we spend too much total time to have a fast node-averaged complexity.

\begin{corollary}\label{cor:linearGap}
    Any LCL with worst case complexity $\Omega(n)$ has node-averaged-complexity $\Omega(n)$.
\end{corollary}
\begin{proof}
    Since the LCL has worst case complexity $\Omega(n)$, there must exist an instance $I$ and a node $v$ in this instance, such that $v$ runs for at least $cn$ rounds for some constant $c$. By \cref{lem:neighborsrun} this means that there are at least $\frac{cn}{2}$ nodes that run for at least $\frac{cn}{2}$ rounds which already results in a node-averaged complexity of at least
    \[
    \frac{1}{n} \cdot \frac{cn}{2} \cdot \frac{cn}{2} = c' n,
    \]
    for an adequate constant $c'$.
\end{proof}

Next we will show that the node-averaged complexity landscape of LCLs on bounded degree trees is dense in the region $(\log^*n)^{\Omega(1)}$--$o(\log^*n)$. Things will be a bit more tricky, but \cref{thm:LowerBoundWeighted3.5} together with \cref{thm:LowerUpperBound}, which we restate here for ease of reading, will be enough.

\LowerLowerBound*

\LowerUpperBound*

Again we start by showing that $\alpha_1(x)$ is well behaved.
\begin{lemma} \label{lem:LowerNiceFunction}
$\alpha_1(x) = \frac{1}{1 + (1-x)\sum_{j=0}^{k-2} (2-x)^j}$ is a continuous monotonically increasing function over the interval $[0,1] \subset \mathbb{R}$.
\end{lemma}
\begin{proof}
Notice that $P(x) := 1 + (1-x)\sum_{j=0}^{k-2} (2-x)^j$ is just a polynomial with real coefficients and therefore continuous. By the product rule,
\[
\frac{\partial P}{\partial x}(x) =(1-x) \sum_{j=0}^{k-2} -j(2-x)^{j-1} - \sum_{j=0}^{k-2} (2-x)^j,
\]
which for $x \in [0,1]$ is always a negative number. Therefore $P(x)$ is strictly decreasing on $[0,1]$ and as a result $\alpha_1(x) = \frac{1}{P(x)}$ is strictly increasing and continuous on $[0,1]$ as long as we do not divide by 0. But since $\alpha_1(0) = \frac{1}{2
^{k-1}}$ and as just discussed $\alpha_1(x)$ is strictly increasing on $[0,1]$, this also means that $\alpha_1(x)$ is well behaved on our interval.
\end{proof}

Since there is a gap between the results in \cref{thm:LowerBoundWeighted3.5,thm:LowerUpperBound}, we have to make sure that this gap gets arbitrarily small.

\begin{lemma}\label{lem:GapIsSmall}
For all $\varepsilon>0$ and any $x = \frac{a}{b}\in \mathbb{Q}\cap [0,1]$, there exists $\Delta,d$ such that $x = \frac{\log(\Delta - d- 1)}{\log(\Delta-1)}$ and $x' = \frac{\log(\Delta - d + 1)}{\log(\Delta-1)}$ and $|x -x'| < \varepsilon$.
\end{lemma}
\begin{proof}
Fix $\varepsilon,a,b$, we will show that proper $\Delta,d$ exist.
It suffices to show that, for any constant $c$, we can pick $\Delta,d$ such that $\Delta - d- 1 = 2^{ca}$ and $\Delta -1 = 2^{cb}$, because then $x = \frac{ca}{cb} = \frac{a}{b}$ and 
\begin{align*}
x' &= \frac{\log(\Delta -d +1)}{\Delta -1} = \frac{\log(\Delta -d - 1 +2)}{cb} = \frac{\log(2^{ac} +2)}{cb} \leq \frac{1}{cb} \cdot (\log(2^{ac}) + \frac{2}{2^{ac}} \\
&= x + \frac{2}{2^{ac}} < x + \varepsilon,
\end{align*}
where we used that the Taylor Expansion of $\log(z)$ at $z=2^{ac}$ is a proper upper bound, because the logarithm is a concave function and the fact that we can choose $c$ such that $\frac{2}{2^{ac}} < \varepsilon$ for any $\varepsilon>0$.\\
Now, to prove that for any constant $c$ we can pick $\Delta,d$ such that $\Delta - d- 1 = 2^{ca}$ and $\Delta -1 = 2^{cb}$, we first set 
\[
\Delta -1 = 2^{cb} \Rightarrow \Delta = 2^{cb} +1,
\]
and then we use that in 
\[
\Delta -d - 1 = 2^{ca} \Rightarrow \Delta - d - 1 = 2^{cb} +1 -d -1 = 2^{cb} -d = 2^{ca} \Rightarrow d = 2^{cb} - 2^{ca} >0.
\]
\end{proof}

Since we can now make the gap between the upper and the lower bound arbitrarily small, the rest follows in the same way as it did in the polynomial regime. 

\LowDensity*
\begin{proof}
We first start by fixing $k$ such that $\frac{1}{2^{k-1}}\leq r_1 < 1/k$, which is possible, since $r_1>0$. As a result of this choice of $k$ the function $\alpha_1(x)$ actually outputs values from $\frac{1}{2^{k-1}}$ to $\frac{1}{k}$. Now since by \cref{lem:LowerNiceFunction} the function $\alpha_1(x) = \frac{1}{1 + (1-x)\sum_{j=0}^{k-2} (2-x)^j}$ is continuous and strictly increasing on the real interval $[0,1]$,
we obtain $x_1 = \alpha^{-1}_1(r_1), x_2= \alpha^{-1}_1(r_2)$, and we choose $\frac{a}{b}\in \mathbb{Q}$, such that $0<x_1<\frac{a}{b}<x_2<1$. As a result it holds that $r_1 = \alpha_1(x_1) < \alpha_1(\frac{a}{b}) < \alpha_1(x_2) = r_2$, again because $\alpha_1$ is a continuous strictly increasing function.
Choose $\varepsilon' := \min\{\varepsilon, (r_2 - \alpha_1(\frac{a}{b}))/2\}$ to ensure that everything fits in the gap between $\alpha_1(\frac{a}{b})$ and $r_2$.
Now since $\alpha_1(x)$ is continuous, for every $\varepsilon' > 0$ there exists some $\delta >0$ such that for all values $\tau \in (\alpha_1(\frac{a}{b}) - \delta, \alpha_1(\frac{a}{b}) + \delta)$ it holds that $\alpha_1(\frac{a}{b}) - \varepsilon' < \alpha_1(\tau)< \alpha_1(\frac{a}{b}) + \varepsilon'$. So now given $\varepsilon'$ we get a $\delta$ that we can use in \cref{lem:GapIsSmall} and obtain values $\Delta, d$ such that $x = \frac{a}{b}$ and $|x-x'| < \delta$. As a result $|\alpha_1(x) - \alpha_1(x')| < \varepsilon'$ because of our choice of $\delta$.\\
Now by \cref{thm:LowerBoundWeighted3.5,thm:LowerUpperBound} the complexity of $\Pi^{3.5}_{\Delta,d,k}$ is between $\Omega((\log^*n)^{\alpha_1(x)})$ and $O((\log^*n)^{\alpha_1(x')}) \subset O((\log^*n)^{\alpha_1(x)+\varepsilon'})$.
\end{proof}

\section{More efficient weight}\label{sec:efficientWeight}
Observe that \Cref{lem:density} allows us to only get arbitrarily close to the worst case complexity, but it does not give us LCLs that have worst-case and node-averaged complexities that are the same. For example, choose $k=2$. Then, the worst case complexity is $\Theta(\sqrt{n})$, but \Cref{lem:density} only gives us LCLs with node-averaged complexity $o(\sqrt{n})$. 
As it turns out, if we design our LCLs differently, we can also obtain a problem with node-averaged complexity $\Theta(\sqrt{n})$. According to the analysis of the previous sections, this happens exactly if the efficiency factor $x$ is $1$. However, with the way we defined our weight augmented LCLs, we can only push $x$ arbitrarily close to $1$. What we really need is a nice twist on the rules that weight nodes need to follow, that allows us to reach $x=1$. However, at the same time, we have to ensure that our worst case complexity does not become $\Theta(n)$. The way the LCLs $\Pi^{Z}_{\Delta,d,k}$ are defined results in a complexity $O(\log n)$ for the problem on weight nodes. To solve this issue and achieve $\Theta(\sqrt{n})$ node-averaged complexity, we will require weight nodes to solve a problem with worst case complexity $\Theta(\sqrt{n})$. This way we keep the same worst case complexity, but as we will soon see, make much more efficient use of weight nodes.

To get the desired worst case runtime we must create a problem that only allows $k$ compresses. The problem we use is essentially the problem of computing a $(\gamma, \ell, L)$-decomposition (\cref{def:modi}, but modified to be an LCL). Instead of outputting the exact layer number in the decomposition, we require nodes to only choose an output label that represents how many compresses have been used so far. To enforce an actual decomposition, we require the nodes to output the orientation implied by a rake and compress decomposition.
\begin{definition}[$k$-hierarchical labeling]
For any integer $k$, the $k$-hierarchical labeling problem consists of the input set $\Sigma_{\mathrm{in}}=\emptyset$ and the output set $\Sigma_{\mathrm{out}} = \{R_0, R_1, \dots, R_k, C_1, C_2, \dots, C_k\}$. We call the labels $R_1, \dots, R_{k}$ \emph{rake labels} and $C_1,\dots,C_{k-1}$ compress labels. Furthermore, there is an ordering of the labels $R_1 < C_1 <R_2<C_2<R_3 <\dots<C_{k-1}< R_k$. Any legal labeling must satisfy the following rules:
\begin{enumerate}
    \item All edges adjacent to a rake label must be oriented. \label{rule2:orientRakes}
    \item Each node $v$ has at most one edge $e=(v,u)$ oriented away from itself, except for compress nodes that have two compress neighbors, who must not have any outgoing edge. \label{rule2:oneOutgoing}
    \item For all oriented edges $(u,v)$ the label of $v$ is larger than or equal to the label of $u$. \label{rule2:orderedOrientation}
    \item For all compress labels, the subgraph induced by the nodes of that label consists only of disjoint paths. \label{rule2:compressPaths}
    \item Two nodes that have a different compress label must not be adjacent.\label{rule2:pathsDisjoint}
    \item Any node $v$ with a rake label has at most one compress label neighbor pointing towards it. Furthermore if there is such a neighbor, then all neighbors that point towards $v$ have a strictly lower label.\label{rule2:oneCompressNeighbor}
\end{enumerate}
\end{definition}

Some intuition on this rules.
\begin{itemize}
    \item \cref{rule2:orientRakes,rule2:oneOutgoing} force us to consistently orient all of the rake label components towards some component specific root node.
    \item \cref{rule2:compressPaths,rule2:pathsDisjoint} force us to use compress labels only in paths and to ensure that two different compress paths are separated by at least one rake label.
    \item \cref{rule2:orderedOrientation} ensures that we have to keep track of how many compresses we have used so far. Basically the way we will be handling a long path $P$ is, to have every node but the endpoints take compress label $C_i$ and have both endpoints pick rake label $R_i$. Then we orient the edges connecting the endpoints towards the endpoints. Clearly this works only if we have used only $R_1, \dots, R_{i-1}$ so far.
    \item \cref{rule2:oneCompressNeighbor} will help us argue about how many nodes have to wait when we use this formulation to augment other LCLs.
\end{itemize}

Clearly all of the rules can be checked by just looking at the immediate neighbors of a node.

\begin{corollary}
For any constant $k$, the $k$-hierarchical labeling is an LCL.
\end{corollary}

\begin{lemma}
The $k$-hierarchical labeling problem has worst case complexity $O(n^{\frac{1}{k}})$.
\end{lemma}
\begin{proof}
We will first use \cref{lem:decomposition} with $L=k$ and $\ell = 4$ to obtain a $(O(n^{\frac{1}{k}}), 4, k)$-decomposition in $O(n^{\frac{1}{k}})$ rounds. We will now describe how every node $v$ can compute its output label for the $k$-hierarchical labeling problems using only the information from its direct neighbors. Every node $v$ is either assigned to a Rake Layer, or to a Compress Layer. We will handle both cases separately.
\begin{itemize}
    \item Rake Layer: Let $v$ be in Rake Layer $V_{i,j}^R$ with $1 \leq i \leq k$ and $j \in O(n^{\frac{1}{k}})$. We now have $v$ choose output label $R_i$. By Property \ref{prop:isolated} of the decomposition there exists at most one neighbor with a higher layer than $v$. If such a neighbor $u$ exists, then orient the edge between $v$ and $u$ from $v$ to $u$.
    \item Compress Layer: Let $v$ be in Compress Layer $V_i^C$ with $1 \leq i \leq k-1$. 
    then by Property \ref{prop:compress} of the decomposition $v$ is in path $P$ of length in $[4,8]$, in which all of the nodes have the same layer. If $v$ is not an endpoint of the path, then $v$ will output $C_i$. If instead $v$ is an endpoint of the path $v$ will output $R_{i+1}$. Since $v$ is an endpoint it has exactly one neighbor $u$ in $P$, we will orient the edge from $u$ to $v$. Furthermore by Property \ref{prop:compress} of the decomposition $v$ has exactly one neighbor $w$ in a higher layer. We will orient the edge from $v$ to $w$.
\end{itemize}
Now we argue that all rules of the $k$-hierarchical labeling problem are satisfied.
\begin{enumerate}
    \item Consider a node $v$ that did output a rake label. Consider any neighbor $u$ of $v$, if in the initial decomposition $u$ was assigned a lower layer than $v$, then $u$ did orient the edge towards $v$. If $v$ had a higher layer, then $v$ oriented the edge towards $u$. Lastly the only possible way for $u$ and $v$ to be in the same layer in the decomposition is if both were in the same compress layer. But then $v$ must be an endpoint of that layer, since it did output a rake label. Furthermore since compress layers have at least $4$ nodes, $u$ cannot be an endpoint of the path and so the edge was oriented from $u$ to $v$. Also no inner node of the compress path has to orient both edges.
    \item Whenever we orient edges above, we only orient edges away from the current node and if we do that we only do it once per node.
    \item Again whenever we orient an edge $(u,v)$, the decomposition layer of $u$ is smaller than $v$. Also notice that when we assign output labels the ordering on the layers of the decomposition will directly translate to the ordering we need for our problem.
    \item The compress labels will form paths of length $[2,6]$ since both endpoints adopt a rake label.
    \item Different compress layers are disjoint in the decomposition and the endpoints of the original compress layer use a rake label. As a consequence, any two nodes that use different compress labels will be separated by these rake label nodes.
    \item The only case when compress nodes orient towards a rake label is at the endpoints. Since compress layers are disjoint paths, no node can be the endpoint of more than one compress path. Furthermore if we consider such an endpoint, then by Property \ref{prop:compress} of the decomposition none of its neighbors are in a higher layer. So the nodes that might output an equal label are only the ones inside its own compress layer, but these all output a smaller compress label.
\end{enumerate}
\end{proof}

The definition of a $(\gamma, \ell, L)$-decomposition immediately gives that the rake layers have diameter at most $\gamma$. Now because only nodes that are in rake layer $V_i^R$ are assigned rake label $R_i$, the nodes that output rake labels also form components of diameter at most $\gamma$.
\begin{corollary}
Using the construction above, rake labels induce components ofsize at most $O(n^\frac{1}{k})$.
\end{corollary}

Now we also add the weight functionality into the ruleset.

\begin{definition}[$k$-hierarchical weight augmented $2\frac{1}{2}$-coloring]
For any $k$, the $k$-hierarchical weight augmented $2\frac{1}{2}$-coloring LCL has input label set $\Sigma_{\mathrm{in}} = \{\wei,\act\}$. We call nodes with input $\act$ active nodes and nodes with input $\wei$ weight nodes. Each active node has to output a label from $\Sigma_{\text{out}}^{\act}$, where $\Sigma_{\text{out}}^{\act}$ is the output label set of $k$-hierarchical $2\frac{1}{2}$-coloring. Each weight node has to output a label from the set $\{R_1,\dots, R_k, C_1, \dots, C_{k-1}\}$. \\
Further all weight nodes must output a secondary output from $\Sigma_{\text{out}}^{\act} \cup \{\dec\}$.
The following rules must be satisfied.
\begin{enumerate}
    \item \label{rule1:activeNodes} The active nodes have to compute a valid solution to $k$-hierarchical $2\frac{1}{2}$-coloring on the subgraph induced by active nodes and output only labels from $\Sigma_{\text{out}}^{\act}$.
    \item \label{rule2:copyActive} The weight nodes have to output a valid solution to the $k$-hierarchical labeling problem on the subgraph induced by weight nodes.
    \item Any weight node $w$ adjacent to at least one active node $v$ has to orient the edge $(w,v)$ towards exactly one of these active nodes $v$. Furthermore $w$ must have as secondary output exactly the output of $v$. ($v$ must output from $\Sigma_{\text{out}}^{\act}$ according to \cref{rule1:activeNodes}) 
    \item \label{rule2:copyFromRake} All weight nodes pointing towards other weight nodes must output the same secondary output, unless they also point towards an active node and copy that output.
    \item Only compress label nodes may output $\dec$ as their secondary output and they do so if and only if they are not also adjacent to an active node. If they are adjacent to an active node, then they must have that nodes output as their secondary output.
\end{enumerate}
\end{definition}

Now we prove the effectiveness of this construction. 

\begin{lemma}
For any integer $k \in O(1)$, consider any $k$-hierarchical weight augmented LCL $\Pi'$. Consider an active node $v$ with a balanced $\Delta$-regular tree of $w$ weight nodes attached to it. Let $v$ output $x \in \Sigma_{out}^\Pi$, then $\Omega(w)$ of these weight nodes must have $x$ as their secondary output.
\end{lemma}
\begin{proof}
Let $r$ be the weight node that is attached to $v$ and consider it as the root of the attached tree. Now because of \cref{rule2:copyActive} of the $k$-hierarchical weight augmented definition, $r$ must have $x$ as secondary output and point towards $v$. Since the tree is $\Delta$-regular, the fan-out of the tree is $\Delta-1$ and so as a result a subtree of height $h$ has exactly $(\Delta -1)^h - 1$ many nodes. 
Now $r$ can have either a rake label or a compress label.
\begin{itemize}
    \item If $r$ has a rake label, according to \cref{rule2:oneOutgoing} of the $k$-hierarchical labeling problem, $r$ is pointing to at most one other node. So it has at least $\Delta -2$ incoming edges.
    \item If instead $r$ has a compress label, then according to the rules of the $k$-hierarchical labeling problem $r$ has at most 2 compress neighbors and if it has only one compress neighbor, then it has just one outgoing edge. As a result, $r$ has at least $\Delta -3$ incoming edges.
\end{itemize}
We call nodes adjacent to $r$ that have an edge oriented towards $r$, children of $r$.
Consider one such child $u$.
\begin{enumerate}
    \item $u$ has a rake label: Then by \cref{rule2:copyFromRake} of the $k$-hierarchical weight augmented problem, $u$ must also have $x$ as its secondary output.
    \item $u$ has a compress label: Then $u$ must output $\dec$ as its secondary output. So all of the nodes in the subtree of $u$ must also have $\dec$ as their secondary output. But because the tree is balanced, this is at most an $\frac{1}{\Delta -1}$ fraction of nodes. Furthermore if this is the case, then by \cref{rule2:oneCompressNeighbor} all of the other children of $r$ must output a rake label that is at least one smaller than that of $r$.
\end{enumerate}
If $r$ is a rake node, then we have that $\Delta -2$ children have to also output $x$, if instead $r$ is a compress node at least $\Delta-3$ children have to also output $x$. For each one of these children that have to copy, they must be rake nodes, and therefore have $\Delta -1$ children. For each of them we make the same argument that all rake children have to copy $x$. It is only in the second case, where one such child is a compress child, that there are actually nodes that do not have to have $x$ as their secondary output. However since then all other children must have a strictly smaller rake label and since there are only $k$ different rake labels, this can happen only $k-1$ times. In the worst case this happens in the first $k$ levels of the tree, so we lose $(\Delta -1)^{h-i} - 1$ many nodes in the $i$-th level.
Therefore, the number of weight nodes that must have $x$ as their secondary output is upper bounded by
\[
w - (\Delta -1)^{h-1} -1 - (\Delta -1)^{h-2} -1 - \dots - (\Delta -1)^{h-(k-1)} -1 = w - (k-1) - \sum_{1 \leq i < k} \frac{w}{(\Delta -1)^i}
\]
Which is in $\Omega(w)$ when $\Delta$ and $k$ are constant.
\end{proof}

This basically means that we can attach trees with efficiency factor $x=1$. So by following the same lower bound argumentation as in \cref{sec:weightedlower}, we get that the node-averaged complexity of $k$-hierarchical weight augmented $2\frac{1}{2}$-coloring is in $\Theta(n^{1/k})$. (The upper bound follows from the fact that both $k$-hierarchical $2\frac{1}{2}$-coloring is worst case $O(n^{1/k})$ and also $k$-hierarchical labeling problem is in $O(n^{1/k})$.
\begin{lemma}\label{lem:TightNodeAveraged}
The node-averaged complexity of $k$-hierarchical weight augmented $2\frac{1}{2}$-coloring is in $\Theta(n^{1/k})$.
\end{lemma}

\section{The \texorpdfstring{$\omega(1)$ -- $(\log^* n)^{o(1)}$}{constant to o(ploy(log*n))} Gap}\label{sec:gap}
In this section, we prove that there is no LCL with deterministic node-averaged complexity that lies between $\omega(1)$ and $(\log^* n)^{o(1)}$. Before diving into this proof, we summarize some notions and high-level ideas from previous work, that will be useful in our proofs.

It is known that there are no LCLs with a worst-case complexity (neither deterministic nor randomized) that lies between $\omega(\log n)$ and $n^{o(1)}$ \cite{CP19timeHierarchy}, and that there are no LCLs with a worst-case complexity (neither deterministic nor randomized) that lies between $\omega(n^{1/(k+1)})$ and $o(n^{1/k})$ for any integer $k > 0$ \cite{chang20}. These results have been shown as follows.
\begin{itemize}
    \item It is shown that there exists a generic algorithm for solving all problems $\Pi$ that have $O(\log n)$ worst-case complexity. This generic algorithm requires the existence of a suitable function called $f_{\Pi,\infty}$.
    \item It is shown that an algorithm with $n^{o(1)}$ worst-case (deterministic or randomized) complexity implies the existence of a function $f_{\Pi,\infty}$.
    \item It is shown that there exists a generic algorithm for solving all problems $\Pi$ that have $O(n^{1/(k+1)})$ worst-case complexity. This generic algorithm requires the existence of a suitable function called $f_{\Pi,k+1}$.
    \item It is shown that an algorithm with $o(n^{1/k})$ worst-case (deterministic or randomized) complexity implies the existence of a function $f_{\Pi,k+1}$.
\end{itemize}
Moreover, it is shown that the existence of a function $f_{\Pi,\infty}$ or a function $f_{\Pi,k}$, for a given integer $k>0$, is decidable (i.e., there exists a centralized algorithm that takes as input a problem $\Pi$ and always terminates, outputting the function if it exists, or an error message otherwise).

In \cite{fullversion} it is shown that there is no LCL with node-averaged complexity that lies between $\omega(\log^* n)$ and $n^{o(1)}$, and that if a problem cannot be solved in $n^{o(1)}$ worst-case complexity, then it cannot be solved in $n^{o(1)}$ node-averaged complexity. These results are proved by showing the following.
\begin{itemize}
    \item If a suitable function $f_{\Pi,\infty}$ exists, then $\Pi$ can be solved in $O(\log^* n)$ node-averaged complexity.
    \item If a problem has (deterministic or randomized) worst-case complexity $\Omega(n^{1/k})$, then it has randomized node-averaged complexity $\Omega(n^{1/(2^k - 1)} / \log n)$ and  deterministic node-averaged complexity $\Omega(n^{1/(2^k - 1)})$. This is shown by using a randomized algorithm with node-averaged complexity $o(n^{1/(2^k - 1)} / \log n)$, or a deterministic algorithm with node-averaged complexity $o(n^{1/(2^k - 1)})$, to construct a function $f_{\Pi,k+1}$. 
\end{itemize}
In this section, we prove the following.
\begin{itemize}
    \item Recall that the existence of a suitable function $f_{\Pi,\infty}$ implies that the problem $\Pi$ has node-averaged complexity $O(\log^* n)$, and that it is possible to use an algorithm with $O(\log^* n)$ node-averaged complexity to construct a suitable function $f_{\Pi,\infty}$. We show that, if the function $f_{\Pi,\infty}$ satisfies some additional properties, then $\Pi$ can be solved in $O(1)$ rounds. Moreover, we show that it is decidable whether such function exists.
    \item If there exists an algorithm with deterministic node-averaged complexity $(\log^* n)^{o(1)}$, then a function with such properties exists.
\end{itemize}
Hence, by showing the above, we obtain that the existence of an algorithm with deterministic node-averaged complexity $(\log^* n)^{o(1)}$ implies the existence of an algorithm with deterministic node-averaged complexity $O(1)$, and hence that there are no LCLs with a deterministic node-averaged complexity that lies in the range $\omega(1)$ -- $(\log^* n)^{o(1)}$. Moreover, we also obtain that, whether a problem can be solved in $O(1)$ node-averaged rounds, is decidable. Hence, we obtain the following theorem.
\lowgap*

In the reminder of the section, we start by giving a recap of some notions  presented in \cite{fullversion}, and then we prove our statements. More precisely, the section is structured as follows.
\begin{itemize}
    \item In \cite{fullversion}, it is first shown that all LCLs on trees can be converted into problems described in a specific formalism, called the \emph{black-white formalism}, while preserving the node-averaged complexity of the problem under consideration. In \Cref{ssec:blackWhite} we present this formalism.

    \item In \cite{fullversion} it is shown that, every LCL in the black-white formalism can be solved in a specific generic way, that starts by decomposing the tree in some way. This result is an adaptation, to the black-white formalism, of results already presented in \cite{CP19timeHierarchy,chang20}, We give a high-level overview of this decomposition, and of the generic algorithm, in \Cref{ssec:treeDec}, and \Cref{ssec:genericAlgo}, respectively.

    \item Fundamental ingredients of the generic algorithm are the concepts of \emph{classes} and \emph{label-sets}, that we present in \Cref{ssec:classes}.

    \item In \Cref{ssec:worstToAvg}, we provide a high-level overview of how the described ingredients are used in \cite{fullversion} to obtain an algorithm with $O(\log^* n)$ node-averaged complexity.

    \item As already mentioned, the generic algorithm relies on the existence of a function with special properties. In \Cref{sec:testingProcedure}, we show how the existence of this function is determined.

    \item Finally, in the reminder of the section, we show that, if this function satisfies some additional properties, then the problem can be solved in $O(1)$ deterministic node-averaged complexity.

\end{itemize}

\subsection{LCLs in the Black-White Formalism}\label{ssec:blackWhite}
\begin{definition}[\cite{fullversion}]
  A problem $\Pi$ described in the black-white formalism is a tuple $(\Sigma_{\mathrm{in}},\Sigma_{\mathrm{out}},C_W,C_B)$, where:
\begin{itemize}
    \item $\Sigma_{\mathrm{in}}$ and $\Sigma_{\mathrm{out}}$ are finite sets of labels.
    \item $C_W$ and $C_B$ are both sets of multisets of pairs, where each pair $(\ell_{\mathrm{in}},\ell_{\mathrm{out}})$ is in $\Sigma_{\mathrm{in}} \times \Sigma_{\mathrm{out}}$. 
\end{itemize}
Solving a problem $\Pi$ on a graph $G$ means that:
\begin{itemize}
    \item $G = (W \cup B,E)$ is a graph that is properly $2$-colored, and in particular each node $v \in W$ is labeled $c(v) = W$, and each node $v \in B$ is labeled $c(v) = B$.
    \item To each edge $e \in E$ is assigned a label $i(e) \in \Sigma_{\mathrm{in}}$.
    \item The task is to assign a label $o(e) \in \Sigma_{\mathrm{out}}$ to each edge $e \in E$ such that, for each node $v \in W$ (resp.\ $v \in B$) it holds that the multiset of incident input-output pairs is in $C_W$ (resp.\ in $C_B$).
\end{itemize}  
\end{definition}
In \cite{fullversion}, it is shown how, by starting from an LCL $\Pi$ on trees, one can define an LCL $\Pi'$ in the black-white formalism that has the same asymptotic node-averaged complexity as $\Pi$. Observe that, when considering a problem in the black-white formalism, it is assumed that the tree is $2$-colored.

\subsection{A Tree Decomposition}\label{ssec:treeDec}
The generic algorithms used in \cite{CP19timeHierarchy,chang20} are based on the idea of decomposing the tree into layers. More in detail, the first step is running a procedure called \emph{rake-and-compress}, that takes as input a parameter $\gamma$ and works as follows. For $i = 1,\ldots$ perform the following.
\begin{itemize}
    \item For $j = 1,\ldots,\gamma$ remove nodes of degree 1. Call the removed nodes at step $(i,j)$ \emph{rake nodes of layer $(i,j)$}, and let $V^R_{i,j}$ be the set of these nodes. Call all the nodes removed at step $i$ \emph{rake nodes of layer $i$} and let $V^R_i$ be the set of these nodes.
    \item Remove nodes of degree 2 and call the removed nodes \emph{compress nodes of layer $i$}, and let $V^C_i$ be the set of these nodes.
    \item Repeat until the tree becomes empty.
\end{itemize}
The choice of $\gamma$ affects the number of obtained layers: if $\gamma = 1$, then it is possible to prove that the obtained layers are at most $O(\log n)$, while if $\gamma = \Theta(n^{1/k})$ then it is possible to prove that rake layers are bounded by $k$ and compress layers are bounded by $k-1$.
This decomposition algorithm can be slightly tweaked to obtain additional properties, summarized in \Cref{def:modi}. The worst-case time required to compute such a decomposition is summarized in \Cref{lem:decomposition}.
\begin{definition}[$(\gamma,\ell,L)$-decomposition \cite{fullversion}]\label{def:modi}
    Let $G=(V(G), E(G))$ be a tree. Let $G[V']$, where $V'\subseteq V(G)$, be the subgraph induced by the nodes in $V'$. Given three integers $\gamma, \ell, L$, a \emph{$(\gamma,\ell,L)$-decomposition} is a partition of the nodes in $V(G)$ into $2L-1$ layers $V_1^R = (V^R_{1,1},\ldots,V^R_{1,\gamma}), \ldots, V_{L}^R  = (V^R_{L,1},\ldots,V^R_{L,\gamma})$, $V_1^C, \ldots, V_{L-1}^C$ such that the following hold.
	\begin{enumerate}
		\item \label{prop:compress} \textbf{Compress layers}. The connected components of each $G[V_i^C]$ are paths of length in $[\ell,2\ell]$, the endpoints have exactly one neighbor in a higher layer, and all other nodes do not have any neighbor in a higher layer.
		\item \label{prop:rake}\textbf{Rake layers}. The diameter of the connected components in $G[V_i^R]$ is $O(\gamma)$, and for each connected component at most one node has a neighbor in a higher layer.
		\item \label{prop:isolated}The connected components of each sublayer $G[V^R_{i,j}]$ consist of isolated nodes. Each node in a sublayer $V^R_{i,j}$ has at most one neighbor in a higher layer or sublayer. 
	\end{enumerate}
\end{definition}
\begin{lemma}[\cite{CP19timeHierarchy, chang20}]\label{lem:decomposition}
	Assume $\ell = O(1)$. Then the following hold.
    \begin{itemize}
        \item For any positive integer $k$ and $\gamma = n^{1/k}(\ell / 2)^{1 - 1/k}$, a $(\gamma,\ell,k)$-decomposition can be computed in $O(k\cdot n^{1/k})$ rounds.
        \item For $\gamma = 1$ and $L = O(\log n)$, a $(\gamma,\ell,L)$-decomposition can be computed in $O(\log n)$ rounds.
    \end{itemize}
\end{lemma}

\subsection{The Generic Algorithm}\label{ssec:genericAlgo}
The rake-and-compress procedure produces a useful layering on the tree, that is used as follows. Consider, for simplicity, the case in which there are no nodes in compress layers: according to the ordering $(i,j) < (i',j')$ if $i < i'$ or $i = i' \land j < j'$, we get that each node has at most one neighbor in a higher rake layer. We call the edge connecting a node $u$ to its unique higher-layer neighbor $v$ the \emph{outgoing edge} of $u$, and all the other edges incident to $u$ are called \emph{incoming edges} of $u$. 

Recall that, solving a problem $\Pi$ in the black-white formalism means assigning one label for each edge such that the constraints of the problem are satisfied. We now process the tree to assign a set of labels to each edge.
We process nodes from lower to higher layers. To each node $u$, we assign a set of labels to its outgoing edge, such that, for any choice in such a set, there exists a choice in the sets assigned to the incoming edges of $u$, such that the constraints of the problem are satisfied on $u$. Note that such a set is always non-empty, as long as the problem is solvable in any tree. In other words, we compute what labels we could put on the outgoing edge of $u$ such that we can pick a valid labeling on the whole subtree rooted at $u$. This set of labels, informally, is called \emph{class} of the subtree rooted at $u$. Observe that, once a set has been assigned to each edge, we can process the nodes in reverse order to assign a label to each edge and make the constraints of the problem satisfied on all nodes.

While we discussed how to handle rake layers, handling compress layers is more complicated.
In fact, if we define incoming and outgoing edges analogously as in the case of rake nodes, we get that compress paths have \emph{two} outgoing edges, one for each endpoint. The issue in having two outgoing edges is that, fixing a label on the edge of one endpoint of a compress path may affect what are the valid labelings on the edge of the other endpoint, and hence there is no straightforward way to assign sets of labels to the two endpoints such that this labeling is completable into a valid solution in the path.

This issue is handled in \cite{CP19timeHierarchy,chang20} by first using a decomposition that gives only paths of constant size, and by then using some function that maps the sets of the incoming edges of a path into two sets for the two outgoing edges of the path. This function must satisfy a special property: we want the whole process to never compute empty sets (since otherwise we cannot then assign labels to the edges), and in \cite{CP19timeHierarchy,chang20} is proved that such a function always exists (conditioned on the problem being solvable in the target runtime).
We now provide a formal definition of class, also extended to the case of compress paths.

\subsection{Classes and Label-Sets}\label{ssec:classes}
While the generic algorithms are provided already in \cite{CP19timeHierarchy,chang20}, in \cite{fullversion} the same algorithms are provided in a more accessible form, thanks to the fact that they are restricted to LCLs in the black-white formalism. We now present some of the ingredients provided in \cite{fullversion}. The first ingredient is the definition of label-sets and classes.
\begin{definition}[\cite{bcmos21,fullversion}]\label{def:classes}
    Assume we are given an LCL $\Pi = (\Sigma_{\mathrm{in}},\Sigma_{\mathrm{out}},C_W,C_B)$ in the black-white formalism.
    Consider a tree $G = (V,E)$, and a connected subtree $H = (V_H, E_H)$ of $G$. Assume that the edges connecting nodes in $V_H$ to nodes in $V \setminus V_H$ are split into two parts, $F_{\mathrm{incoming}}$ and $F_{\mathrm{outgoing}}$, that are called, respectively, the set of incoming and outgoing edges. Assume also that for each edge $e \in F_{\mathrm{incoming}}$ is assigned a set $L_e \subseteq \Sigma_{\mathrm{out}}$. This set is called \emph{the label-set of $e$}. Let $\mathcal{L}_{\mathrm{incoming}} = (L_e)_{e \in F_{\mathrm{incoming}}}$. 
    A \emph{feasible labeling} of $H$ w.r.t.\ $F_{\mathrm{incoming}}$, $F_{\mathrm{outgoing}}$, and $\mathcal{L}_{\mathrm{incoming}}$ is a tuple $(L_{\mathrm{outgoing}},L_{\mathrm{incoming}},L_{\mathrm{H}})$ where:
    	\begin{itemize}
		\item $L_{\mathrm{incoming}}$ is a labeling $(l_e)_{e\in F_{\mathrm{incoming}}}$ of $F_{\mathrm{incoming}}$ satisfying $l_e\in (\mathcal{L}_{\mathrm{incoming}})_e$ for all $e \in F_{\mathrm{incoming}}$,
		\item $L_{\mathrm{outgoing}}$ is a labeling $(l_e)_{e\in F_{\mathrm{outgoing}}}$ of ${F_{\mathrm{outgoing}}}$ satisfying $l_e \in \Sigma_{\mathrm{out}}$ for all $e \in F_{\mathrm{outgoing}}$,
		\item $L_H$ is a labeling $(l_e)_{e\in E_H}$ of $E_H$ satisfying $l_e \in \Sigma_{\mathrm{out}}$ for all $e \in E_H$,
		\item the output labeling of the edges incident to nodes of $H$ given by $L_{\mathrm{outgoing}}, L_{\mathrm{incoming}},$ and $L_H$%, and the provided input labeling for the edges incident to nodes of $H$, 
        is such that all node constraints of each node $v\in V_H$ are satisfied.
	\end{itemize}
 	Also, we define the following:
	\begin{itemize}
		\item a \emph{class}  is a set of feasible labelings,
		\item a \emph{maximal class}  is the unique inclusion maximal class, that is, it is the set of all feasible labelings,
		\item an \emph{independent class} is a class $A$ such that
		for any $\big(L_{\mathrm{outgoing}}, L_{\mathrm{incoming}}, L_H\big)\in A$, and for any $\big(L'_{\mathrm{outgoing}}, L'_{\mathrm{incoming}}, L'_H\big)\in A$ the following holds. Let $L''_{\mathrm{outgoing}}$ be an arbitrary combination of $L_{\mathrm{outgoing}}$ and $L'_{\mathrm{outgoing}}$, that is, $L''_{\mathrm{outgoing}} = (l_e)_{e\in F_{\mathrm{outgoing}}}$ where $l_e \in \{ (L_{\mathrm{outgoing}})_e, (L'_{\mathrm{outgoing}})_e \}$. There must exist some $L''_{\mathrm{incoming}}$ and $L''_H$ satisfying $\big(L''_{\mathrm{outgoing}}, L''_{\mathrm{incoming}}, L''_H\big)\in A$.
	\end{itemize}
\end{definition}
Note that the maximal class with regard to some given $\Pi, H, F_{\mathrm{incoming}}, F_{\mathrm{outgoing}},$ and $\mathcal{L}_{\mathrm{incoming}}$, is unique. In contrast, there may be different ways (or none) to restrict a maximal class to a (nonempty) independent class. 

As discussed, the generic algorithm needs to assign label-sets to edges in two specific types of subgraphs $H$, that are the following.
\begin{itemize}
    \item The graph $H$ consists of a single rake node. By construction, $H$ has only one outgoing edge.
    \item The graph $H$ consists of a short compress path. By construction, $H$ has only two outgoing edges.
\end{itemize}
Hence, \Cref{def:classes} is used only for the two specific types of graphs $H$: either single nodes, or short paths. In each of these cases we need to compute a label-set for each outgoing edge. We now report the way to compute label-sets in these two cases as presented in \cite{fullversion}. The cases are illustrated in \Cref{fig:label-set}, which comes from \cite{fullversion}.

\begin{figure}
	\centering
	\includegraphics[width=0.6\textwidth]{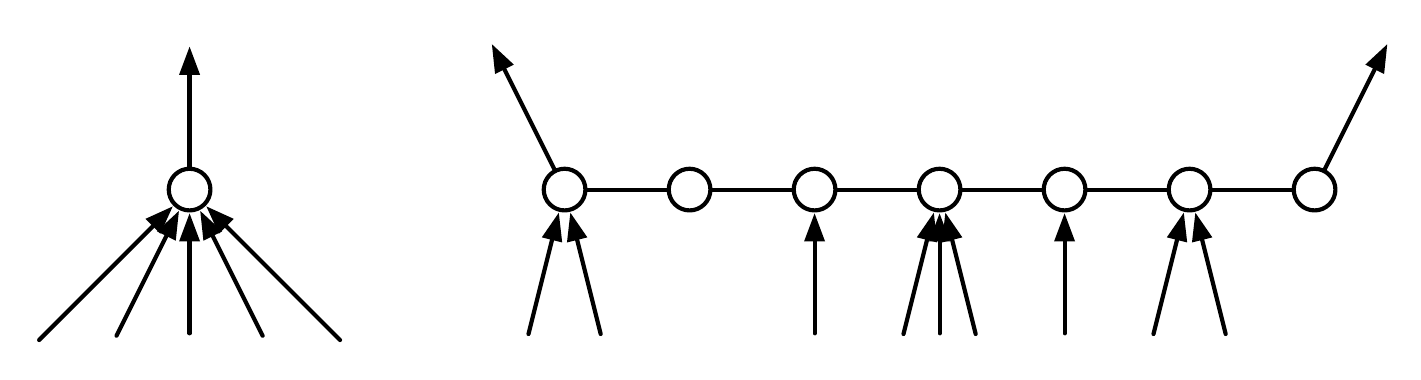}
	\caption{The figure illustrates the two cases of the label-set computation, where it is assumed that the incoming edges have already a label-set assigned and the goal is to assign a label-set to the outgoing edges; the left side depicts the case of a single node, the right side shows the case of a short path.}
	\label{fig:label-set}
\end{figure}

\begin{definition}[label-set computation \cite{fullversion}]\label{def:computing-label-set}
Assume we are given some function $f_{\Pi, k}$ (to be specified later). We define a function $g(v)$ that can be used to compute label-sets for the outgoing edges as a function of $H$, $\Pi$, $F_{\mathrm{incoming}}$, $F_{\mathrm{outgoing}}$, $\mathcal{L}_{\mathrm{incoming}}$, and $f_{\Pi,k}$, for two specific types of graphs $H$.
\begin{itemize}
	\item \textbf{Single nodes:} the graph $H$ consists of a single node $v$ that has a single outgoing edge $e$, and hence $F_{\mathrm{outgoing}} = \{e\}$. All the other edges (which might be $0$) are incoming, and for each of them we are given a label-set (where $\mathcal{L}_{\mathrm{incoming}}$ represents this assignment). We assign, to the outgoing edge, the label-set $g(v)$, that consists of the set of labels that we can assign to the outgoing edge, such that we can pick a label for each incoming edge in a valid manner. More in detail, let $B$ be the maximal class of $H$ w.r.t.\ $\Pi, F_{\mathrm{incoming}}, F_{\mathrm{outgoing}},$ and $\mathcal{L}_{\mathrm{incoming}}$. Then, we denote $g(v)=\bigcup_{(L_{\mathrm{outgoing}}, L_{\mathrm{incoming}}, L_H)\in B}\{(L_{\mathrm{outgoing}})_e\}$. We have $g(v)\subseteq \Sigma_{\mathrm{out}}$. Observe that each node $v$ can compute $g(v)$ if it is given the value of $g(u)$ (that is, the label-set of the edge $\{u,v\}$) for each incoming edge $\{v,u\}$. 
	\item \textbf{Short paths:} the graph $H$ is a path of length between $\ell$ and $2 \ell$, for some $\ell = O(1)$ that depends solely on $\Pi$ and the target running time. The endpoints of the path are $v_1$ and $v_2$. The outgoing edges are $F_{\mathrm{outgoing}} = \{e_1,e_2\}$, where $e_1$ (resp.\ $e_2$) is the outgoing edge incident to $v_1$ (resp.\ $v_2$). Let $B$ be the maximal class of $H$.
	We assume to be given a function $f_{\Pi,k}$, that depends solely on $\Pi$ and some parameter $k$ (that, in turn, depends on the target running time), that maps a class $B$ into an independent class $B' = f_{\Pi,k}(B)$. For $i \in \{1,2\}$, let $g(v_i) = \bigcup_{(L_{\mathrm{outgoing}}, L_{\mathrm{incoming}}, L_H)\in B'}\{(L_{\mathrm{outgoing}})_{e_i}\}$. We have $g(v_i)\subseteq \Sigma_{\mathrm{out}}$.  The label-set of $e_1$ (resp.\ $e_2$) is $g(v_1)$ (resp.\ $g(v_2)$). Observe that the values of $g(v_i)$, for $i\in\{1, 2\}$, can be computed given $H$ and $\mathcal{L}_{\mathrm{incoming}}$.
\end{itemize}
\end{definition}
Observe that, by the definition of independent class, we get that, for any choice of labels in the two label-sets assigned to the outgoing edges of a compress path, there exists a valid labeling for the compress path that is compatible with the label-sets of the incoming edges.

\subsection{From Worst-Case to Node-Averaged Case}\label{ssec:worstToAvg}
\Cref{def:computing-label-set} provides a way to process the layers of a rake-and-compress decomposition to assign label-sets to all edges, such that then nodes can be processed in reverse order to assign a label to each edge. In fact, consider the following ordering.
\begin{definition}[Layer ordering \cite{fullversion}]\label{def:ordering}
    We define the following total order on the (sub)layers of a $(\gamma, \ell, L)$-decomposition.
    \begin{itemize}
        \item  $V_{i,j}^R < V_{i',j'}^R$ iff $i < i' \lor (i = i' \land j < j')$
        \item $V_{i,j}^R < V_{i}^C$
        \item $V_{i}^C < V_{i+1,j}^R$
    \end{itemize}
\end{definition}
We can first process the nodes according to the ordering of \Cref{def:ordering} to assign label-sets to all edges by using \Cref{def:computing-label-set}. Then, we can process the nodes in reverse order and assign a label to each edge such that the constraints of the problem are satisfied. For more details, see \cite[Section 4]{fullversion}. 

In \cite{CP19timeHierarchy,chang20} it is shown that the parameter $\ell$ required by the decomposition can be computed solely as a function of $\Pi$, and that $\ell$ is $O(1)$. Given a function $f_{\Pi,\infty}$, by using a $(1,O(1),O(\log n))$-decomposition, we obtain an algorithm that solves $\Pi$ in $O(\log n)$ worst-case deterministic rounds, while by using a $(O(n^{1/k}),O(1),k)$-decomposition, for $k=O(1)$, we obtain an algorithm that solves $\Pi$ in $O(n^{1/k})$ worst-case deterministic rounds. We call these generic algorithms \emph{solvers}.

In order to obtain $O(\log^* n)$ node-averaged complexity, in \cite{fullversion} it is shown how to compute a tree decomposition with $O(\log^* n)$ node-averaged complexity, and how to additionally tweak it so that the label-set computation, and the following label picking phase, can be computed in $O(1)$ node-averaged complexity. For more details, see \cite[Section 5]{fullversion}. The $O(\log^* n)$ time is actually only spent for splitting compress paths into shorter paths (of length between $\ell$ and $2\ell$), and the whole algorithm of \cite{fullversion} would actually require $O(1)$ rounds if splitting long paths into short ones is not needed. What we show in the rest of the section is that, if there exists an algorithm with $(\log^* n)^{o(1)}$ deterministic node-averaged complexity, then this splitting is not needed, implying an algorithm with $O(1)$ deterministic node-averaged complexity.

\subsection{Finding a Function}\label{sec:testingProcedure}
Recall that, in the described solver procedure, it is required to use a function $f_{\Pi,k}$ for solving a problem in $O(n^{1/k})$ worst-case deterministic rounds, or a function $f_{\Pi,\infty}$ for solving a problem in $O(\log n)$ worst-case deterministic rounds. Moreover, recall that the function used in the solver needs to satisfy the condition that the solver never creates empty label-sets. We call \emph{good} a function that satisfies this condition.

In \cite{CP19timeHierarchy,chang20}, it is argued that there is a \emph{finite} amount of possible (good or bad) functions $f_{\Pi,\infty}$, and for a given $k$, there is a \emph{finite} amount of possible functions $f_{\Pi,k}$. Moreover, it is shown that it is \emph{decidable} whether a given function is good, implying that it is decidable whether a good function exists, and if it exists it is possible to compute it. 

We now present the algorithm, shown in \cite{fullversion}, that tests whether a given function is good. This function is called \emph{testing procedure}.
The procedure depends not only on the function to be tested, but also on a parameter $\ell$ that, in \cite{CP19timeHierarchy,chang20}, it is shown that it can be determined solely as a function of $\Pi$. The testing procedure is well-defined for any integer $k > 0$, but also for $k = \infty$. The idea of the testing procedure is to keep track of all possible label-sets that one could possibly obtain while running the solver.
For each of these label-sets $L$, the procedure also keeps track of a subtree (where nodes are also marked with the layers of a decomposition) where, if we run the solver by using the function that we are testing, we would obtain an edge with label-set $L$. These trees, called \emph{representative trees}, will be used later. 

\vspace{0.5cm}
\hrule
\vspace{0.3cm}
\captionof{algorithm}{The testing procedure of \cite{fullversion}}
\hrule
\begin{enumerate}
    \item Initialize $S$ with all the possible values of the label-set $g(v)$ of $v$ (as defined in \Cref{def:computing-label-set}) that could be obtained when $v$ is a leaf. Note that the possible values are a finite amount that only depends on the amount of input labels of $\Pi$. Initialize $\mathcal{R}_1$ by inserting one pair $((\tilde{T},u),L)$ for each element $L$ in  $S$, where $\tilde{T}$ is a tree composed of $2$ nodes $\{u,v\}$ and $1$ edge $\{u,v\}$, and $L=g(v)$. Node $v$ is marked as a rake node of layer $1$, while $u$ is marked as a temporary node.\label{item:leaves}
    
    \item For $i = 1, \ldots, k$ do the following. If, at any step, an empty label-set is obtained, then the tested function is not good.\label{item:tryall}
    
    \begin{enumerate}
        \item Do the following in all possible ways. Consider $x$ arbitrary elements $((\tilde{T}_j,v_j),L_j)$ of $\mathcal{R}_i$, where $1\le j\le x$ and $1 \le x \le \Delta$. Construct the tree $T$ as the union of all trees $\tilde{T}_j$, where all the nodes $v_j$ (note that each node $v_j$ has degree $1$) are identified as a single node, call it $v$, which, after this process, has degree $x$ in $T$. Let $F_{\mathrm{incoming}}$ be the set of edges connected to $v$, and let $\mathcal{L}_{\mathrm{incoming}}$ be the label-set assignment given by the sets $L_j$. The node $v$ is marked as a rake node of layer $i$.
        If $v$ has an empty maximal class w.r.t. $F_{\mathrm{incoming}}$, $F_{\mathrm{outgoing}} = \{\}$, and $\mathcal{L}_{\mathrm{incoming}}$, then the tested function is not good.\label{item:top-level-can-complete}
        
        \item Do the following in all possible ways. Consider $x$ arbitrary elements $((\tilde{T}_j,v_j),L_j)$ of $\mathcal{R}_i$, where $1\le j\le x$ and $1 \le x \le \Delta -1$. Construct the tree $T$ as the union of all trees $\tilde{T}_j$, where all the nodes $v_j$ are identified as a single node, call it $v$. Attach an additional neighbor $u$ to $v$. Let $F_{\mathrm{outgoing}} = \{\{u,v\}\}$. Let $F_{\mathrm{incoming}}$ be the set of edges connected to $v$, excluding $\{u,v\}$, and let $\mathcal{L}_{\mathrm{incoming}}$ be the label-set assignment given by the sets $L_j$. The node $v$ is marked as a rake node of layer $i$, while $u$ is marked as a temporary node.
        Let $L = g(v)$ (as defined in \Cref{def:computing-label-set}). If $L$ is empty, then the function is not good. Add $((T, u), L)$ to $\mathcal{R}_i$ if no pair with second element $L$ is already present.\label{item:new-rake-labelsets}
        
        \item Repeat the previous two step until nothing new is added to $\mathcal{R}_i$. This must happen, since there are a finite amount of possible label-sets. \label{constructfunction-b}
        
        \item If $i = k$, stop.
        
        \item Initialize $\mathcal{C}_i = \emptyset$.
        
        \item Do the following in all possible ways. Construct a graph starting from a path $H$ of length between $\ell$ and $2 \ell$ where we connect nodes of degree $1$ to the nodes of $H$ satisfying: (i) all nodes in $H$ have degree at most $\Delta$; (ii) the two endpoints of $H$ have an outgoing edge that connects respectively to nodes $u_1$ and $u_2$ that are nodes of degree $1$; (iii) all the other edges connecting degree-$1$ nodes to the nodes of $H$ are incoming for $H$.
        Next, replace each incoming edge $e$ and the node of degree $1$ connected to it with a tree $\tilde{T}$ of a pair $((\tilde{T},u),L)$ in $\mathcal{R}_i$, by identifying $u$ with the node of the path connected to $e$. Different trees can be used for different edges. The nodes $u_1$ and $u_2$ are marked as temporary nodes, while the nodes of the path are marked as compress nodes of layer $i$.
        Use the function as described in \Cref{def:computing-label-set} to compute the label-sets $L_1$ and $L_2$ of the two endpoints. If $L_1$ or $L_2$ is empty, then the function is not good.
        Otherwise, add the pair $((H,u_1),L_1)$ (resp.\ $((H,u_2),L_2)$) to $\mathcal{C}_i$ if no pair with second element $L_1$ (resp.\ $L_2$) is already present. The \emph{representative tree} of $P = (H,F_{\mathrm{incoming}},F_{\mathrm{outgoing}},\mathcal{L}_{\mathrm{incoming}})$ is defined as $r(P) = T$.\label{item:compress-paths}
        
        \item Set $\mathcal{R}_{i+1} = \mathcal{R}_i \cup \mathcal{C}_i$. If $\mathcal{R}_{i+1} = \mathcal{R}_{i}$, stop.
    \end{enumerate}
\end{enumerate}
\hrule
\vspace{0.5cm}

For more details about this procedure, we refer the reader to \cite[Section 7]{fullversion}.
As discussed in \cite{fullversion}, it is possible to prove that the testing procedure generates exactly those label-sets that could possibly be obtained by running the solver \cite{CP19timeHierarchy, chang20}. Hence, if empty label-sets are never obtained, then the function can indeed be used to solve a problem.
We now prove a useful property of this testing procedure.
\begin{lemma}\label{lem:kmaxfunction}
    Let $\Pi$ be an LCL in the black-white formalism. Let $F$ be the (finite) set of all possible functions $f_{\Pi,\infty}$. 
    Let $k_{\mathrm{max}}$ be the largest value of $i$ reached when testing all the function in $F$. If a good function $f_{\Pi,k_{\mathrm{max}}+1}$ exists, then a good function $f_{\Pi,\infty}$ exists.
\end{lemma}
\begin{proof}
 Define $f_{\Pi,\infty} := f_{\Pi,k_{\mathrm{max}}+1}$.
    The proof follows from the fact that all functions $f_{\Pi,\infty}$ that are not good fail the testing procedure within the first $k_{\mathrm{max}}$ steps, while $f_{\Pi,k_{\mathrm{max}}+1}$ does not fail in the first $k_{\mathrm{max}+1}$ steps.
\end{proof}

\subsection{The Plan for Saving the \texorpdfstring{$O(\log^*n)$}{O(log*n)} Term}
We now discuss how the $O(\log^* n)$ term in the node-averaged complexity of the solver procedure of \cite{fullversion} can be avoided.
The solver procedure, as shown in \cite{fullversion}, precomputes a distance-$O(1)$ $O(1)$-coloring (for some suitable constants), which can be done in $O(\log^* n)$ deterministic worst-case rounds. Then, throughout its execution, the solver handles each compress path $P$ in the compress layer $V^C_i$ as follows. 
\begin{itemize}
    \item First, compute a set $S$ of nodes of the compress path $P$ such that the subpaths induced by the nodes in $P \setminus S$ form paths of length in $[\ell,2\ell]$, and such that endpoints of $P$ are not in $S$. The nodes in $S$ are moved to $V^R_{i+1,1}$. This is done in $O(1)$ deterministic worst-case rounds by exploiting the precomputed coloring.
    \item The function $f_{\Pi,\infty}$ is used to assign label-sets to the edges connecting nodes of the subpaths to their higher-layer neighbors.
    \item Nodes in $S$ now have all incoming edges with label-sets assigned, and no outgoing edges. This means that they can pick a valid labeling for all their incident edges and terminate.
    \item Nodes in the subpaths, except for the first and the last subpath, can now pick a valid labeling for their incident edges and terminate, since such subpaths are connected to nodes in $S$ on both sides, and since the subpaths have constant length.
\end{itemize}
This whole procedure allows to fully label the nodes of a path (except for the first and last subpath) in $O(1)$ worst-case rounds, by exploiting the precomputed coloring. We now provide an alternative way for handling compress paths in $O(1)$ worst-case rounds, that, on the one hand, it does not require to precompute a coloring, but on the other hand, it requires the function $f_{\Pi,\infty}$ to satisfy some additional properties. We thus get that, if there exists a function $f_{\Pi,\infty}$ that satisfies these additional properties, then the algorithm with deterministic node-averaged complexity $O(\log^* n)$ of \cite{fullversion} can be turned into an algorithm with $O(1)$ deterministic node-averaged complexity.
Finally, we will show that if a problem can be solved in $(\log^* n)^{o(1)}$ deterministic node-averaged complexity, then a function with the required properties exists.

\subsection{The New Way of Handling Compress Paths}
We modify the algorithm of \cite{fullversion}, and in particular the handling of compress paths, as follows. In the original algorithm, all compress paths have length in $[\ell,2\ell]$, due to the fact that they have been split by exploiting the precomputed coloring. Since we cannot precompute this coloring anymore, we now need to be able to handle paths of arbitrary length (still at least $\ell$). Observe that if a path $P$ has length at most $3\ell+4$, a coloring in $P$ can be computed in constant time, and hence the original algorithm can be used. Hence, in the following, we assume the path to be of length at least $3\ell+4$. We now show how to handle a path $P$ in compress layer $i$, that is, all nodes of $P$ are in $V^C_i$.

Consider the two nodes at distance exactly $\ell+1$ from an endpoint of $P$.
We split $P$ into three subpaths by promoting these two nodes to $V^R_{i+1,1}$ (that is, the next layer). We call (according to some arbitrary ordering) the first and last promoted node $s_1$ and $s_2$, the first and last subpath $P_1$ and $P_2$, and the middle path $P_m$. Observe that $P_1$ and $P_2$ have length exactly $\ell$, and $P_m$ has length at least $\ell$. We use the function $f_{\Pi,\infty}$ to assign label-sets to the edges connecting the endpoints of $P_1$ and $P_2$ to their higher-layer neighbors (that include $s_1$ and $s_2$).

What remains to be done is assigning labels to the edges of $P_m$, and this must be done in constant time even if $P_m$ is of superconstant length. On a high-level, we have discussed why this is sufficient in order to obtain an algorithm with $O(1)$ deterministic node-averaged complexity, and this is also shown more formally in \cite[Lemma 40, arXiv version]{fullversion}. 

We define a new LCL problem $\Pi'$, as a function of $\Pi=(\Sigma_{\mathrm{in}},\Sigma_{\mathrm{out}},C_W, C_B)$, $f_{\Pi,\infty}$, a set of label-sets $\mathcal{C}$, and the original maximum degree $\Delta$, that captures exactly what we need (that is, to assign labels to the edges of $P_m$).

\begin{definition}[Compress Problem $\Pi'$]\label{def:compressproblem}
We define an LCL problem $\Pi'$ (not in the black-white formalism) that has checkability-radius $O(\ell)$, where inputs are provided to nodes, and outputs are on edges. 
The problem $\Pi'$ is defined on paths, that is, it is assumed that all nodes have degree at most $2$. If the path is shorter than $\ell$ (recall that $\ell=O(1)$), then the problem $\Pi'$ is defined such that any output is allowed. Hence, in the following we assume that the path has length at least $\ell$.
The problem $\Pi' = (\Sigma'_{\mathrm{in}},\Sigma'_{\mathrm{out}},C')$ is defined as follows. Let $\Sigma'_{\mathrm{in}} := \bigcup_{j=0}^{\Delta-1} \mathcal{C}^j$, that is, a node receives as input a tuple of size at most $\Delta-1$, where each element of the tuple is a label-set in $\mathcal{C}$. Let $\Sigma'_{\mathrm{out}} := \Sigma_{\mathrm{out}}$, that is, the possible outputs of $\Pi$ and $\Pi'$ are the same. We now define the constraint $C'$, by consider three possible cases.
    \begin{itemize}
        \item \textbf{Nodes of degree $2$}. The constraint $C'$, on an arbitrary node $v$ of degree $2$, is defined as follows. If the tuple that $v$ received as input is of size $\Delta - 1$, then $v$ is unconstrained (any output is allowed). Let $o_1,o_2$ be the labels outputted by $v$ on its incident edges, and let $(L_1,\ldots,L_k)$ be the input of $v$. It must hold that there exists a choice $\ell_1 \in L_1,\ldots,\ell_k \in L_k$ such that the multiset $\{o_1,o_2\} \cup \{\ell_1,\ldots,\ell_k\}$ is in the constraint of $v$ of the original problem $\Pi$.
        \item \textbf{Nodes of degree $1$}. For a node $v$ of degree $1$, the constraint $C'$ is defined as follows. Let $(L_1,\ldots,L_k)$ be the input of $v$. If $k = 0$, then $v$ is unconstrained (that is, any output is allowed). If $L_1$ is not in the codomain of the function $g$ (that is defined as a function of $f_{\Pi,\infty}$ in \Cref{def:computing-label-set}), then $v$ is unconstrained. Otherwise, the constraint of $v$ is defined in the same way as in the case of nodes of degree $2$.
    \end{itemize}
\end{definition}
Before proving that a constant-time algorithm for $\Pi'$ implies that we can assign labels to the edges of $P_m$ in constant time, we report a useful property about LCLs on paths.
\begin{observation}[\cite{balliu19lcl-decidability}]\label{obs:constant-no-n}
    Let $\Pi$ be an LCL defined on paths. Let $\mathcal{A}$ be an algorithm that solves $\Pi$ in $O(1)$ worst-case rounds, when given as input the size $n$ of the path. Then, there exists a value $n_0=O(1)$ such that $\mathcal{A}$ solves $\Pi$ in constant time on instances of size $n\ge n_0$ when given as input $n_0$ instead of $n$.
\end{observation}

We observe the following. 
\begin{observation}\label{obs:constantpath}
    If $\Pi'$ can be solved in worst-case $O(1)$ rounds, then the edges incident to the nodes of $P_m$ can be labeled in worst-case $O(1)$ rounds such that the following hold.
    \begin{itemize}
        \item The constraints of $\Pi$ are satisfied.
        \item On edges that have an assigned label-set, the output is from that set.
    \end{itemize}
\end{observation}
\begin{proof}
    The claim follows directly from the definition of $\Pi'$. In fact, the path $P_m$ is of length at least $\ell$, and the edges connecting the endpoints of $P_m$ to neighbors of higher layers have a label-set that is in the codomain of $g$. Under these conditions, the constraint of $\Pi'$ is defined in such a way that it satisfies the requirements of the observation. Hence, by running an algorithm for $\Pi'$ we obtain the required labeling. However, in the LOCAL model, it typically is assumed that, when running an algorithm, the size $n$ of the graph is provided to the nodes, but in our case nodes do not know the size of $P_m$. Though, by \Cref{obs:constant-no-n}, such a knowledge of $n$ is not required on instances of size at least some constant $n_0$. On instances smaller than $n_0$ we can anyways compute the size of the path in constant time.
\end{proof}

We are now ready to define the additional property that we require on a function.
\begin{definition}[Constant-good function]\label{def:constantgood}
    A function $f_{\Pi,\infty}$ is called constant-good if its associated problem $\Pi'$ has worst-case complexity $O(1)$.
\end{definition}

Observe that, by combining \Cref{obs:constantpath} with a constant-good function, we can handle compress paths in constant time without precomputing a coloring.
We exploit the following lemma.
\begin{lemma}[\cite{balliu19lcl-decidability}]\label{lem:path-decidability}
    Given an LCL problem $\Pi$ defined on paths, it is possible to decide whether it can be solved in $O(1)$ worst-case rounds.
\end{lemma}
It is easy to modify the testing procedure to additionally check if a tested function $f_{\Pi,\infty}$ is constant-good: first, run the normal testing procedure, then, if a function passes the test, additionally check, by using \Cref{lem:path-decidability}, whether the problem $\Pi'$ defined as a function of $f_{\Pi,\infty}$ can be solved in worst-case $O(1)$ rounds.

\subsection{From a \texorpdfstring{$(\log^* n)^{o(1)}$}{o(poly(log*n))} Node-Averaged Algorithm to a Constant-Good Function}
We now prove that, if there exists an algorithm $\mathcal{A}$ solving $\Pi$ with deterministic node-averaged complexity $(\log^* n)^{o(1)}$, then there exists a function $f_{\Pi,\infty}$ that is constant-good.
We start by reporting a useful lemma shown in \cite{fullversion}.

\begin{lemma}[Lemma 60 in the arXiv version of \cite{fullversion}, rephrased]\label{lem:summarydet}
    Let $\mathcal{A}$ be an algorithm with deterministic node-averaged complexity $o(w^{1/(2^k -1)})$, where $w = w(n)$ is in $O(n)$. Let $\{1,\ldots,n^c\}$ be the range of IDs for which the algorithm is defined, and assume that $c$ is a large-enough constant. Moreover, let $n$ be an integer that is at least a large-enough constant. Then, it is possible to construct a good function $f_{\Pi,k+1}$, as a function of $\mathcal{A}$ and $n$, satisfying the following properties.

    Let $((T,v), L)$ be a pair obtained by the testing procedure by using the function $f_{\Pi,k+1}$. Then, there exists a tree $T'$ satisfying the following properties.
    
    \begin{itemize}
        \item $T'$ contains a node $v'$ of degree one.
        \item $T'$ has at most $N$ nodes, where $N = \Theta(w(n))$, and
        \item there exists a set $\mathcal{X}$ containing $O(n^{c-4})$ disjoint sets of IDs
    \end{itemize}
    satisfying that, for each $X \in \mathcal{X}$, it is possible to assign IDs to the nodes of $T'$, such that:
    \begin{itemize}
        \item All the IDs are from $X$.
        \item There exists a set of nodes $D$ in $T'$ that, by running $\mathcal{A}$ on any supertree of $T'$ of size $n$, they do not see outside $T'$.
        \item Any labeling of $T'$ that agrees with the outputs of the nodes of $D$ restricts the labels on the edge incident to $v'$ to a subset of $L$.
    \end{itemize}
\end{lemma}

Since $\mathcal{A}$ runs in $(\log^* n)^{o(1)}\subseteq n^{o(1)}$, we get that a good function $f_{\Pi,k}$ exists for any arbitrary $k$ \cite{fullversion}.
By \Cref{lem:kmaxfunction}, this implies that a good function $f_{\Pi,\infty}$ exists. In the rest of the section, we prove that some constant-good function $f_{\Pi,\infty}$ exists (but not necessarily the one given by \Cref{lem:kmaxfunction}). Observe that we only need to prove that there exists a function $f_{\Pi,\infty}$ for which its associated problem $\Pi'$ has worst-case $O(1)$ complexity.

For this purpose, we first define a different LCL problem $\Pi''$, and, by exploiting \Cref{lem:summarydet}, we prove that $\Pi''$ has $O(1)$ worst-case complexity, by showing that we can use $\mathcal{A}$ to solve it. We will later show that this implies that there exists a function $f_{\Pi,\infty}$ for which its associated problem $\Pi'$ has worst-case $O(1)$ complexity.

\begin{definition}[The Problem $\Pi''$]
    The problem $\Pi''$ is a problem on paths defined as follows. Let $F = \{f_1,\ldots,f_h\}$ be the set of possible good functions for $\Pi$ (recall that $h = O(1)$), and let $\Pi'_i$ be the problem $\Pi'$ defined as a function of $f_i$. Each node receives an input for $\Pi'_i$ for each $i$. Each node has to output a value $i\in \{1,\ldots,h\}$ and an output for $\Pi'_h$. If two nodes are neighbors, they are required to output the same index $i$. The output for $\Pi'_h$ needs to satisfy the constraints of $\Pi'_h$.
\end{definition}
Observe that the definition of $\Pi''$ does not depend on a specific function, but it solely depends on the original problem $\Pi$.
\begin{lemma}
    The problem $\Pi''$ has worst-case complexity $O(1)$.
\end{lemma}
\begin{proof}
We start by showing how to use $\mathcal{A}$ to solve $\Pi''$. For this purpose, we show how, nodes of a path, can simulate $\mathcal{A}$ on a virtual graph constructed as a function of their input.

Let $n$ be the length of the path. If $n$ is smaller than some large enough constant, we solve the problem by brute force. A solution exists because $\Pi''$ is constructed by using good functions.
Let $w(n) = \sqrt{\log^* n}$. Observe that $(\log^* n)^{o(1)}$ is in $o(w(n)^{1/(2^k -1)})$ for any $k$, and hence we can use \Cref{lem:summarydet}. Let $n' = c \cdot n \cdot \sqrt{\log^* n}$, for some large enough constant $c$. 
We apply \Cref{lem:summarydet} to construct a good function $f = f_{\Pi,k_{\mathrm{max}} + 1}$ with parameters $n'$ and $w$. 
Let $i$ be the index of $f$ in the definition of the problem $\Pi''$ (which exists since the definition of $\Pi''$ considers all the good functions).
Let $P = (v_1,\ldots,v_n)$ be a path, where the input of node $v_j$ at index $i$ is $\mathcal{L} = (L^j_1,\ldots,L^j_{d_j})$. We show that we can use $\mathcal{A}$ to solve $\Pi''$ with index $i$, by constructing a virtual graph where we run $\mathcal{A}$. 

The virtual graph is obtained by connecting to each node $v_j$, $d_j$ trees, one for each label-set in $\mathcal{L}$, as given by \Cref{lem:summarydet}, where the node $v'$ of each tree is identified with $v_j$. It is known that the complexity of an LCL problem on paths does not depend on the range of the IDs \cite{balliu19lcl-decidability}, and thus, for our purpose, we assume that the IDs on the path are in the range $\{1,\ldots,n^{c-5}\}$. We get that, for large enough $n$, there exists an injective function from the IDs on the path to $\{1,\ldots,O(n^{c-4})\}$. Hence, by \Cref{lem:summarydet}, we can label the whole virtual graph with unique IDs that satisfy the following. By running $\mathcal{A}$ (and telling it that there are $n'$ nodes, which, by \Cref{lem:summarydet}, is an upper bounds on the nodes of the virtual graph), we obtain that, by \Cref{lem:summarydet}, for each node $v_j$ of the path, the labeling produced by $\mathcal{A}$ satisfies that the virtual $h$th edge of $v_j$ has a label present in $L^j_h$. Moreover, $\mathcal{A}$ also assign a solution to the edges of the path. Thus, we get a solution for $\Pi'_i$, and hence a solution for $\Pi''$.

We now prove that the deterministic node-averaged runtime is bounded by $o(\log^* n)$, and by \cite{Feuilloley17}, $o(\log^* n)$ deterministic node-averaged complexity on paths implies $O(1)$ worst-case complexity.

The sum of the runtimes on the virtual instance is bounded by $S' \le n' \cdot (\log^* n')^{o(1)} = n' \cdot (\log^* n)^{o(1)}$. 
Let $S$ be the sum of the running times of the nodes in the path. We obtain that the node-averaged complexity is bounded by the following.
\[
\frac{1}{n}S \le \frac{1}{n}S' \le \frac{c \cdot \sqrt{\log^* n}}{n'}S' \le \frac{c \cdot \sqrt{\log^* n}}{n'} \cdot n' \cdot (\log^* n)^{o(1)} \le c \cdot (\log^* n)^{\frac{1}{2} + o(1)} \in o(\log^* n)
\]
\end{proof}

Note that the deterministic algorithm that we constructed for solving $\Pi''$ in $O(1)$ may solve, for different values of $n$, different subproblems $\Pi'_i$, which by itself is not sufficient to show that there is a problem $\Pi'_i$ with deterministic worst-case complexity $O(1)$. However, we now show that, the fact that $\Pi''$ has worst-case complexity $O(1)$, implies that $\Pi'_i$ has worst-case complexity $O(1)$ for some $i$, implying that a constant-good function exists.
\begin{lemma}
    There exists an $i$ for which $\Pi'_i$ has worst-case complexity $O(1)$.
\end{lemma}
\begin{proof}
    By \Cref{obs:constant-no-n}, any large-enough instance of $\Pi''$ can be solved by running an algorithm $\mathcal{A}$ for $\Pi''$ by providing it always the same value of $n$, call it $n_0$.

    We start by proving that, given two large enough instances $I_1$ and $I_2$ for $\Pi''$, even of different sizes, the solution produced by $\mathcal{A}$ in $I_1$ and $I_2$ must satisfy that all nodes, in both instances, output a solution that uses the same index $i$.
    For a contradiction, suppose that there are two different large enough instances $I_1$ and $I_2$ where nodes use different indices. Then, we can create a single instance $I_3$ where there is at least one node that has the same view as in $I_1$ within its runtime, and at least one node that has the same view as in $I_2$ within its runtime. Since $\mathcal{A}$ is run by telling it that there are $n_0$ nodes, we get that these two nodes, in $I_3$, produce the same output as in the instances $I_1$ and $I_2$. This contradicts the correctness of the algorithm, since the constraints of $\Pi''$ require that all nodes output the same index $i$.

    Since, for large enough instances, $\mathcal{A}$ always solves $\Pi''$ with the same index $i$, and for all small enough instances we can find a solution by brute force, we get that an algorithm that solves $\Pi'_i$ in $O(1)$ rounds exists.
\end{proof}

\urlstyle{same}
\bibliographystyle{alpha}
\bibliography{biblio}

\newcommand{\etalchar}[1]{$^{#1}$}
\begin{thebibliography}{BBK{\etalchar{+}}23b}

\bibitem[BBC{\etalchar{+}}19]{balliu19lcl-decidability}
Alkida Balliu, Sebastian Brandt, Yi-Jun Chang, Dennis Olivetti, Mika{\"e}l
  Rabie, and Jukka Suomela.
\newblock The distributed complexity of locally checkable problems on paths is
  decidable.
\newblock In {\em Proc.\ 38th ACM Symposium on Principles of Distributed
  Computing (PODC 2019)}, pages 262--271. ACM Press, 2019.

\bibitem[BBC{\etalchar{+}}22]{B0COSS22_LCLregularTrees}
Alkida Balliu, Sebastian Brandt, Yi{-}Jun Chang, Dennis Olivetti, Jan
  Studen{\'{y}}, and Jukka Suomela.
\newblock Efficient classification of locally checkable problems in regular
  trees.
\newblock In {\em Proc.\ 36th International Symposium on Distributed
  Computing,(DISC 2022)}, pages 8:1--8:19, 2022.

\bibitem[BBE{\etalchar{+}}20]{binary_lcls}
Alkida Balliu, Sebastian Brandt, Yuval Efron, Juho Hirvonen, Yannic Maus,
  Dennis Olivetti, and Jukka Suomela.
\newblock Classification of distributed binary labeling problems.
\newblock In {\em Proc.\ 34th International Symposium on Distributed Computing
  (DISC 2020)}, volume 179 of {\em LIPIcs}, pages 17:1--17:17. Schloss
  Dagstuhl--Leibniz-Zentrum f{\"u}r Informatik, 2020.

\bibitem[BBF{\etalchar{+}}22]{B0FLMOU22}
Alkida Balliu, Sebastian Brandt, Manuela Fischer, Rustam Latypov, Yannic Maus,
  Dennis Olivetti, and Jara Uitto.
\newblock Exponential speedup over locality in {MPC} with optimal memory.
\newblock In {\em 36th International Symposium on Distributed Computing, (DISC
  2022)}, pages 9:1--9:21, 2022.

\bibitem[BBK{\etalchar{+}}23a]{fullversion}
Alkida Balliu, Sebastian Brandt, Fabian Kuhn, Dennis Olivetti, and Gustav
  Schmid.
\newblock On the node-averaged complexity of locally checkable problems on
  trees.
\newblock {\em CoRR}, abs/2308.04251, 2023.

\bibitem[BBK{\etalchar{+}}23b]{Balliu0KOS23}
Alkida Balliu, Sebastian Brandt, Fabian Kuhn, Dennis Olivetti, and Gustav
  Schmid.
\newblock On the node-averaged complexity of locally checkable problems on
  trees.
\newblock In Rotem Oshman, editor, {\em 37th International Symposium on
  Distributed Computing, {DISC} 2023, October 10-12, 2023, L'Aquila, Italy},
  volume 281 of {\em LIPIcs}, pages 7:1--7:21. Schloss Dagstuhl -
  Leibniz-Zentrum f{\"{u}}r Informatik, 2023.

\bibitem[BBO{\etalchar{+}}21]{LCLs_in_rooted_trees}
Alkida Balliu, Sebastian Brandt, Dennis Olivetti, Jan Studen{\'{y}}, Jukka
  Suomela, and Aleksandr Tereshchenko.
\newblock Locally checkable problems in rooted trees.
\newblock In {\em Proc.\ 40th ACM Symposium on Principles of Distributed
  Computing (PODC 2021)}, pages 263--272, 2021.

\bibitem[BBOS20]{BBOS20paddedLCL}
Alkida Balliu, Sebastian Brandt, Dennis Olivetti, and Jukka Suomela.
\newblock How much does randomness help with locally checkable problems?
\newblock In {\em Proc.\ 39th ACM Symposium on Principles of Distributed
  Computing (PODC 2020)}, pages 299--308. ACM Press, 2020.

\bibitem[BBOS21]{BBOS18almostGlobal}
Alkida Balliu, Sebastian Brandt, Dennis Olivetti, and Jukka Suomela.
\newblock Almost global problems in the {LOCAL} model.
\newblock {\em Distributed Comput.}, 34(4):259--281, 2021.

\bibitem[BCM{\etalchar{+}}21]{bcmos21}
Alkida Balliu, Keren Censor{-}Hillel, Yannic Maus, Dennis Olivetti, and Jukka
  Suomela.
\newblock Locally checkable labelings with small messages.
\newblock In {\em 35th International Symposium on Distributed Computing, {DISC}
  2021}, pages 8:1--8:18, 2021.

\bibitem[BGKO23]{BalliuGKO22_average}
Alkida Balliu, Mohsen Ghaffari, Fabian Kuhn, and Dennis Olivetti.
\newblock Node and edge averaged complexities of local graph problems.
\newblock {\em Distributed Comput.}, 36(4):451--473, 2023.

\bibitem[BHK{\etalchar{+}}17]{lcls_on_grids}
Sebastian Brandt, Juho Hirvonen, Janne~H. Korhonen, Tuomo Lempi{\"{a}}inen,
  Patric R.~J. {\"{O}}sterg{\aa}rd, Christopher Purcell, Joel Rybicki, Jukka
  Suomela, and Przemyslaw Uznanski.
\newblock {LCL} problems on grids.
\newblock In {\em Proc.\ 36th ACM Symposium on Principles of Distributed
  Computing (PODC 2017)}, pages 101--110, 2017.

\bibitem[BHK{\etalchar{+}}18]{BHKLOS18lclComplexity}
Alkida Balliu, Juho Hirvonen, Janne~H. Korhonen, Tuomo Lempi{\"a}inen, Dennis
  Olivetti, and Jukka Suomela.
\newblock New classes of distributed time complexity.
\newblock In {\em Proc.\ 50th ACM Symposium on Theory of Computing (STOC
  2018)}, pages 1307--1318. ACM Press, 2018.

\bibitem[BHOS19]{BHOS19HomogeneousLCL}
Alkida Balliu, Juho Hirvonen, Dennis Olivetti, and Jukka Suomela.
\newblock Hardness of minimal symmetry breaking in distributed computing.
\newblock In {\em Proc.\ 38th ACM Symposium on Principles of Distributed
  Computing (PODC 2019)}, pages 369--378. ACM Press, 2019.

\bibitem[BT19]{BarenboimT19}
Leonid Barenboim and Yaniv Tzur.
\newblock Distributed symmetry-breaking with improved vertex-averaged
  complexity.
\newblock In {\em Proc.\ 20th Int.\ Conf.\ on Distributed Computing and
  Networking (ICDCN)}, pages 31--40, 2019.

\bibitem[Cha20]{chang20}
Yi-Jun Chang.
\newblock The complexity landscape of distributed locally checkable problems on
  trees.
\newblock In {\em Proc.\ 34th International Symposium on Distributed Computing
  (DISC 2020)}, volume 179 of {\em LIPIcs}, pages 18:1--18:17. Schloss
  Dagstuhl--Leibniz-Zentrum f{\"u}r Informatik, 2020.

\bibitem[CKP19]{CKP19exponential}
Yi{-}Jun Chang, Tsvi Kopelowitz, and Seth Pettie.
\newblock An exponential separation between randomized and deterministic
  complexity in the {LOCAL} model.
\newblock {\em {SIAM} J. Comput.}, 48(1):122--143, 2019.

\bibitem[CP19]{CP19timeHierarchy}
Yi{-}Jun Chang and Seth Pettie.
\newblock A time hierarchy theorem for the {LOCAL} model.
\newblock {\em SIAM J. Comput.}, 48(1):33--69, 2019.

\bibitem[CSS21]{lcls_on_paths_and_cycles}
Yi-Jun Chang, Jan Studen{\'y}, and Jukka Suomela.
\newblock Distributed graph problems through an automata-theoretic lens.
\newblock In {\em Proc.\ 28th International Colloquium on Structural
  Information and Communication Complexity (SIROCCO 2021)}, LNCS. Springer,
  2021.

\bibitem[Feu17]{Feuilloley17}
Laurent Feuilloley.
\newblock How long it takes for an ordinary node with an ordinary {ID} to
  output?
\newblock In {\em Proc.\ 24th Int.\ Coll.\ on Structural Information and
  Communication Complexity (SIROCCO)}, volume 10641, pages 263--282, 2017.

\bibitem[GK21]{GhaffariK21}
Mohsen Ghaffari and Fabian Kuhn.
\newblock Deterministic distributed vertex coloring: Simpler, faster, and
  without network decomposition.
\newblock In {\em 62nd {IEEE} Annual Symposium on Foundations of Computer
  Science, {FOCS} 2021, Denver, CO, USA, February 7-10, 2022}, pages
  1009--1020. {IEEE}, 2021.

\bibitem[GRB22]{brandt21trees}
Christoph Grunau, V{\'{a}}clav Rozhon, and Sebastian Brandt.
\newblock The landscape of distributed complexities on trees and beyond.
\newblock In {\em Proc.\ 41st ACM Symposium on Principles of Distributed
  Computing (PODC 2022)}, pages 37--47, 2022.

\bibitem[Lin92]{Linial92}
Nathan Linial.
\newblock Locality in distributed graph algorithms.
\newblock {\em {SIAM} J. Comput.}, 21(1):193--201, 1992.

\bibitem[NS95]{NaorStockmeyer95}
Moni Naor and Larry~J. Stockmeyer.
\newblock What can be computed locally?
\newblock {\em {SIAM} J. Comput.}, 24(6):1259--1277, 1995.

\end{thebibliography}

\end{document}